\documentclass[authoryear,12pt]{elsarticle}

\usepackage{natbib}
\usepackage{epstopdf}
\usepackage{amsmath}
\usepackage{amsthm}
\usepackage{amsfonts}
\usepackage{bbm}
\usepackage{booktabs}
\usepackage{multirow}
\usepackage{rotating}
\usepackage{lscape}
\usepackage{graphicx}
\usepackage{bigints}
\usepackage[table]{xcolor}

\biboptions{round}

\newtheorem{lemma}{Lemma}[section]

\def\0{\mbox{\bf{0}}}

\newcommand{\suchthat}{\;\ifnum\currentgrouptype=16 \middle\fi|\;}

\newtheorem{theorem}{Theorem}[section]
{
\theoremstyle{plain}

  }

\DeclareMathOperator*{\argmax}{argmax}
\newtheorem{prop}{Proposition}

\newcommand\tab[1][1cm]{\hspace*{#1}}

\begin{document}

\begin{frontmatter}



\title{\LARGE Time-Varying Poisson Autoregression}


\author[bo]{Giovanni Angelini}
\author[bo,ex]{Giuseppe Cavaliere}
\author[vu]{Enzo D'Innocenzo}
\author[bo]{Luca De Angelis}

\address[bo]{Department of Economics, University of Bologna}
\address[vu]{School of Business and Economics, Vrije Universiteit Amsterdam}
\address[ex]{Essex Business School, University of Essex}

\begin{abstract}
In this paper we propose a new time-varying econometric model, called Time-Varying Poisson AutoRegressive with eXogenous covariates (TV-PARX), suited to model and forecast time series of counts. {We show that the score-driven framework is particularly suitable to recover the evolution of time-varying parameters and provides the required flexibility to model and forecast time series of counts characterized by convoluted nonlinear dynamics and structural breaks.} We study the asymptotic properties of the TV-PARX model and prove that, under mild conditions, maximum likelihood estimation (MLE) yields strongly consistent and asymptotically normal parameter estimates. Finite-sample performance and forecasting accuracy are evaluated through Monte Carlo simulations.
The empirical usefulness of the time-varying specification of the proposed TV-PARX model is shown by analyzing the number of new daily COVID-19 infections in Italy and the number of corporate defaults in the US.
\end{abstract}
%

\end{frontmatter}


\section{Introduction}


In the last years time series of counts has been experiencing an increasing interest among researchers and practitioners. Applications on count data span over several different areas, including finance \citep{fokianos2011log}, corporate defaults \citep{Agosto2016}, environmental economics \citep{Wang2014}, 
epidemiology \citep{davis2003}, COVID-19 infections and deaths \citep{khismatullina2020nonparametric,li2021will}, sports \citep{angelini2017parx,koopman2019}. Models for time series of counts are generally based on the assumption that the observations follow a Poisson distribution conditional on past observations and explanatory variables, see \cite{davis2016} for an overview of recent developments on the econometric models for discrete values.
In particular, \cite{davis2003} propose an observation-driven Poisson model for which they derive the conditions for stationarity and ergodicity and the large-sample properties of the MLE. They apply their method to daily counts of asthma presentations at a Sydney hospital. \cite{fokianos2009poisson} develop and establish the asymptotic theory of a Poisson autoregression (PAR) model where the conditional distribution of the observations $\{y_t\}_{t \in \mathbb{Z}}$ given the intensity process $\{\lambda_t\}_{t \in \mathbb{Z}}$ is Poisson with parameter $\lambda_t$ and they apply this new model to describe the number of transactions per minute of the stock Ericsson B. \cite{Agosto2016} generalise the PAR model by including a set of exogenous covariates in the dynamic specification of the intensity process (PARX) to improve the model performance for forecasting the number of aggregate corporate default counts in US.
\cite{Wang2014} propose a self-excited threshold structure in the PAR model specification 
by considering a two regime-switching approach according to the magnitude of the lagged observations, in order to model the number of major earthquakes in the world.
 
In this paper we contribute to this literature by introducing a new class of Poisson AutoRegressive models allowing for time-variation in the intensity and in the parameters which is particularly suitable to capture several (non-standard) features, including structural breaks, regime-switching and non-linearities, overdispersion and negative correlation, thus providing better fitting and forecasting performance than the original PARX model.  The extension to the time-varying parameter specification of the PARX model, labeled TV-PARX model in what follows, is achieved using a score-driven approach.
Originally proposed by \cite{creal2013} and \cite{Harvey2013} to estimate the evolution of time-varying parameters, the Generalized Autoregressive Score (GAS) framework is now a popular methodology to develop time-varying econometrics models. 
\cite{Blasques2019} propose a new class of GAS models for continuous time series allowing for time-varying parameters and provide the theoretical foundations for their acceleration method by showing optimality in terms of Kullback-Leibler divergence.

The literature on GAS models for time series of counts has been growing remarkably in the very recent years. 
\cite{gorgi2020beta} uses the beta–negative binomial distribution for the narcotics trafficking reports in Sydney and the euro–pound sterling exchange rate, \cite{koopman2019} use the bivariate Poisson distribution for the number of goals in football matches and the Skellam distribution for the score difference, and \cite{gorgi2018} uses the Poisson
distribution as well as the negative binomial distribution for the number of offensive conduct reports in the city of Blacktown, Australia.

The TV-PARX model we propose generalizes the accelerated score-driven model by \cite{Blasques2019} in the framework of time series of counts and allowing for the inclusion of exogenous covariates and deterministic components in the model specification.
We provide conditions for stationarity and ergodicity, as well as asymptotic results for the MLE.
Moreover, we show the relevance of our modeling approach by focusing on a Monte Carlo simulation study and {two} main applications. 
First, we analyze the dynamics of new COVID-19 confirmed cases in Italy and, as in \cite{li2021will}, we found that the time-varying model specification is required. Indeed, the model parameters substantially change over time due to the different `waves' of the coronavirus infections and their association with the introduction of government's measures to control the spread of the virus, e.g. lockdown, curfew, and mass vaccination. 
Second, we study the dynamics of US corporate default counts. Using the data set analyzed by \cite{Agosto2016}, we show that the time-varying specification of the PARX model allows to better capture and forecast the number of corporate defaults and to disentangle the main drivers of the US recessions.  
  
The paper is organized as follows. Section \ref{sec:model} describes the new 
TV-PARX model and discusses the properties of TV-PARX. Section \ref{sec:ml_est} presents the MLE for the TV-PARX.  Section \ref{sec:monte_carlo} shows the finite sample properties and Section \ref{sec:empirical} reports the empirical illustrations on COVID-19 data and on corporate defaults. Section \ref{sec:conclusion} concludes the paper. Details on notation and all the proofs are provided in 
\ref{proofs}.

\section{Time-Varying Poisson Autoregression}
\label{sec:model}
Consider a time series of counts $\{ y_t \}_{t \in \mathbb{Z}}$ with a Poisson conditional distribution,
\begin{align}
\label{cond_distribution}
y_t | \mathcal{F}_{t-1} \sim \mathcal{P}\textit{ois}(\lambda_t) \tab t \in \mathbb{Z},
\end{align}
where $\mathcal{F}_{t-1}= \sigma\{ y_{t-1}, y_{t-2}, y_{t-3}, \dots \}$ denotes the filtration generated by the count process $\{y_{t-i}\}_{i>0}$ up to time $t-1$, whereas $\lambda_t := \mathbb{E}[y_t | \mathcal{F}_{t-1}]$ is the conditional mean of the count process, which is allowed to vary over time. Examples are the Poisson autoregressive models by \cite{fokianos2009poisson} and \cite{Agosto2016}. While these papers consider a time-invariant parameter approach, we
 assume that the evolution of the conditional mean is described by a time-varying coefficient model representation using the score-driven framework $-$ see e.g.\ \cite{creal2013} and \cite{Harvey2013} for a detailed review $-$ thus allowing for time variation in the parameters of the dynamic equation for $\lambda_t$. 
 \cite{Koopman2008} show that a score-driven model for Poisson time series of counts encompasses most of the observation-driven models considered by \cite{davis2003}. Then, our contribution is to generalize the score-driven model for the intensity process and formally derive the stochastic properties of our Poisson score-driven model and the asymptotic theory of the maximum likelihood estimator.
In particular, we consider a score-driven model with the observation equation in \eqref{cond_distribution}, which implies that the conditional $\log$-density of $y_t$ is given up to a constant by
\begin{align}
\label{cond_dens}
\log p(y_t | \mathcal{F}_{t-1})
=
{y_t} \log \lambda_t -\lambda_t.
\end{align} 

We propose the following nonlinear Poisson autoregressive model with exogenous covariates and deterministic components defined as
\begin{eqnarray}
\log \lambda_{t+1} &=& \omega + \beta \log \lambda_t + \alpha_{t+1} (y_t - \lambda_t)\lambda_t^{-1} +  \gamma_{t+1}^{\prime}x_t + \psi^\prime d_t  \label{eq:tv_lambda_parx}
\end{eqnarray}
where $x_t \in \mathbb{R}^m$ denotes $m$ exogenous covariates, whereas in $d_t \in \mathbb{R}^d$ are stacked the deterministic components, e.g.\ seasonal dummies, impulse and step dummies, trend functions, etc.
In the following, the model in \eqref{eq:tv_lambda_parx} is called Time-Varying Poisson AutoRegressive with eXogenous covariates, TV-PARX.
Indeed, while $(\omega, \beta) \in\mathbb{R}^2$ and $ \psi  \in \mathbb{R}^d$ are unknown static parameters, the key feature of model \eqref{eq:tv_lambda_parx} is that $\alpha_t$ and $\gamma_t$ are time-varying.

 Similarly to \cite{Blasques2019}, we characterize $\alpha_t$ by the following updating equation 
\begin{eqnarray}
\alpha_{t+1} 
&= 
\delta_\alpha + \phi_\alpha \alpha_t + \kappa_\alpha  (y_t - \lambda_t)\lambda_t^{-1} (y_{t-1} - \lambda_{t-1})\lambda_{t-1}^{-1},\label{eq:tv_alpha_parx}
\end{eqnarray}
where $\delta_\alpha$, $\phi_\alpha$ and $\kappa_\alpha$ are treated as fixed and unknown parameters. We note that, analogously to the models for continuous time series considered in \cite{Blasques2019}, the law of motion of $\alpha_t$ is driven by the product of current and past innovations. 
Moreover, by scaling the score innovations with the inverse of information quantity $\mathcal{I}_{t} = \lambda_t$, \cite{Harvey2013} shows that the law motion for $\log \lambda_t$ is given by the first-order autoregressive process in \eqref{eq:tv_lambda_parx} in the case of no exogenous covariates ($x_t=0$) and no deterministic components ($d_t=0$).
With the same rationale of $\alpha_t$, we define the score-driven-based updating of the parameters for the exogenous covariates by
\begin{eqnarray}
\gamma_{t+1} &=& \delta_{\gamma} +\phi_{\gamma}\gamma_{t} +\kappa_{\gamma} (y_t - \lambda_t)\lambda_t^{-1} x_t \label{eq:tv_gamma_parx}
\end{eqnarray}
where $\delta_\gamma$, $\phi_\gamma$ and $\kappa_\gamma$ are fixed and unknown parameters.
Allowing for time-varying coefficients of the exogenous covariates is particularly useful to capture dynamic effects of the covariates on the conditional mean of the count process. 
In particular, this feature of the model enables a much better fit of the structural breaks and time-varying causality due to the occurrence of exogenous shocks, which is particularly recurrent in economic and financial phenomena. 

Note also that in \eqref{eq:tv_lambda_parx} we adopt an exponential link function, which allows the conditional mean $\lambda_t$ to be always positive, i.e.\ there is no need to impose further restrictions to have $\lambda_t>0$, for any $t$. 
Moreover, there is no need to transform the $m$ exogenous covariates in $x_t$ to make them positive as done, e.g., in \cite{Agosto2016} (see Section \ref{sec:defaults} for an empirical illustration). 

\medskip

\noindent {\bf Remark 1.} 
Several models discussed in the literature are special cases of the model proposed here.
For instance, the model proposed by \cite{Harvey2013} is a special case of the model in \eqref{eq:tv_lambda_parx} with no covariates and deterministics, i.e.\ $x_t=0$ and $d_t=0$, and where the parameter $\alpha$ is time-invariant.
Also, the resulting model is equivalent to the Poisson model of \cite{davis2003} for the case of time-invariant coefficients of the exogenous covariates.  \hfill $\square$

\medskip

Despite the time-varying specification of the model in \eqref{eq:tv_lambda_parx}-\eqref{eq:tv_gamma_parx} is in the same vein of \cite{Blasques2019}, i.e.\ the coefficient $\alpha$ is allowed to change at each time period by means of the score-driven approach, with respect to \cite{Blasques2019} our proposal is developed to model time series of counts instead of continuous responses and we can include $m$ exogenous covariates in the model specification.
Moreover, our proposal extends the Poisson autoregressive model by \cite{fokianos2009poisson} and static PARX model by \cite{Agosto2016} by allowing for time-varying parameters, including the coefficients for the exogenous covariates.  

%

\medskip


We complete this section by obtaining sufficient conditions for geometric ergodicity of the Poisson autoregressive model described above. As usual in the literature, let $\{N_t(\,\cdot \,) \}_{t\in\mathbb{Z}}$ denotes a sequence of independent Poisson processes with unity intensity, so that the process $\{ y_t \}_{t\in\mathbb{Z}}$ in \eqref{cond_dens} can be expressed as 
\begin{align*}
y_t = N_t(\lambda_t),
\end{align*}
where $\lambda_t$ follows the dynamics given in \eqref{eq:tv_lambda_parx} and \eqref{eq:tv_gamma_parx}, see e.g.\ \cite{fokianos2009poisson}. We begin by introducing the first result.

\begin{prop} 
\label{prop_ergodic_lnLambda}
Consider the model in \eqref{eq:tv_lambda_parx} and \eqref{eq:tv_gamma_parx} under the condition $|\phi_\alpha|<1$. Then, $\{ \alpha_t \}_{t\in\mathbb{Z}}$ admits a strictly stationary and ergodic solution, with $\bar{\alpha} := \mathbb{E}[\alpha_t]= \delta_{\alpha} / (1 - \phi_\alpha) < \infty$. Moreover, under the additional conditions that $0<\beta<1$ and $ \beta| \beta + \bar{\alpha}| <1$, then, $\{ \log \lambda_t \}_{t\in\mathbb{Z}}$ is geometrically ergodic and, moreover, $\{ \lambda_t \}_{t\in\mathbb{Z}}$ is geometrically ergodic with $\mathbb{E}[\lambda_t]<\infty$.
\end{prop}

\noindent {\bf Remark 2.} 
In line with the $\log$-linear Poisson autoregressive model of \cite{fokianos2009poisson}, by considering the full model given by equations \eqref{eq:tv_lambda_parx} and \eqref{eq:tv_gamma_parx} instead of the restricted on with no covariates and dummy variables $x_t$ and $\delta_t$, respectively, does not crucially affect the conditions for stationarity and ergodicity stated in Proposition \ref{prop_ergodic_lnLambda}. In fact, since $\delta_t$ are deterministic, as soon as the exogenous covariate processes $\{ x_t \}_{t\in \mathbb{Z}}$ are Markov chains, we only need to retrive a set of separate conditions for their transition mechanism, together with need the additional assumption that $|\phi_\gamma|<1$. 
 \\

It is well-known that when dealing with nonlinear time series models it is usually not easy to establish clear and/or simple stationarity conditions. However, we note that the sufficient conditions given in Proposition \ref{prop_ergodic_lnLambda} are remarkably easy to check. Moreover, a direct consequence of Proposition \ref{prop_ergodic_lnLambda} is that if the process $\{ \log \lambda_t \}_{t\in\mathbb{Z}}$ is initialized at its stationary distribution, then the process $\{ y_t \}_{t\in \mathbb{Z}}$ will also be stationary and ergodic with a finite first moment $\mathbb{E}[y_t] = \mathbb{E}[\lambda_t] < \infty$; see \cite{davis2003}. 

We conclude this section with a new proposition which gives sufficient conditions for $\mathbb{E}[y_t^k]< \infty$, with $k$ is a strictly positive number.

\begin{prop}
\label{prop_sec_mom}
Under the assumptions of proposition \ref{prop_ergodic_lnLambda} $\{ \log \lambda_t \}_{t\in\mathbb{Z}}$ and $\{ y_t \}$ have finite moments of any order, i.e.\ for any $k>0$, $\mathbb{E}[|\log \lambda_t|^k]< \infty$ and $\mathbb{E}[y_t^k]< \infty$. 
\end{prop}

It is worth to note that Proposition \ref{prop_sec_mom} we do not need stricter contraction conditions in order to obtain the finiteness of the higher-order moments. This feature of our model is particularly useful since it allow us to derive bounds for proving the asymptotic properties of the MLE, which is the main argument of the next Section, without imposing further restrictions on the data generating process.


\section{Maximum likelihood estimation}
\label{sec:ml_est}
The TV-PARX model can be easily estimated by standard ML, since the predictive log-likelihood is available in closed form. In the following, for the sake of simplicity, we consider the case of no exogenous covariates and no deterministic components, i.e.
\begin{eqnarray}
\log \lambda_{t+1} &=& \omega + \beta \log \lambda_t + \alpha_{t+1} (y_t - \lambda_t)\lambda_t^{-1} 
\label{ln_tv_lambda}
\end{eqnarray}
which corresponds to \eqref{eq:tv_lambda_parx} where $x_t=0$ and $d_t=0$, and is labeled the TV-PAR model. Note that in the case where $x_t \neq 0$ and $d_t \neq 0$, the asymptotic theory developed for our MLE below, could be adapted straightforwardly using partial likelihood theory, see \cite{Wong1986}.


Denote the parameter vector $\boldsymbol{\theta} = (\boldsymbol{\xi}^\prime, \boldsymbol{\psi}^\prime)^\prime \in \boldsymbol{\Theta} \subset \mathbb{R}^{5}$, where $\boldsymbol{\xi} = (\omega, \beta)^\prime$, $\boldsymbol{\psi} = (\delta_\alpha, \phi_\alpha, \kappa_\alpha)^\prime$, and $\boldsymbol{\Theta}$ is a compact parameter space. The true value of the combined parameter vector is denoted by $\boldsymbol{\theta}_0$, and we assume that $\{y_t\}_{t=1}^T$ is generated according to the TV-PAR process described by equations \eqref{cond_distribution}, \eqref{ln_tv_lambda} and \eqref{eq:tv_alpha_parx}, evaluated at $\boldsymbol{\theta}_0$. The $\log$-likelihood function is given by
\begin{align}
\label{approx_log_lik_tot}
\hat{L}_T(\boldsymbol{\theta})
=
\frac{1}{T}
\sum_{t=1}^T \hat{l}_t(\boldsymbol{\theta}),
\end{align}
where
\begin{align}
\label{approx_log_lik}
\hat{l}_t(\boldsymbol{\theta})
= 
y_t \log\hat{\lambda}_t(\boldsymbol{\theta})  - \hat{\lambda}_t(\boldsymbol{\theta}) 
\end{align}
and 
\begin{eqnarray}
\label{ln_lambda_theta}
\log\hat{\lambda}_{t+1}(\boldsymbol{\theta}) 
&=&
\omega + \beta \log\hat{\lambda}(\boldsymbol{\theta}) 
+ \hat{\alpha}_{t+1}(\boldsymbol{\theta}) 
(y_t - \hat{\lambda}_t(\boldsymbol{\theta}) )\hat{\lambda}^{-1}_t(\boldsymbol{\theta})\\
\label{alpha_theta}
\hat{\alpha}_{t+1}(\boldsymbol{\theta}) 
&=& 
\delta_\alpha + \phi_\alpha \hat{\alpha}_t(\boldsymbol{\theta}) 
 + \kappa_\alpha  (y_t - \hat{\lambda}_t(\boldsymbol{\theta}) )\hat{\lambda}_t^{-1}(\boldsymbol{\theta})  (y_{t-1} - \hat{\lambda}_{t-1}(\boldsymbol{\theta}) )\hat{\lambda}_{t-1}^{-1}(\boldsymbol{\theta}).
\end{eqnarray}
Note that the time-varying parameters in \eqref{ln_lambda_theta} and \eqref{alpha_theta} are obtained recursively by using some fixed starting values $\hat{\lambda}_{1}(\boldsymbol{\theta})\in \mathbb{R}_+$, $\hat{\alpha}_1(\boldsymbol{\theta})\in \mathbb{R}$ and the observations $\{y_t\}_{t=1}^T$.

The MLE $\hat{\boldsymbol{\theta}}_T$ of $\boldsymbol{\theta}$ is defined as
\begin{align}
\label{MLE}
\hat{\boldsymbol{\theta}}_T
=
\argmax_{\theta \in \boldsymbol{\Theta}}
\hat{L}_T(\boldsymbol{\theta}).
\end{align}

The conditions stated in Proposition \ref{prop_ergodic_lnLambda} implies that $\alpha_t$ and $\log \lambda_t$ have stationary representations. Now, for the likelihood analysis and the asymptotic properties of the MLE we need to derive the stochastic limit properties of the filtered parameters $\{ \hat{\lambda}_t(\boldsymbol{\theta})  \}_{t \in \mathbb{N}}$ and $\{ \hat{\alpha}_t(\boldsymbol{\theta})  \}_{t \in \mathbb{N}}$, which can be seen as two stochastic functions. 
In particular, because of the initializations $\hat{\lambda}_{1}(\boldsymbol{\theta})$ and $\hat{\alpha}_{1}(\boldsymbol{\theta})$, the filtered parameters in \eqref{ln_lambda_theta} and \eqref{alpha_theta}, and therefore the approximate $\log$-likelihood function in \eqref{approx_log_lik_tot}, are in general non-stationary. Therefore, in the next proposition, we derive sufficient conditions, which ensure that the effect that $\hat{\lambda}_{1}(\boldsymbol{\theta})$ and $\hat{\alpha}_{1}(\boldsymbol{\theta})$ have on the approximate $\log$-likelihood function $\{ \hat{L}_T (\boldsymbol{\theta}) \}_{t \in \mathbb{N}}$ vanish almost surely and exponentially fast (e.a.s.) uniformly over $\boldsymbol{\Theta}$. This phenomenon is well-known in the literature of nonlinear time series models, which is usually referred to the notion of invertibility, see \cite{Straumann_Mikosch2006}. Stationarity and ergodicity of $\{\log \lambda_t \}$ (as established in Proposition \ref{prop_ergodic_lnLambda}), and Proposition \ref{prop_invertibility} below, are the key ingredients for establish asymptotic properties of the MLE.

\begin{prop} 
\label{prop_invertibility}
Consider the model in \eqref{ln_tv_lambda} and \eqref{alpha_theta} under the conditions of Proposition \ref{prop_ergodic_lnLambda}. Let $\boldsymbol{\Theta}$ be compact 
with $\kappa_\alpha>0$ and assume that
\begin{align}
\mathbb{E}\Big[\log \sup_{\boldsymbol{\theta} \in \boldsymbol{\Theta}}
\Big|
\beta \exp\Big\{ 
\omega + \bar{\alpha}_{t+1}(\boldsymbol{\theta})(y_t e^{-\ell} - 1) - \ell(1-\beta)
\Big\} \Big| \Big] < 0, \label{eq:contraction}
\end{align}
where 
\begin{align*}
\bar{\alpha}_{t+1}(\boldsymbol{\theta})
&:=
\delta_\alpha + \phi_\alpha \bar{\alpha}_{t}(\boldsymbol{\theta})
+ \kappa_\alpha(y_t e^{-\ell} - 1)(y_{t-1} e^{-\ell} - 1),\\
\ell
&:=
\frac{\omega - \frac{\delta_\alpha + \kappa_\alpha}{1 - \phi_\alpha}}{1 - \beta}
\end{align*}
Then,
\begin{align}
\label{contr_cond}
\sup_{\boldsymbol{\theta} \in \boldsymbol{\Theta}}|
\hat{\alpha}_t(\boldsymbol{\theta}) - {\alpha}_t(\boldsymbol{\theta})
| \xrightarrow[]{e.a.s.}  0
\tab 
\sup_{\boldsymbol{\theta} \in \boldsymbol{\Theta}}|
\hat{\lambda}_t(\boldsymbol{\theta}) - {\lambda}_t(\boldsymbol{\theta})
| \xrightarrow[]{e.a.s.}  0
\,\,\,\,\,\text{as} \,\,\,\,\, t \rightarrow \infty,
\end{align}
where $\{ {\alpha}_t(\boldsymbol{\theta}) \}_{t\in\mathbb{Z}}$ and $\{ {\lambda}_t(\boldsymbol{\theta}) \}_{t\in\mathbb{Z}}$ are stationary and ergodic.

Moreover, if
\begin{align}
\label{contr_cond_mom}
 \sup_{(\boldsymbol{\theta}\times y) \in 
 (\boldsymbol{\Theta} \times \mathbb{N}_0)}
\Big|
\beta \exp\Big\{ 
\omega + \bar{\alpha}_{t+1}(\boldsymbol{\theta})(y e^{-\ell} - 1) - \ell(1-\beta)
\Big\} \Big| <1,
\end{align}
then $\exists m >0$ such that $\mathbb{E}[\sup_{\boldsymbol{\theta} \in \boldsymbol{\Theta}} |\lambda_t(\boldsymbol{\theta})|^m]<\infty$.													
\end{prop}


Condition \eqref{eq:contraction} is a sufficient condition which ensures the invertibility of the TV-PARX model by an application of Proposition 5.2.12 in \cite{Straumann2005}, or equivalently Theorem 2.8 in \cite{Straumann_Mikosch2006}. It is different from the conditions for stationarity and ergodicity given in Theorem \ref{prop_ergodic_lnLambda}, and moreover, it is only concerned with the filtering recursion relative to $\lambda_t(\boldsymbol{\theta})$ because the stationarity condition imposed in Proposition \ref{prop_ergodic_lnLambda} for $\alpha_t(\boldsymbol{\theta})$ is also sufficient for filter invertibility. The complex form of the contraction condition in \eqref{eq:contraction} is due to the exponential function. A similar problem with related discussion can be found in \cite{Wintenberger2013} for the EGARCH(1,1) model. Nevertheless, despite its sufficiency, this condition it is still important as it underline that the invertibility region is not degenerate. 

Finally, by slightly reinforcing the contraction condition in \eqref{eq:contraction} as in \eqref{contr_cond_mom}, Proposition \ref{prop_invertibility} also ensures the existence of an arbitrary large number of unconditional moments for the stationary and ergodic intensity process $\{\lambda_t(\boldsymbol{\theta})\}_{t\in\mathbb{Z}}$. This is a crucial property which will be useful for proving the consistency and asymptotic normality of the MLE.


To derive the asymptotic properties of the MLE we follow the classic theory in \cite{white_1994}.
First, for consistency we consider the almost sure uniform convergence of the approximate $\log$-likelihood function in \eqref{approx_log_lik_tot}. In particular, Lemma \ref{lemma_as_approxLogLik} below ensures that the average $\log$-likelihood $\hat{L}_T(\boldsymbol{\theta})$ converges uniformly to a function $L(\boldsymbol{\theta})$.

\begin{lemma}
\label{lemma_as_approxLogLik}
Let the conditions of Propositions \ref{prop_ergodic_lnLambda}-\ref{prop_invertibility} hold true. Then,
\begin{align*}
\sup_{\boldsymbol{\theta}\in\boldsymbol{\Theta}}
| \hat{L}_T(\boldsymbol{\theta}) - {L}(\boldsymbol{\theta}) |
\xrightarrow[]{a.s.}  0
\,\,\,\,\,\text{as} \,\,\,\,\, T \rightarrow \infty,
\end{align*}
where ${L}(\boldsymbol{\theta}) := \mathbb{E}[l_1(\boldsymbol{\theta})] = \mathbb{E}[y_1\log\lambda_1(\boldsymbol{\theta}) - \lambda_1(\boldsymbol{\theta})]$.
\end{lemma}


Second, we verify the identifiability of the model parameterization, which is the argument of the next Lemma.

\begin{lemma}
\label{lemma_identifiability}
The true parameter vector $\boldsymbol{\theta}_0$ is the unique maximizer of ${L}(\boldsymbol{\theta})$ in $\boldsymbol{\Theta}$; that is, for any $\boldsymbol{\theta} \neq \boldsymbol{\theta}_0$, then ${L}(\boldsymbol{\theta}) < {L}(\boldsymbol{\theta}_0)$.
\end{lemma}


As a consequence, given the results obtained in Lemma \ref{lemma_as_approxLogLik} and Lemma \ref{lemma_identifiability}, an application of Theorem 2.11 of \cite{white_1994} implies the strong consistency of the MLE.

\begin{theorem}
\label{thm_consistency}
Let the conditions of Propositions \ref{prop_ergodic_lnLambda}-\ref{prop_invertibility} hold true. Then,
$\hat{\boldsymbol{\theta}}_T
\xrightarrow[]{a.s.}
\boldsymbol{\theta}_0$ as
$T \rightarrow \infty.$
\end{theorem}


Next, we discuss the asymptotic distribution of the MLE. 
Consider the normalized score evaluated at $\boldsymbol{\theta}_0$
\begin{align}
\label{score_fun}
\boldsymbol{\bar{\eta}}_{T}
=
\sqrt{T}
\frac{\partial L_T(\boldsymbol{\theta}_0)}{\partial \boldsymbol{\theta}}
=
\frac{1}{\sqrt{T}} \sum_{t=1}^{T} \boldsymbol{\eta}_t
\ \ , \ \ 
\boldsymbol{\eta}_t = e_t \lambda_t \log\lambda^{\boldsymbol{\theta}_0}_t
,
\end{align}
where $e_t := (y_t - \lambda_t(\boldsymbol{\theta}_0))\lambda^{-1}_t(\boldsymbol{\theta}_0)$, $\lambda_t := \lambda_t(\boldsymbol{\theta}_0)$ and $\log\lambda^{\boldsymbol{\theta}_0}_t := \frac{\partial\log\lambda_t(\boldsymbol{\theta})}{\partial \boldsymbol{\theta}}|_{\boldsymbol{\theta}=\boldsymbol{\theta}_0}$. Since $\boldsymbol{\theta} = (\boldsymbol{\xi}^\prime, \boldsymbol{\psi}^\prime)^\prime$, we have that
$ \log\lambda^{\boldsymbol{\theta}_0}_t=( (\log\lambda^{\boldsymbol{\xi}_0}_t)^\prime, (\log\lambda^{\boldsymbol{\psi}_0}_t)^\prime )^\prime$, where $\log\lambda^{\boldsymbol{\xi}_0}_t = \left( \frac{\partial\log\lambda_t(\boldsymbol{\theta}_0)}{\partial \omega}^\prime,\frac{\partial\log\lambda_t(\boldsymbol{\theta}_0)}{\partial \beta}^\prime \right) ^\prime$ and $\log\lambda^{\boldsymbol{\psi}}_t = \alpha^{\boldsymbol{\psi}}_t e_t$, with $\alpha^{\boldsymbol{\psi}_0}_t = \left(\frac{\partial\alpha_t(\boldsymbol{\theta}_0)}{\partial \delta_{\alpha}}^\prime,
\frac{\partial\alpha_t(\boldsymbol{\theta}_0)}{\partial \phi_{\alpha}}^\prime,
\frac{\partial\alpha_t(\boldsymbol{\theta}_0)}{\partial \kappa_{\alpha}}^\prime \right)$, as defined in Lemma \ref{lemma_deriv_proc} in the Appendix.

Asymptotic normality then follows by a standard central limit theorem for martingale difference sequences, see 
e.g.\ Corollary 3.1 of \cite{Hall1980}, by showing that
\begin{align}
\label{condition_hall1}
 \sum_{t=1}^T
\mathbb{E}[\boldsymbol{\eta}_t \boldsymbol{\eta}_t^\prime | \mathcal{F}_{t-1}]
 \xrightarrow[]{P} \boldsymbol{V},
\end{align}
where
\begin{align}
\label{V}
\boldsymbol{V}
=
\lim_{T\rightarrow\infty}
\frac{1}{T} \sum_{t=1}^T \lambda_t (\log\lambda^{\boldsymbol{\theta}_0}_t 
\log\lambda^{\boldsymbol{\theta}_0^{\prime}}_t),
\end{align}
and, the Lindberg condition holds, i.e., $\forall \epsilon > 0$,
\begin{align}
\label{condition_hall2}
\mathbb{E}[\boldsymbol{\eta}_t \boldsymbol{\eta}_t^\prime
\mathbbm{1}(\| \boldsymbol{\eta}_t \| > \epsilon)
 | \mathcal{F}_{t-1}]
 \xrightarrow[]{P} \boldsymbol{0}	.
\end{align}
These conditions are verified in the following proposition.

\begin{prop}
\label{prop_norm_score}
Let the conditions of Propositions \ref{prop_ergodic_lnLambda}-\ref{prop_invertibility} hold true. Moreover, assume that $\mathbb{E}[|A_t|^2]<1$ and $0<|\phi_\alpha|<1$, then \eqref{score_fun}, \eqref{condition_hall1} and \eqref{V} hold and the score function in \eqref{score_fun} 
satisfies
\begin{align*}
\boldsymbol{\bar{\eta}}_{T}
\overset{d}\to 
\mathcal{N}(\boldsymbol{0}, \boldsymbol{V}).
\end{align*}
\end{prop}


Denote the information matrix at the true value by $\boldsymbol{J}$, i.e.
\begin{align}
\label{J_matrix}
\boldsymbol{J}
=
\mathbb{E}\bigg[
\frac{1}{\lambda_t}
\frac{\partial \lambda_t(\boldsymbol{\theta}_0)}{\partial \boldsymbol{\theta}}
\frac{\partial \lambda_t(\boldsymbol{\theta}_0)}{\partial \boldsymbol{\theta}^\prime}
\bigg].
\end{align}
Then, we finally state the following theorem.

\begin{theorem}
\label{thm_asy_norm}
Let the conditions of Propositions \ref{prop_ergodic_lnLambda}-\ref{prop_norm_score} hold true. Moreover, assume that the true parameter vector $\boldsymbol{\theta}_0$ is in the interior of $\boldsymbol{\Theta}$. Then, the MLE $\hat{\boldsymbol{\theta}}_T$ is asymptotically normal,
\begin{align*}
\sqrt{T}(\hat{\boldsymbol{\theta}}_T - {\boldsymbol{\theta}}_0)
\Rightarrow
\mathcal{N}(\boldsymbol{0}, \boldsymbol{J}^{-1}).
\end{align*} 
\end{theorem}


\noindent {\bf Remark 3.}
It is worth noting that the results obtained for the MLE $\hat{\boldsymbol{\theta}}_T$ in this section hold true even if the conditional distribution is misspecified. This is a consequence of the fact that the Poisson distribution belongs to the linear exponential family. In particular, as proved by \cite{white1982regularity} and \cite{Gourieroux1984}, under high-level assumptions the MLE is a QMLE and maintain the same properties, the only change will occur on the asymptotic variance-covariance matrix of $\sqrt{T}(\hat{\boldsymbol{\theta}}_T - {\boldsymbol{\theta}}_0)$, which will assume the classic sandwitch form, that is
\begin{align*}
\sqrt{T}(\hat{\boldsymbol{\theta}}_T - {\boldsymbol{\theta}}_0)
\Rightarrow
\mathcal{N}(\boldsymbol{0}, \boldsymbol{\Sigma}),
\end{align*}
where $\boldsymbol{\Sigma} = \boldsymbol{J}^{-1}\boldsymbol{I}\boldsymbol{J}^{-1}$, with $\boldsymbol{J}$ as in equation \eqref{J_matrix} and $\boldsymbol{I}= \mathbb{E}\Big[ \frac{\mathbb{V}[y_t|\mathcal{F}_{t-1}]}{\lambda^2_t(\boldsymbol{\theta}_0)}\frac{\partial \lambda_t(\boldsymbol{\theta}_0)}{\partial \boldsymbol{\theta}}
\frac{\partial \lambda_t(\boldsymbol{\theta}_0)}{\partial \boldsymbol{\theta}^\prime}\Big]$. See also \cite{Ahmad2016} or \cite{Aknouche2021}.


\section{Finite sample performance}
\label{sec:monte_carlo}

In this section we investigate the empirical performance in finite samples of the TV-PAR model compared to the standard PAR model. In particular, the aim of this Monte Carlo exercise is to investigate how quickly the $\lambda_t$ can adapt to (structural) changes. The time series are generated from the following DGP

\begin{equation}
y_t = N_t(\lambda_t^0), \label{eq:MC1}
\end{equation}

\noindent where $N_t(\cdot)$ denotes a sequence of independent Poisson processes with unity intensity. Similarly to \cite{Blasques2019}, for the time-varying $\lambda_t^0$, we consider a deterministic step-function
\begin{align}
\lambda_t^0
=
\begin{cases}
2 & \tab \text{if} \tab \sin(\gamma^{-1} 10^{-2} (\pi t - 1) \geq 0 \\
2+\delta & \tab \text{if} \tab \sin(\gamma^{-1} 10^{-2} (\pi t - 1) < 0
\end{cases}
\label{step}
\end{align}

\noindent where $\delta \in \{0,2,4 \}$ and $\gamma \in \{1, 1.5, 2\}$. In this Monte Carlo exercise we generate $m=1000$ time series of sample sizes $T \in \{250, 500, 1000\}$. 
Figure \ref{fig:lambda_0} reports the time varying intensity $\lambda_t^0$ for all the combinations of $\delta$ and $\gamma$ considered. For each series we estimate the time-varying PAR in \eqref{ln_tv_lambda} and the static PAR, i.e.\ where $\alpha_t=\alpha$ in \eqref{ln_tv_lambda}.
Table \ref{tab:monte_carlo_RMSE} reports the Root Mean Squared Errors (RMSE) for all the cases considered. The results in Table \ref{tab:monte_carlo_RMSE} show that the RMSE of the TV-PAR is smaller for all DGPs considered, except for, unsurprinsigly, the case of no steps. However, in the case of no steps, i.e.\ when $\delta=0$, the difference in RMSE between the time-varying and static models is relatively small in all the cases considered.

Conversely, when the $\delta>0$, the RMSE of the TV-PAR model is smaller than the one of the standard PAR model for the all cases considered. In particular, the higher the $\delta$, the better the performance of TV-PAR model compared to the standard PAR model, regardless of the sample size considered. Therefore, with more convoluted DGP dynamics, the time-varying model specification works better than to its static counterpart.

To illustrate, we report in Figure \ref{fig:mc_est} the cases of the DGPs with $\delta =4$ and $\gamma = 1$ for $T = 250, 500,$ and 1000. This figure shows the mean $\hat{\lambda}_t$ and the corresponding 95\% confidence bands for the TV-PAR (in green) and the static PAR (in red) models. 
From Figure \ref{fig:mc_est}, it clearly emerges that the time-varying model specification provides a much faster reaction and better fit to the jumps created by the step-function in \eqref{step}. 


\section{Empirical illustrations}
\label{sec:empirical}
In this section, we show the usefulness of the time-varying specification of the TV-PARX model by considering two empirical illustrations.
Section \ref{sec:covid} deals with the daily counts of infections of SARS-COV-2 virus in Italy until the end of May 2021.
Section \ref{sec:defaults} analyzes the monthly corporate default counts in the US as previously analyzed by \cite{Agosto2016}. 

\subsection{COVID-19}
\label{sec:covid}

The COVID-19 pandemic, started in 2019 in the city of Wuhan in China's region of Hubei and spread all around the world in 2020, has caused a tremendous public health emergency and a huge negative shock on the business cycle. Since then, researchers of various fields have started to work on that topic in order to help policy makers in their decisions.
Are the number of cases (deaths) predictable? Are the evolution of the number of cases (deaths) similar across countries? 
These questions have increased the interest among researchers and practitioners for new models suited for time series of counts. Epidemiologists, for instance, use reduced-form models to produce forecasts of the number of infections (or deaths), see for instance \cite{batista2020estimation} for a logistic growth model and \cite{giordano2020modelling} and \cite{hethcote2000mathematics} for reviews on several variants of the Susceptible, Infectious and Removed (SIR) model.

Due to the convoluted time evolution of the COVID-19 pandemic (see the number of new cases depicted in blue in Figure \ref{fig:emp_fitted_actual}), the time-varying specification of the model is expected to capture the different waves of the pandemic and its non-standard dynamics are likely to outperforms much better time-invariant alternatives. 
Among others, 
\cite{li2021will} found that the model parameters change both over time and across countries. Country-specific behaviour and  local, transient outbreaks of COVID-19 is also found by \cite{khismatullina2020nonparametric} and \cite{dong2020time} using non-parametric approaches and by \cite{palmer2021count} using Bayesian hierarchical models.

Therefore, we fit the TV-PARX model in \eqref{eq:tv_lambda_parx}-\eqref{eq:tv_gamma_parx} with $\gamma_t=0$ to the daily counts of confirmed new cases in Italy from the beginning of the COVID-19 pandemic outbreak (March 2020) to the end of May 2021, which involves three or four main waves of cases (see Figure \ref{fig:emp_fitted_actual}).
The counts dynamics in Figure \ref{fig:emp_fitted_actual} shows a seasonal component mainly due to reporting. 
Hence, we consider the following TV-PARX model 
\begin{eqnarray}
\log \lambda_{t+1} &=& \omega + \beta \log \lambda_t + \alpha_{t+1} (y_t - \lambda_t)\lambda_t^{-1} +   \psi^\prime d_t  
\nonumber
\\
\alpha_{t+1} &=& \delta_\alpha + \phi_\alpha \alpha_t + \kappa_\alpha  (y_t - \lambda_t)\lambda_t^{-1} (y_{t-1} - \lambda_{t-1})\lambda_{t-1}^{-1} 
\label{covid}
\end{eqnarray}
where $d_t$ denotes seasonal (daily) dummies.

For comparison purposes, we also fit a time-invariant PARX model where $\phi_\alpha = \kappa_\alpha=0$ in \eqref{covid}.
Table \ref{tab:estimates} shows the estimates and Table \ref{tab:empirical} reports the differences in the log-likelihoods, information criteria and Root Mean Squared Errors (RMSE) of the estimated models.
In terms of in-sample performance, these results show that the time-varying specification of the TV-PARX model is preferred over its time-invariant counterpart.
Moreover, albeit both the time-varying model (TV-PARX) and its time-invariant counterpart (PARX) provide a good fit to the observed data, from Figure \ref{fig:emp_fitted_actual} it emerges that the TV-PARX model approximates better the different waves of the COVID-19 pandemic in Italy.
In particular, the estimated time-varying parameter $\hat{\alpha}_t$ ranges from around -0.2 to 0.8, and the minimum value virtually corresponds to the end of the first wave (end of June 2020), while the maximum values are observed at specific time points which are closely related to the beginning of each wave, e.g.\ March 2020, September-November 2020 and January 2021, as well as abrupt changes in $\hat{\alpha}_t$ for the last wave begun in March 2021.

\subsection{Corporate defaults}
\label{sec:defaults}
In this section, we stress the usefulness of the time-varying specification of the TV-PARX model compared to the standard approach based on the constant parameter version of the PARX model by analyzing the count time series of corporate defaults in US.
\cite{Agosto2016} contributed to the literature on modeling and forecasting time series of corporate defaults, by exploring the stylized fact that defaults tend to cluster over time and investigating the causes of this phenomenon. In particular, using a (constant parameter) PARX model they discriminate between comovements in corporate solvency caused by common underlying macroeconomic and financial risk factors, i.e.\ ``systematic risk'', and feedback effects, i.e.\ ``contagion'', where, conditionally on the common risk factors, the current number of defaults affects the probability of other firms' future insolvency. As pointed out in \cite{Agosto2016}, one limitation of their approach is that the measure
of contagion ignores feedback effects from defaults to covariates; that is, the case when $x_t$ is affected by past defaults, which in turn will affect future defaults. This issue is softened by the time-varying model specification of the TV-PARX as the time-dependent parameters may allow to capture the feedback effects from past defaults to the covariates.
The data set on default counts analyzed by \cite{Agosto2016} consists of the number of
bankruptices among Moody's rated industrial firms in the US collected at monthly frequency between the 1982 and 2011, for a total of $T = 360$ observations, from the Moody's Credit Risk Calculator (CRC). The blue line in Figure \ref{fig:emp_fitted_actual_defaults} shows the default count time series. As discussed in \cite{Agosto2016}, Figure \ref{fig:emp_fitted_actual_defaults} discloses that default counts is a process characterised by (i) high persistence, (ii) clusters over time, and (iii) overdispersion.
These three stylized facts are satisfactorily captured by the PARX model. However, the peak of defaults observed in 2008 (Figure \ref{fig:emp_fitted_actual_defaults}) as well as the high increase in bankruptices around the year 2000, possibly interpretable as structural breaks, may be better captured by the time-varying specification of the TV-PARX model, which can locally allow for explosive roots rather than standard PARX model where the stationary condition must hold. 
Therefore, we consider the TV-PARX model, described in \eqref{eq:tv_lambda_parx}-\eqref{eq:tv_gamma_parx}, with no deterministic componenents, i.e.\ $d_t=0$.
For the sake of comparison, in the TV-PARX model specification we include the same covariates considered in the PARX model in \cite{Agosto2016}.\footnote{Note that, with the specification in logs of the TV-PARX model, we are able to consider the whole time series of LI, and not only the negative values (in modulus to ensure non-negativity) as in \cite{Agosto2016}.
}  In particular, we consider $x_t=(RV_t, LI_t)^\prime$, where 
$RV_t$ is the realized volatility computed
as a proxy of the S\&P 500 monthly realized volatility using daily squared returns, and 
$LI_t$ is the Leading Index released
by the Federal Reserve. More formally, we consider the following TV-PARX model
\begin{eqnarray}
\log \lambda_{t+1} &=& \omega + \beta \log \lambda_t + \alpha_{t+1} (y_t - \lambda_t)\lambda_t^{-1} +  \gamma_{1,t+1}RV_t + \gamma_{2,t+1}LI_t 
\nonumber
\\
\alpha_{t+1} &=& \delta_\alpha + \phi_\alpha \alpha_t + \kappa_\alpha  (y_t - \lambda_t)\lambda_t^{-1} (y_{t-1} - \lambda_{t-1})\lambda_{t-1}^{-1} \nonumber \\
\gamma_{1,t+1} &=& \delta_{\gamma_1} +\phi_{\gamma_1}\gamma_{1,t} +\kappa_{\gamma_1} (y_t - \lambda_t)\lambda_t^{-1} RV_t \nonumber\\
\gamma_{2,t+1} &=& \delta_{\gamma_2} +\phi_{\gamma_2}\gamma_{2,t} +\kappa_{\gamma_2} (y_t - \lambda_t)\lambda_t^{-1} LI_t.  \label{eq:emp_default}
\end{eqnarray}

We provide an analysis for the full sample 1982-2011 and compare the results for the time-invariant
PARX model specification. 
For comparison purposes, we also fit a time-invariant PARX model where $\phi_{\alpha}=\kappa_{\alpha}=\phi_{\gamma_1}=\kappa_{\gamma_1}=\phi_{\gamma_2}=\kappa_{\gamma_2}=0$ in \eqref{eq:emp_default}. The maximum likelihood estimates are reported in Table \ref{tab:estimates}.
According to the information criteria and the RMSE reported in Table \ref{tab:empirical}, the TV-PARX model provides a better in-sample fit than the PARX model, except for the BIC.\footnote{It is interesting to note that our TV-PARX model outperforms the best model specification in \cite{Agosto2016}, namely (non-logs) PARX(2,1).} 
As can be noted from Figure \ref{fig:emp_fitted_actual_defaults}, both models capture the actual default counts dynamics well, but the goodness of fit of the TV-PARX model is superior during the periods where the number of defaults is larger, especially around the year 2009.
Figure \ref{fig:emp_alphagamma_defaults} shows the filtered dynamic parameters, $\alpha_t$ and $\gamma_t$. From this figure, it can be noted that the estimated time-varying parameters are able to adjust according to the observed default count time series.
Interestingly, the time-varying specification allows us to determine the main driver of a specific crisis. 
In particular, the dynamics of the estimated parameters $\gamma_{1,t}$ and $\gamma_{2,t}$ in Figure \ref{fig:emp_alphagamma_defaults} suggests that financial volatility is the main driver of the spike of the defaults in the US during the 2008-09 crisis, while 
during the early 90s recession the main driver is mainly macroeconomic.    
Therefore, the time-varying version of the PARX model proposed in \eqref{eq:tv_lambda_parx}-\eqref{eq:tv_gamma_parx} provides a more  flexible alternative to the time-invariant counterpart.

\section{Conclusion}
\label{sec:conclusion}
In this paper, we have developed a novel class of score-driven models for time series of counts, allowing for time-varying parameters and the inclusion of exogenous variables. 
We have provided conditions for stationarity and ergodicity, as well as asymptotic results for the MLE. Moreover, we have shown the relevance of our modeling approach by focusing on a Monte Carlo simulation study. 
We have considered two empirical illustrations where we apply our nonlinear time-varying Poisson autoregressive model: the first applies the TV-PARX model with deterministic components to the number of COVID-19 infections in Italy; the second focuses on count time series of corporate defaults in the US by including (time-varying) exogenous covariates in the TV-PARX model specification. 
For these relevant illustrations, we find that the proposed time-varying framework is capable of improving the in-sample fit.
Further issues are left to future research. In particular, in this paper we have specified an univariate TV-PARX model. A multivariate extension could be extremely useful to capture the relationships between different time series of counts (e.g. COVID-19 cases or number of defaults in different countries) and is currently under investigation by the authors. Moreover, the time-varying specification of $\alpha_t$ could be used to construct a new test for the presence of structural breaks for time series of counts.

\newpage
\bibliographystyle{apalike}
\bibliography{references}

\newpage

\begin{table}[h]
\centering
 \begin{tabular}{c c c c c c c}
 \hline \hline
  & \multicolumn{2}{c}{$\gamma=1$}  & \multicolumn{2}{c}{$\gamma=1.5$}   & \multicolumn{2}{c}{$\gamma=2$}  \\
 & PARX & TV-PARX & PARX & TV-PARX  & PARX & TV-PARX  \\
 \hline  
 \multicolumn{7}{c}{$T=250$} \\
$\delta=0$	&	\textbf{0.0077}	&	0.0086	&	\textbf{0.0077}	&	0.0086	&	\textbf{0.0077}	&	0.0086	\\
$\delta=2$	&	0.4052	&	\textbf{0.3686}	&	0.3317	&	\textbf{0.2630}	&	0.3358	&	\textbf{0.2776}	\\
$\delta=4$	&  0.6331	&	\textbf{0.5928}	&	0.5034	&	\textbf{0.3749}	&	0.5002 &	\textbf{0.3842}	\\
 \hline 
 \multicolumn{7}{c}{$T=500$} \\
$\delta=0$	&	\textbf{0.0033}	&	0.0049 	&	\textbf{0.0033}	&	0.0049 	&	\textbf{0.0033}	&	0.0049 	\\
$\delta=2$	&	 0.4097	&	\textbf{0.3914}	&	 0.3729	&	\textbf{0.3477}	&	0.3371	&	\textbf{0.3115}	\\
$\delta=4$	&  0.6324 	&	\textbf{0.6092}	&	0.5726	&	\textbf{0.5359}	&	0.5198 &	\textbf{0.4810}	\\
 \hline 
 \multicolumn{7}{c}{$T=1000$} \\
$\delta=0$	&	\textbf{0.0026}	&	0.0031 	&	\textbf{0.0026}	&	0.0031 &	\textbf{0.0026}	&	0.0031 	\\
$\delta=2$	&	0.4245	&	\textbf{0.4121}	&	0.3784	&	\textbf{0.3644}	&	0.3373	&	\textbf{0.3212}	\\
$\delta=4$	&  0.6529 	&	\textbf{0.6363}	&	0.5863	&	\textbf{0.5626}	&	0.5224 &	\textbf{0.4931}	\\
  \hline \hline
\end{tabular} 
\caption{Root Mean Squared Error (RMSE) where the error is between the true $\lambda_t^0$ and the filtered parameter $\lambda_t$ from TV-PARX and PARX models for different true values of $\delta$ and $\gamma$.  
\label{tab:monte_carlo_RMSE}}
\end{table}

\newpage

\begin{table}[h]
\centering
 \begin{tabular}{c c c c c}
 \hline \hline
 & \multicolumn{2}{c}{\textbf{SARS-COV-2}}  & \multicolumn{2}{c}{\textbf{Defaults}} \\ 
 &  TV-PARX & PARX &  TV-PARX & PARX \\ 
 \hline  
$\omega          $& $\underset{(0.0002)}{0.0045}$ & $\underset{(0.0002)}{0.0039}$  &  $\underset{(0.0323)}{-0.0227}$ & $\underset{(0.0056)}{0.0219}$\\
$\beta           $& $\underset{(0.0001}{0.9990}$ & $\underset{(0.0001)}{0.9990}$ &  $\underset{(0.0386)}{0.9423}$ & $\underset{(0.0038)}{0.9441}$  \\
$\psi_1          $& $\underset{(0.0197)}{-0.0514}$ & $\underset{(0.0002)}{-0.0416}$ & - & - \\  
$\psi_2          $& $\underset{(0.0071)}{-0.1310}$ & $\underset{(0.0001)}{-0.1246}$ & - & - \\  
$\psi_3          $& $\underset{(0.0004)}{-0.2767}$ & $\underset{(0.0001)}{-0.2724}$ & - & - \\  
$\psi_4          $& $\underset{(0.0106)}{0.1739}$ & $\underset{(0.0005)}{0.1794}$ & - & - \\  
$\psi_5          $& $\underset{(0.0065)}{0.1578}$ & $\underset{(0.0002)}{0.1602}$ & - & - \\  
$\psi_6          $& $\underset{(0.0100)}{0.1043}$ & $\underset{(0.0006)}{0.1083}$ & - & - \\  
$\delta_{\alpha} $& $\underset{(0.0017)}{-0.0013}$ &  $\underset{(0.0021)}{0.7033}$ &  $\underset{(0.0061)}{-0.0012}$ & $\underset{(0.0021)}{0.2369}$ \\
$\phi_{\alpha}   $& $\underset{(0.0004)}{0.9984}$ & - &  $\underset{(0.0227)}{0.9946}$ & - \\
$\kappa_{\alpha} $& $\underset{(0.0250)}{0.5000}$ & - &  $\underset{(0.0798)}{0.0157}$ & -\\
$\delta_{\gamma_1} $& - & - &  $\underset{(2.6856)}{1.1954}$ & $\underset{(0.1994)}{11.8267}$ \\
$\phi_{\gamma_1}   $& - & - &  $\underset{(0.1418)}{0.9498}$ & - \\
$\kappa_{\gamma_1} $& - & - &  $\underset{(273.6449)}{917.6204}$ & - \\
$\delta_{\gamma_2} $& - & - &  $\underset{(0.0062)}{-0.0004}$ & $\underset{(0.0018)}{0.0010}$ \\
$\phi_{\gamma_2}   $& - & - &  $\underset{(0.3673)}{0.9766}$ & - \\
$\kappa_{\gamma_2} $& - & - &  $\underset{(0.0021)}{0.0070}$ & - \\
 \hline \hline
\end{tabular} 
\caption{Estimated parameters, standard errors in brackets. 
\label{tab:estimates}}
\end{table}

\newpage

\begin{table}[h]
\centering
 \begin{tabular}{c c c}
 \hline \hline
 & \textbf{COVID-19} & \textbf{Defaults} \\ 
 &  TV-PARX $-$ PARX &  TV-PARX $-$ PARX \\ 
 \hline  
log-likelihood &  422.53 & 15.11 \\ 
AIC  &  -841.05	& -18.21 \\ 
HQC  &  -837.79 & -8.94  \\ 
BIC  &  -832.78 &  5.10 \\ 
RMSE &   -14.29 & -0.01 \\ 
  \hline \hline
\end{tabular} 
\caption{Difference between the information criteria and the Root Mean Squared Error (RMSE) of the estimated TV-PAR and PAR models for the empirical analysis of the daily number of new COVID-19 cases in Italy and the number of monthly corporate defaults in the US.  
\label{tab:empirical}}
\end{table}

\newpage

\begin{figure}[h]
	\centering
	 \begin{tabular}{c c c}
	  \multicolumn{3}{c}{$T=250$} \\
	\includegraphics[width=3cm]{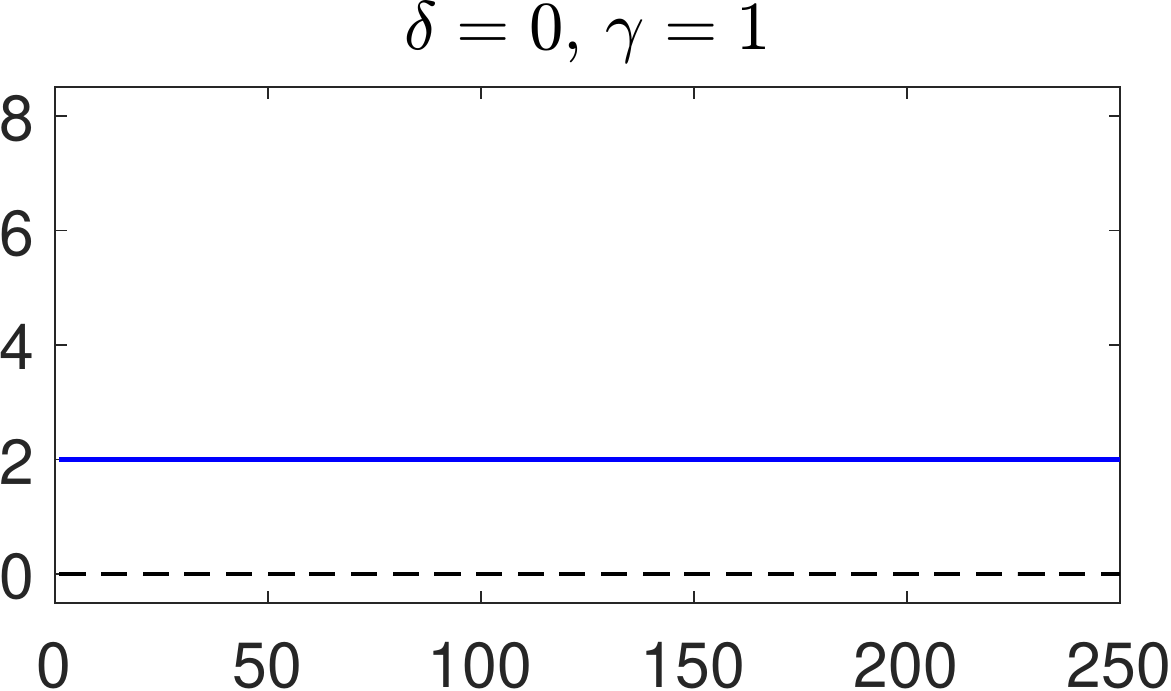} & \includegraphics[width=3cm]{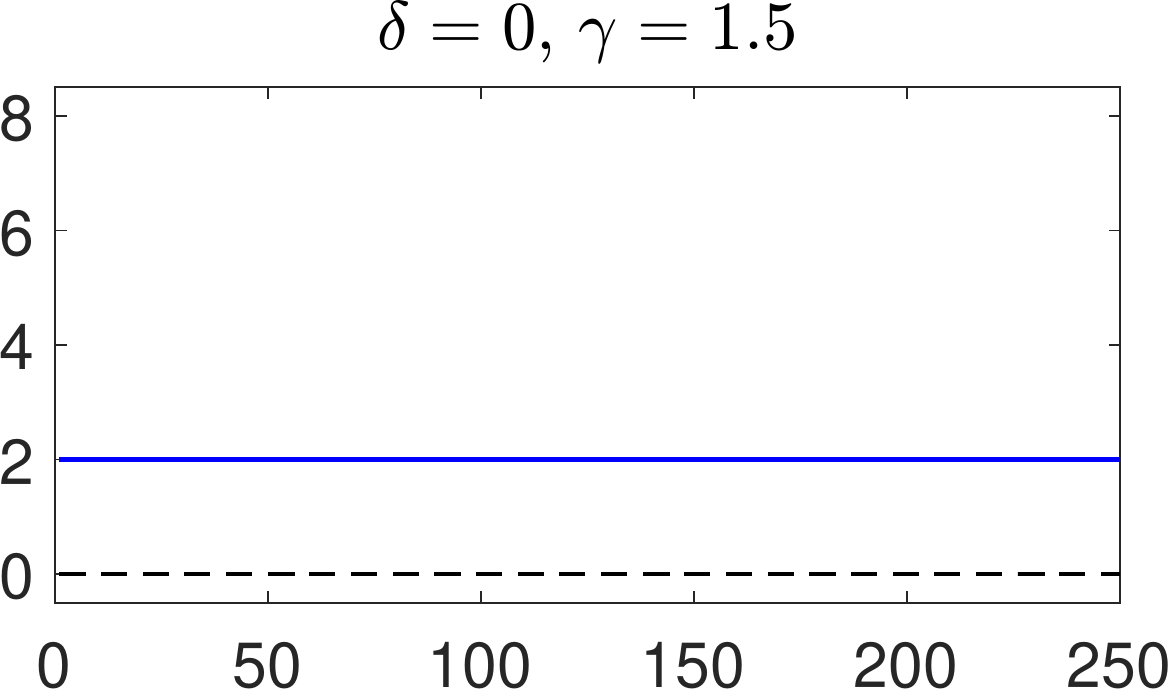} & \includegraphics[width=3cm]{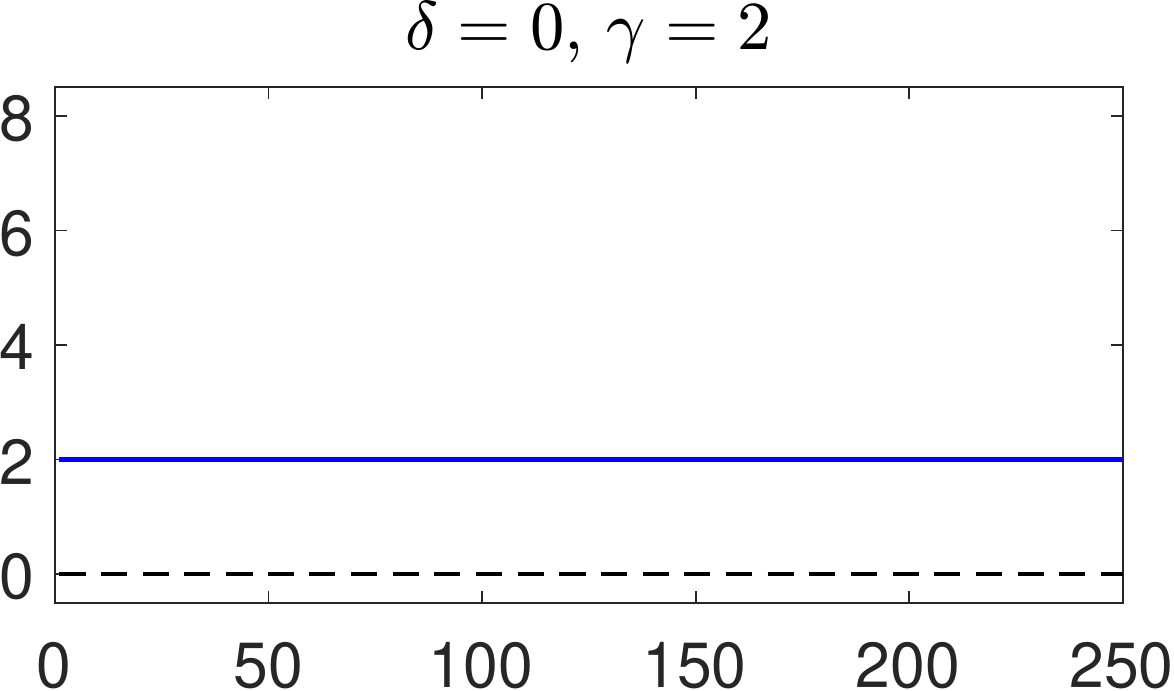} \\
		\includegraphics[width=3cm]{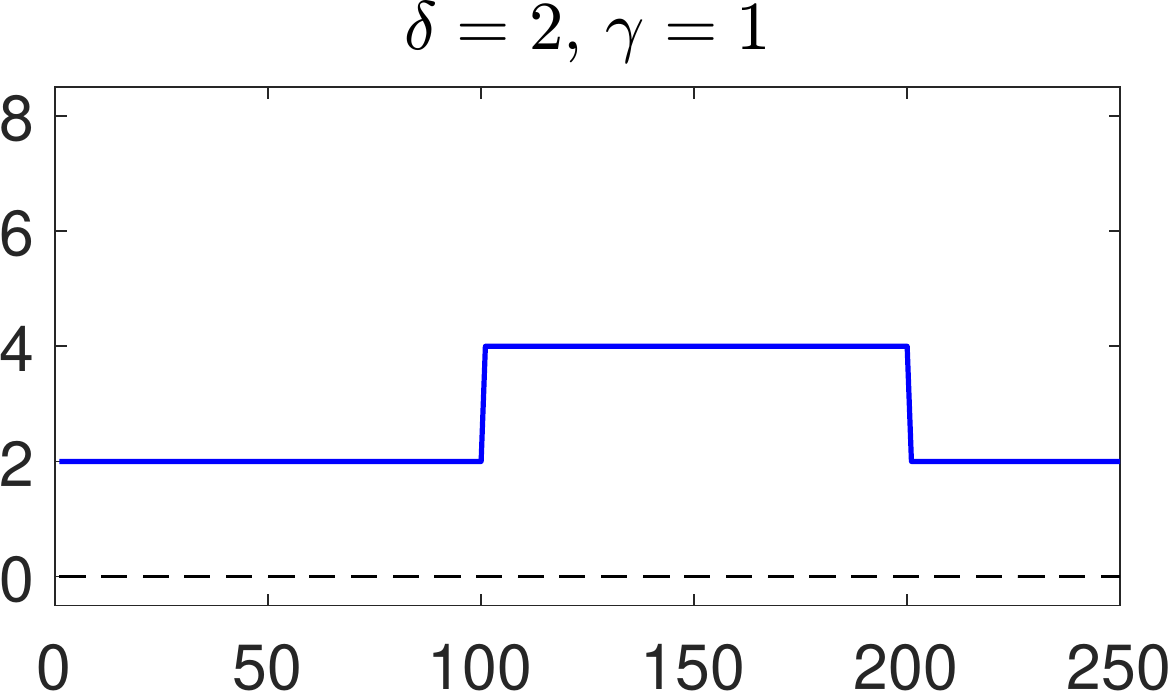} & \includegraphics[width=3cm]{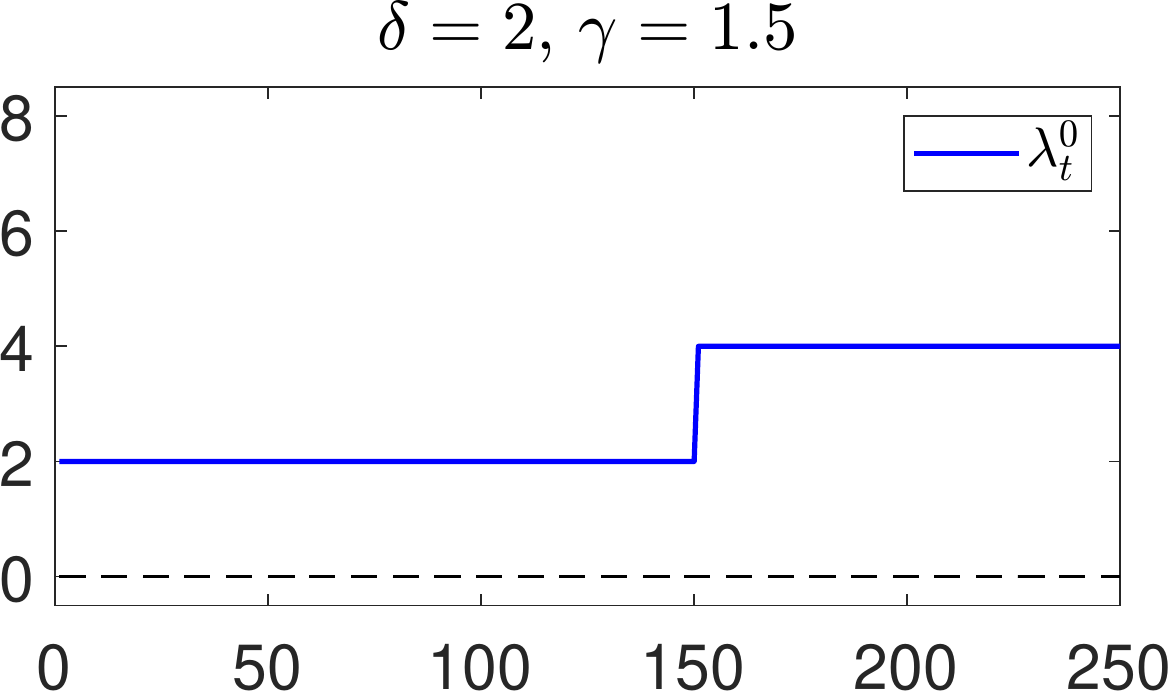} & \includegraphics[width=3cm]{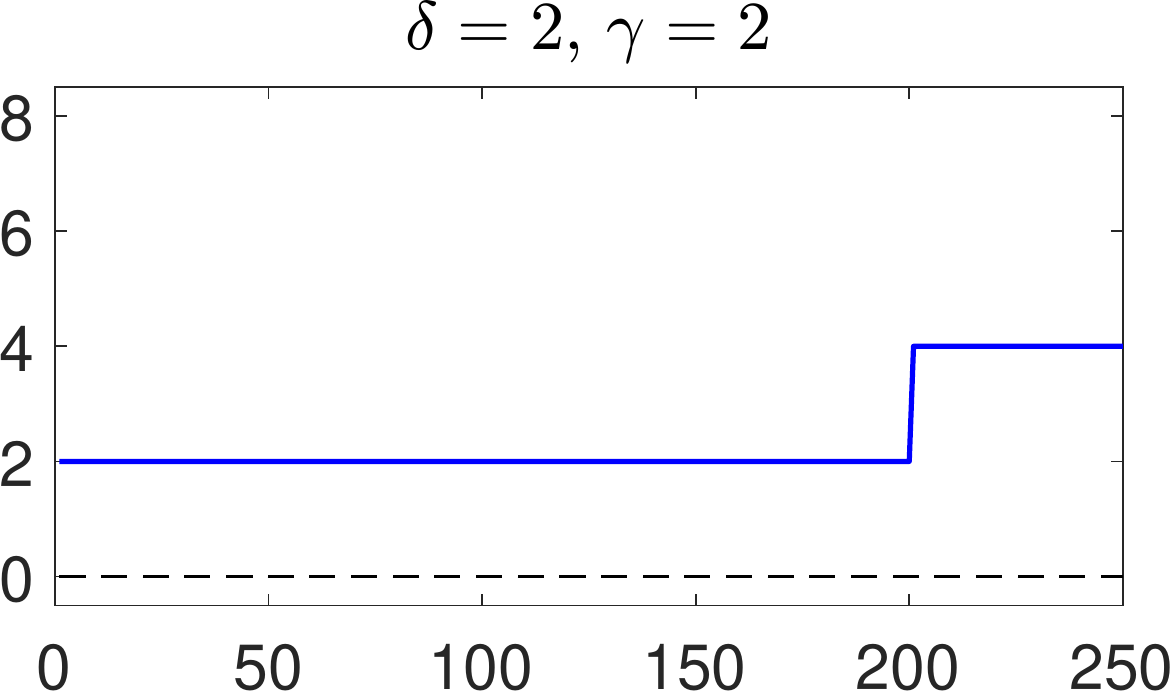} \\
			\includegraphics[width=3cm]{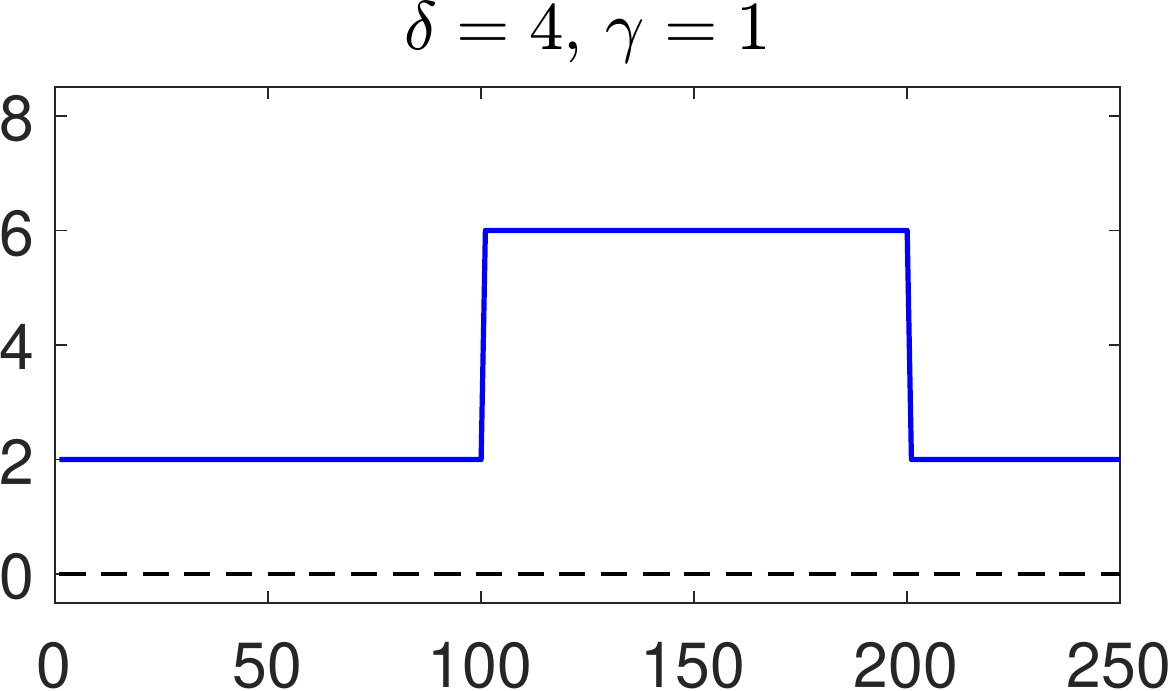} & \includegraphics[width=3cm]{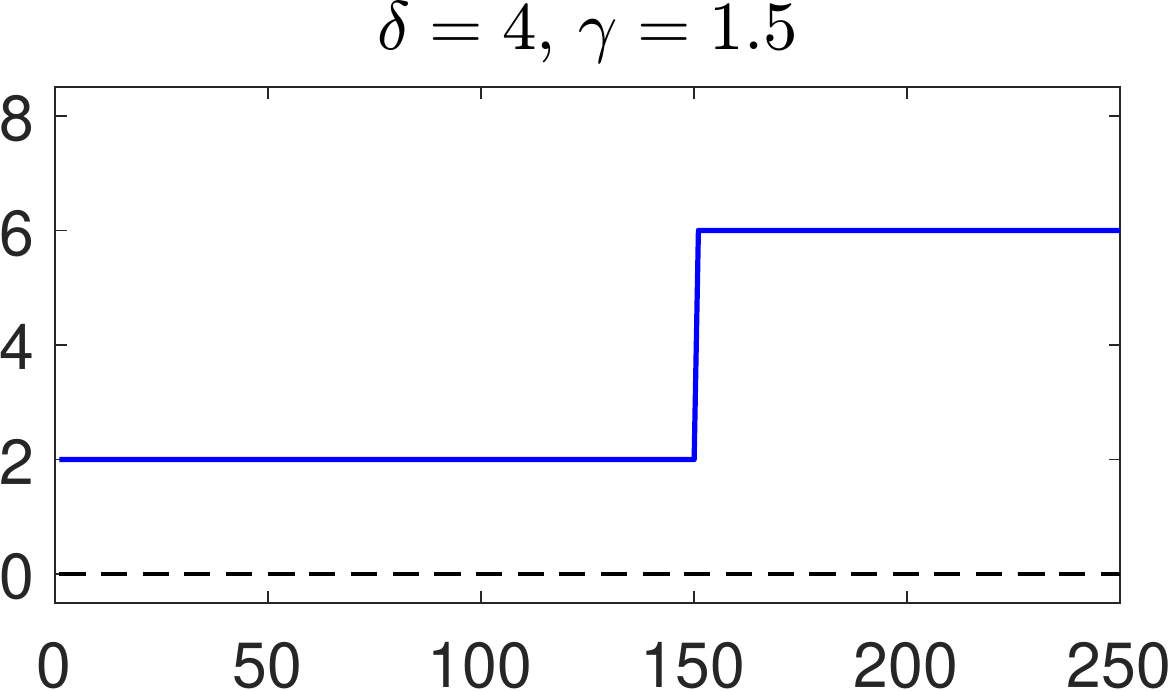} & \includegraphics[width=3cm]{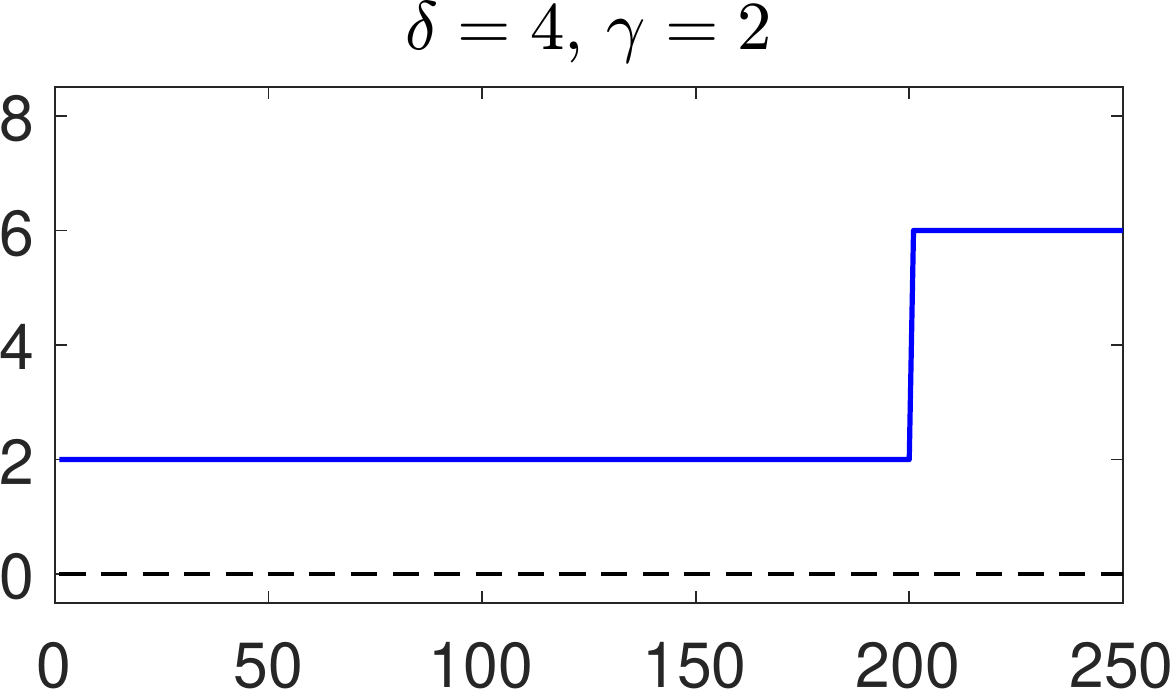} \\
				  \multicolumn{3}{c}{$T=500$} \\
	\includegraphics[width=3cm]{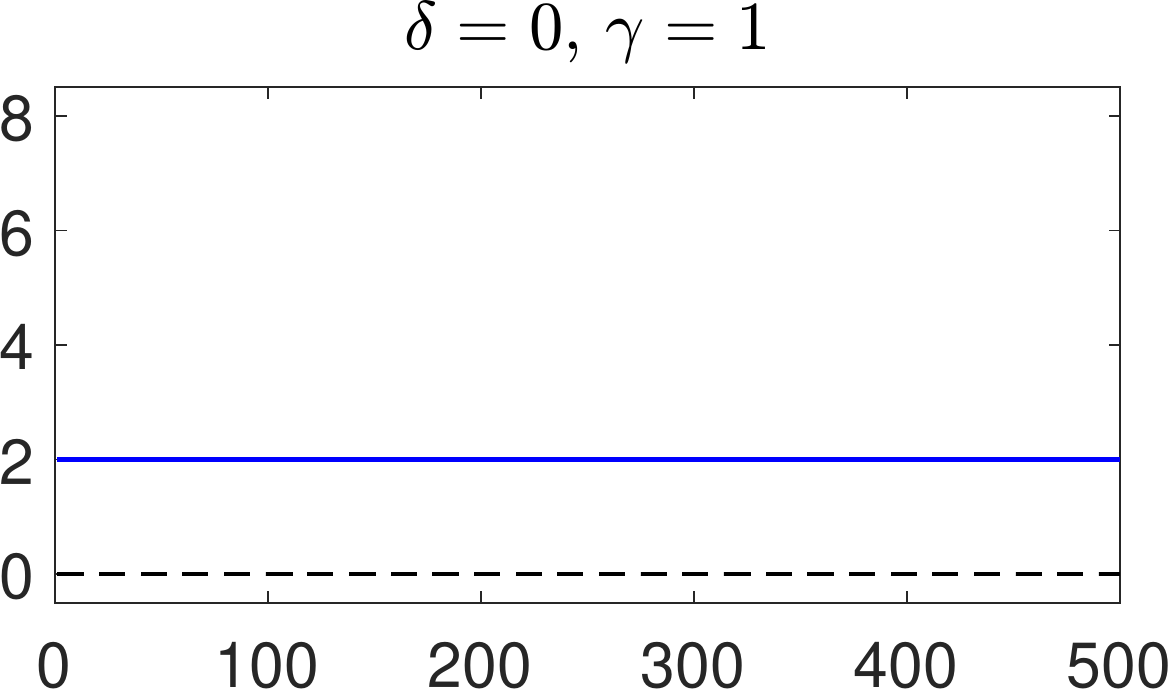} & \includegraphics[width=3cm]{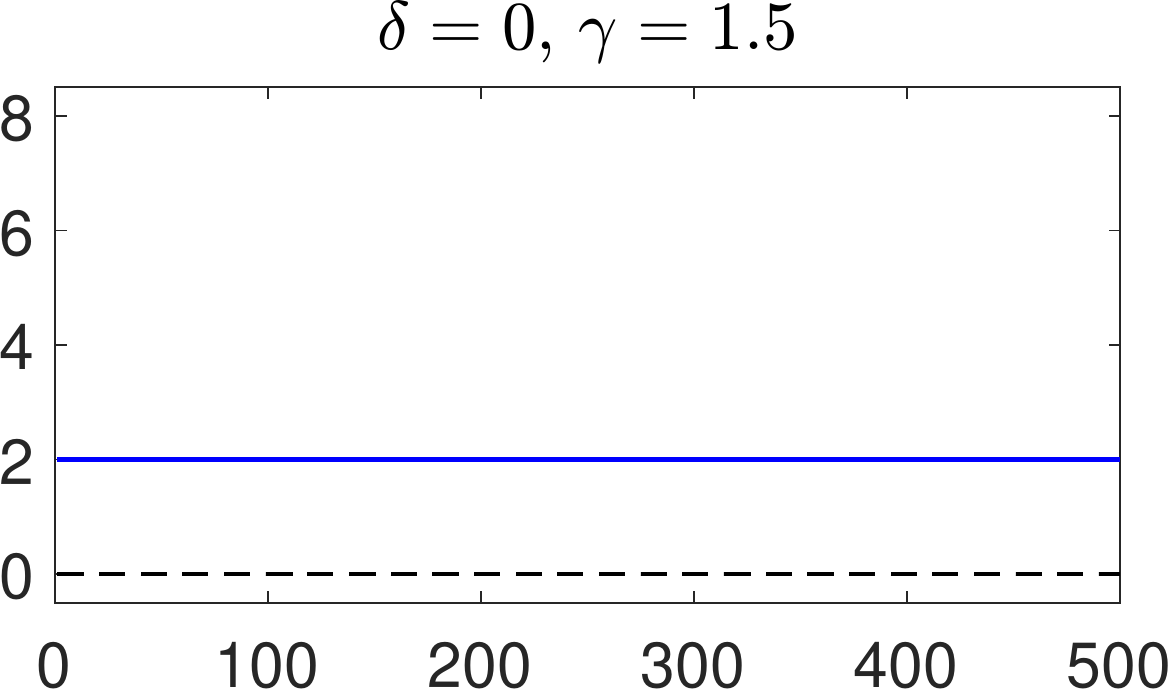} & \includegraphics[width=3cm]{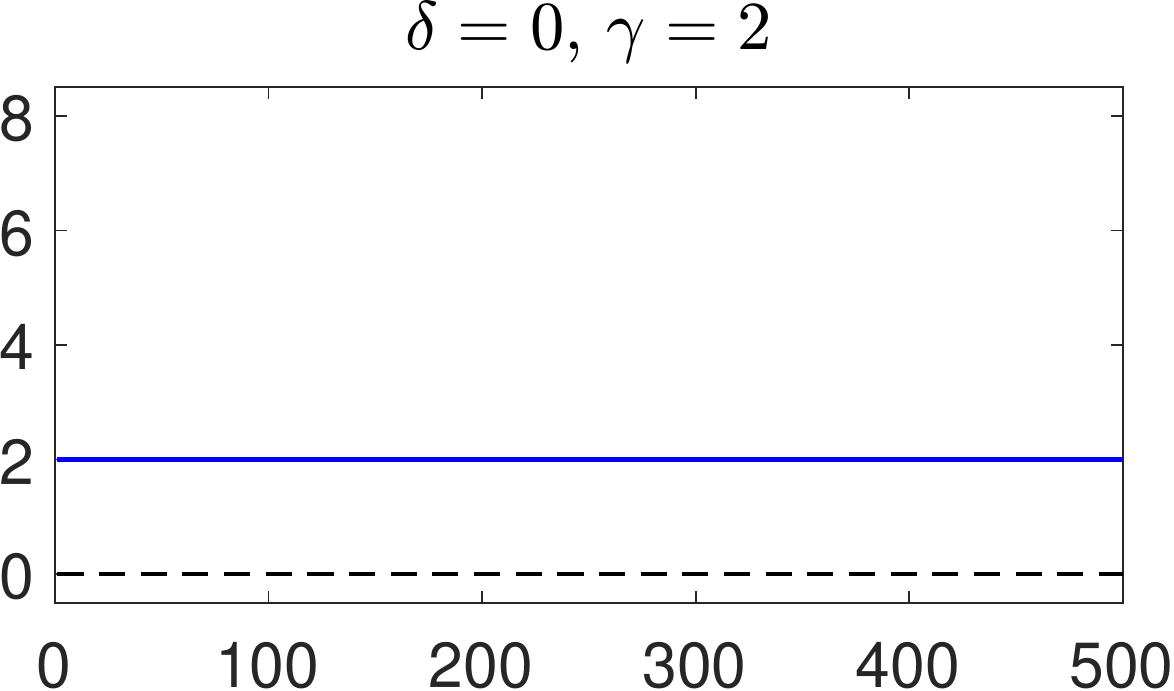} \\
		\includegraphics[width=3cm]{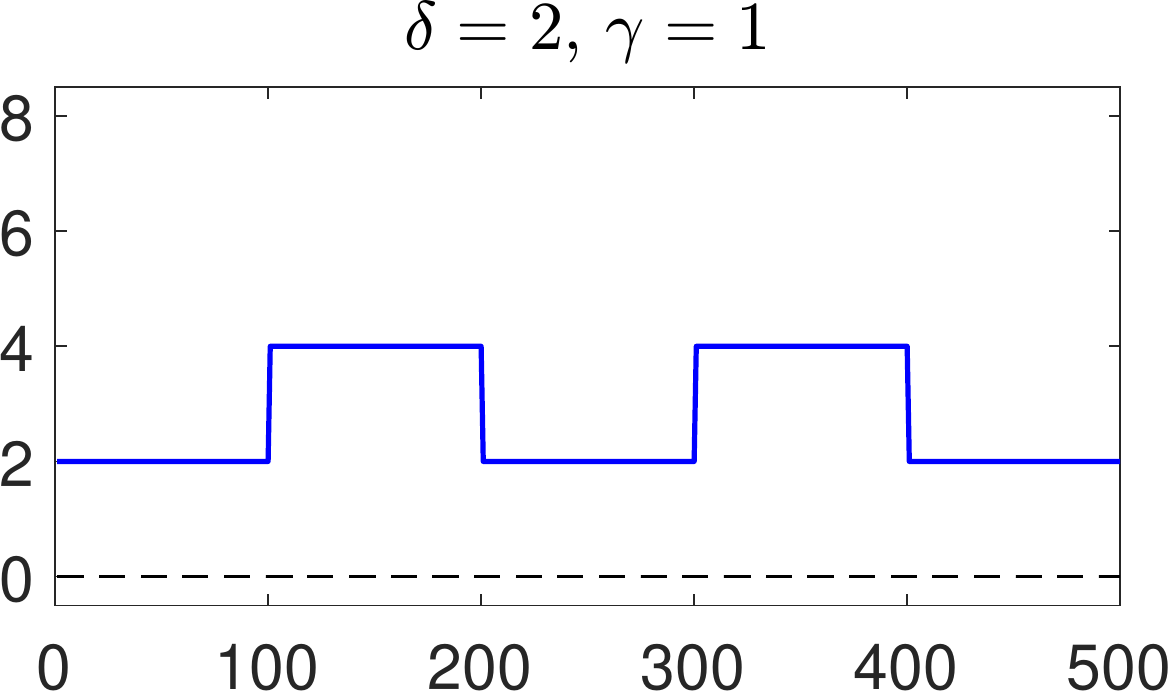} & \includegraphics[width=3cm]{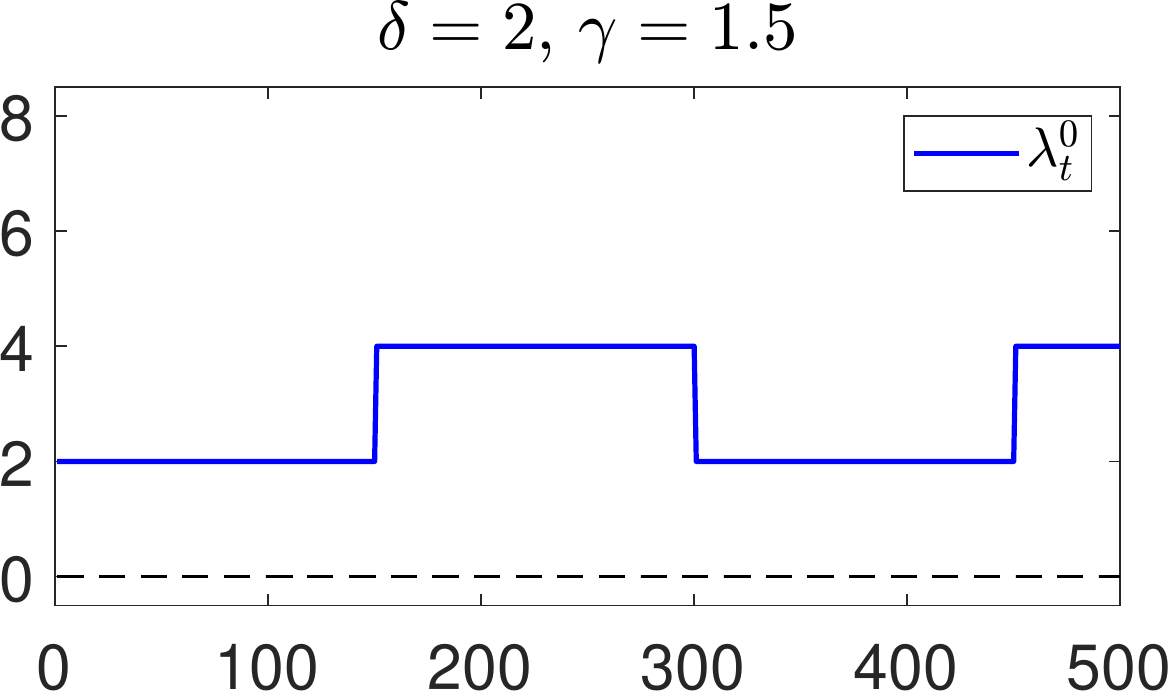} & \includegraphics[width=3cm]{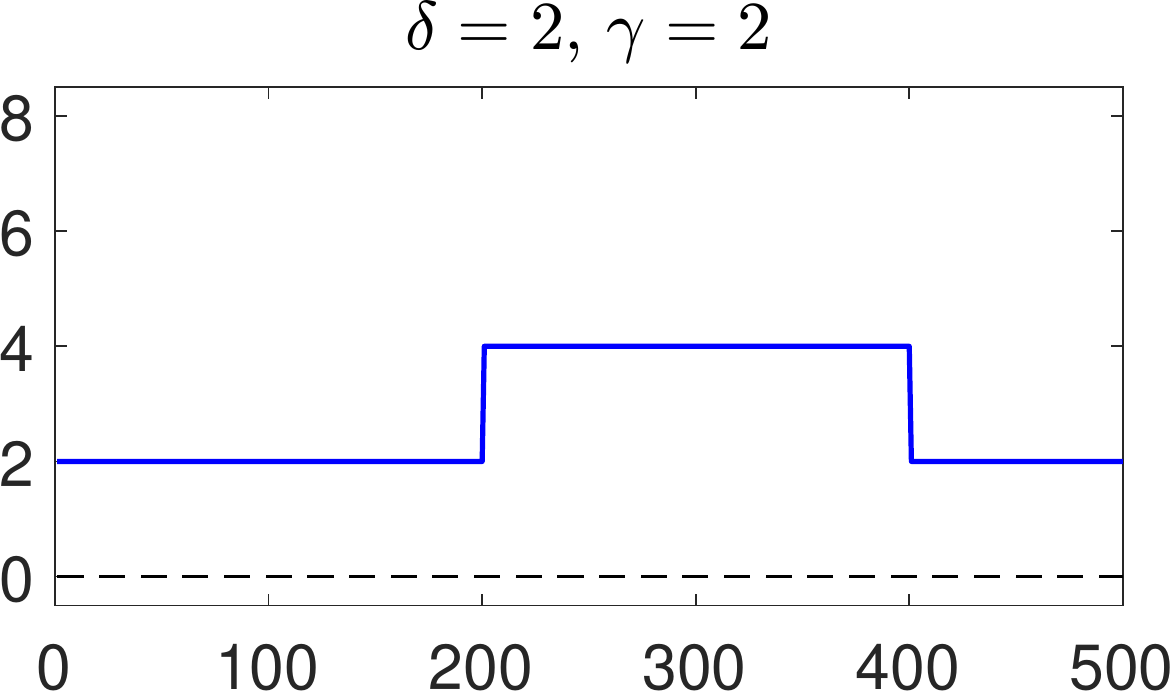} \\
			\includegraphics[width=3cm]{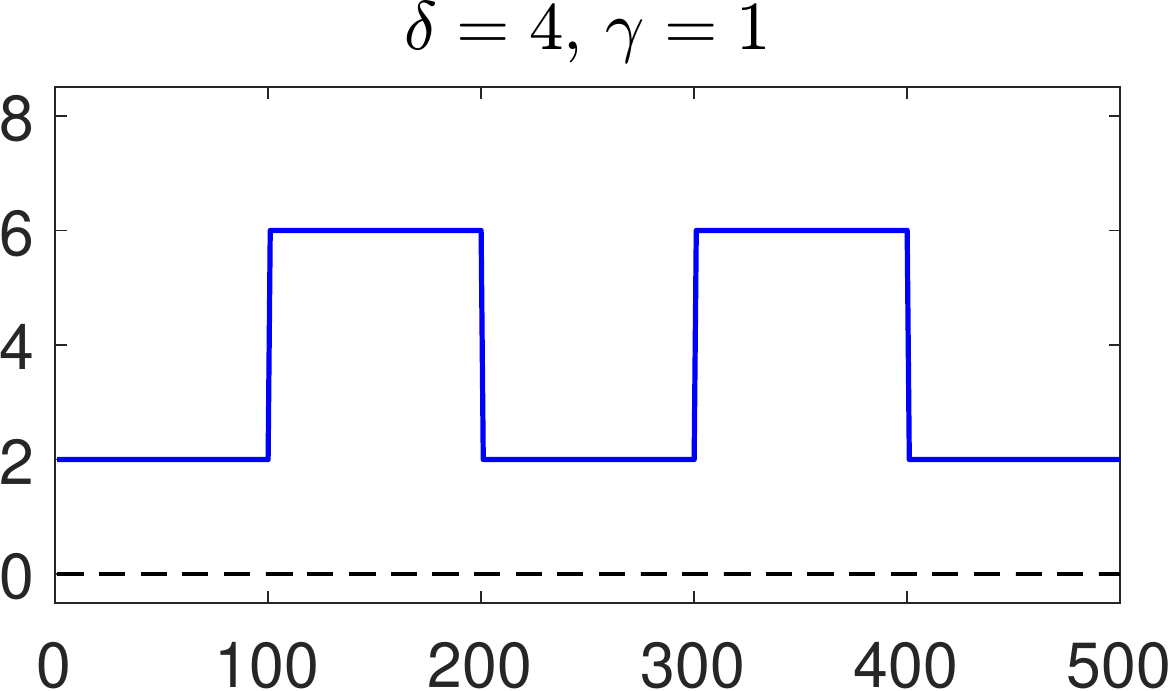} & \includegraphics[width=3cm]{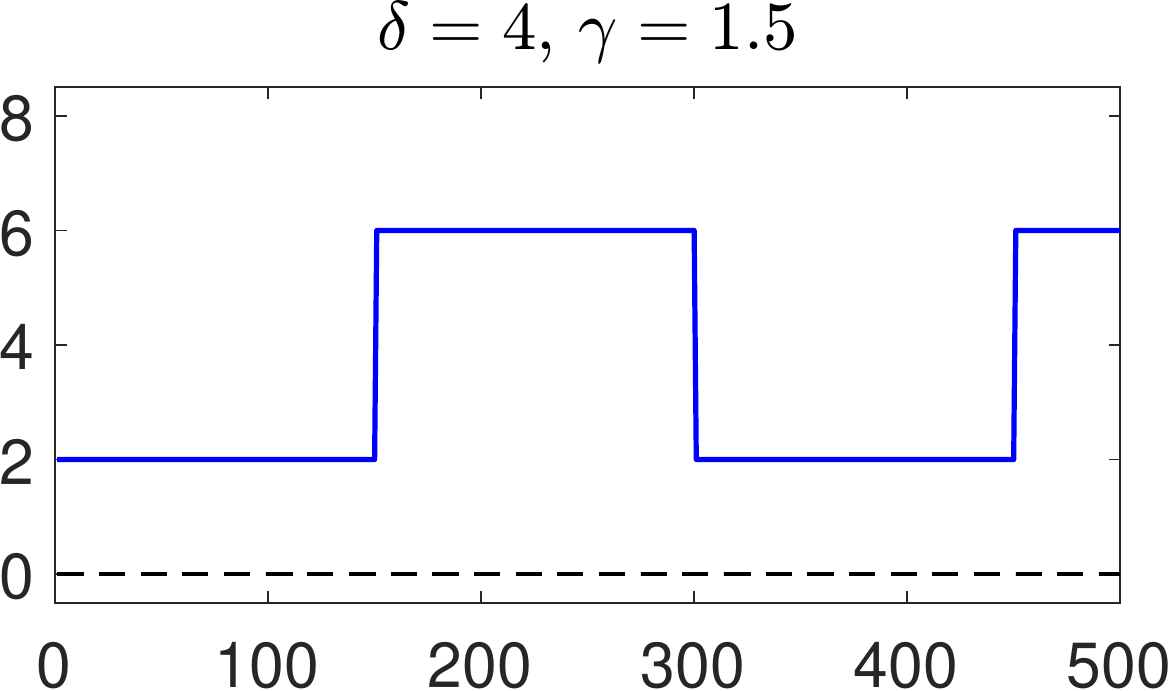} & \includegraphics[width=3cm]{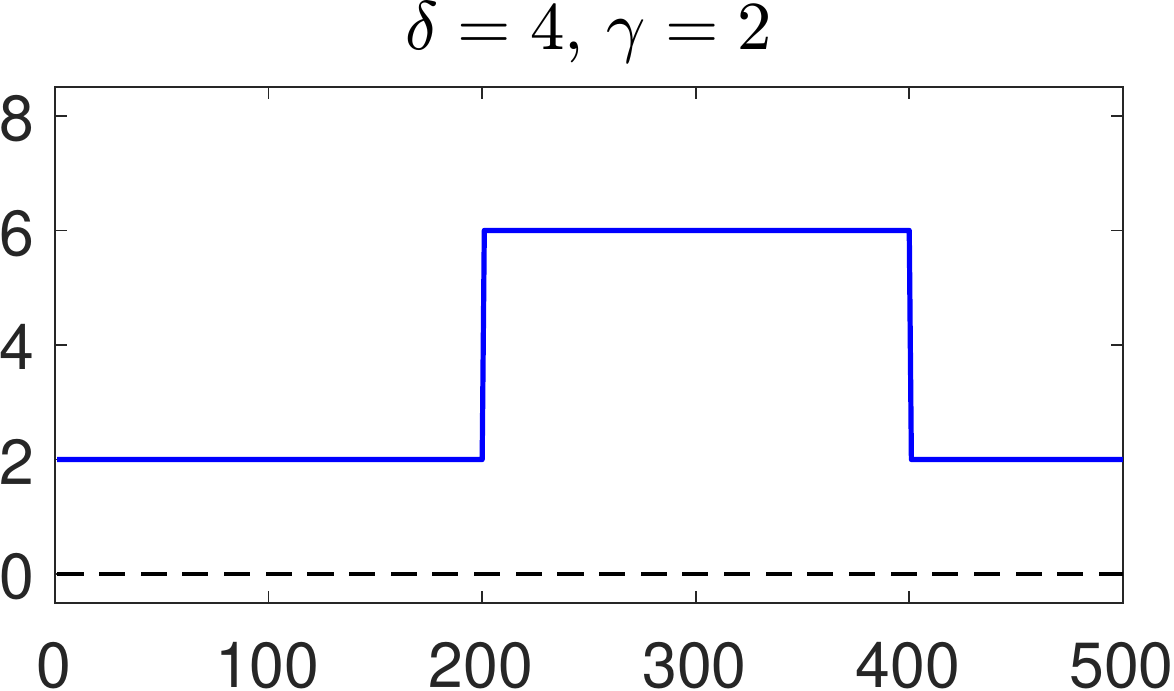} \\
				  \multicolumn{3}{c}{$T=1000$} \\
	\includegraphics[width=3cm]{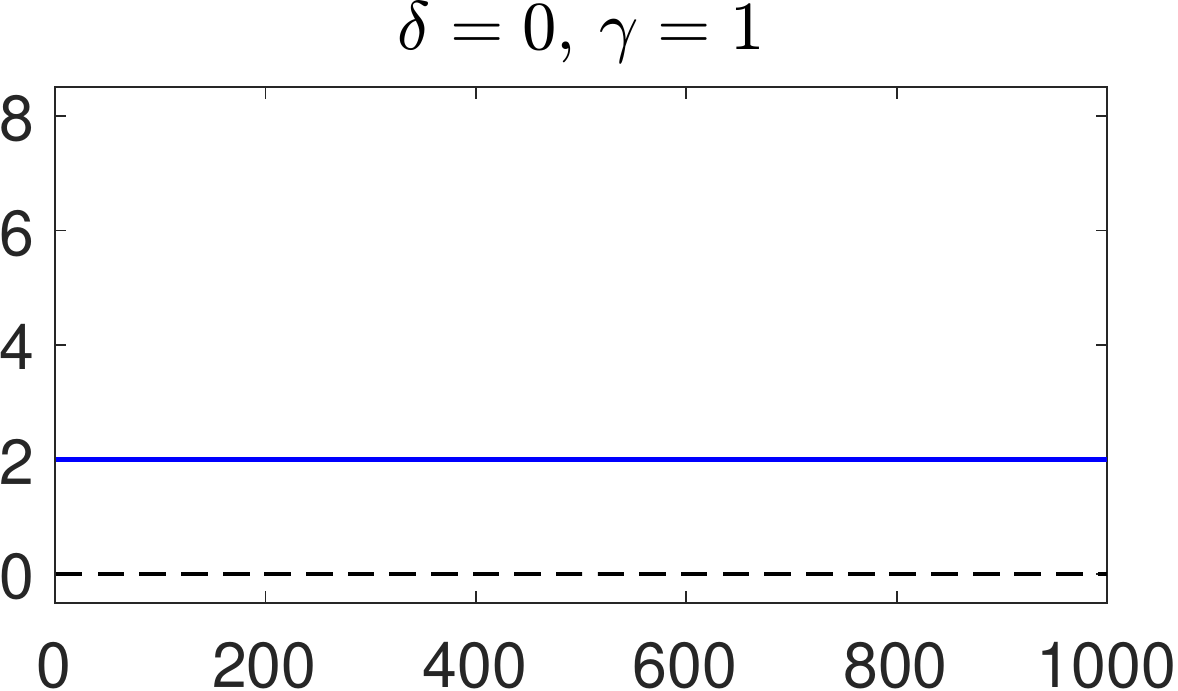} & \includegraphics[width=3cm]{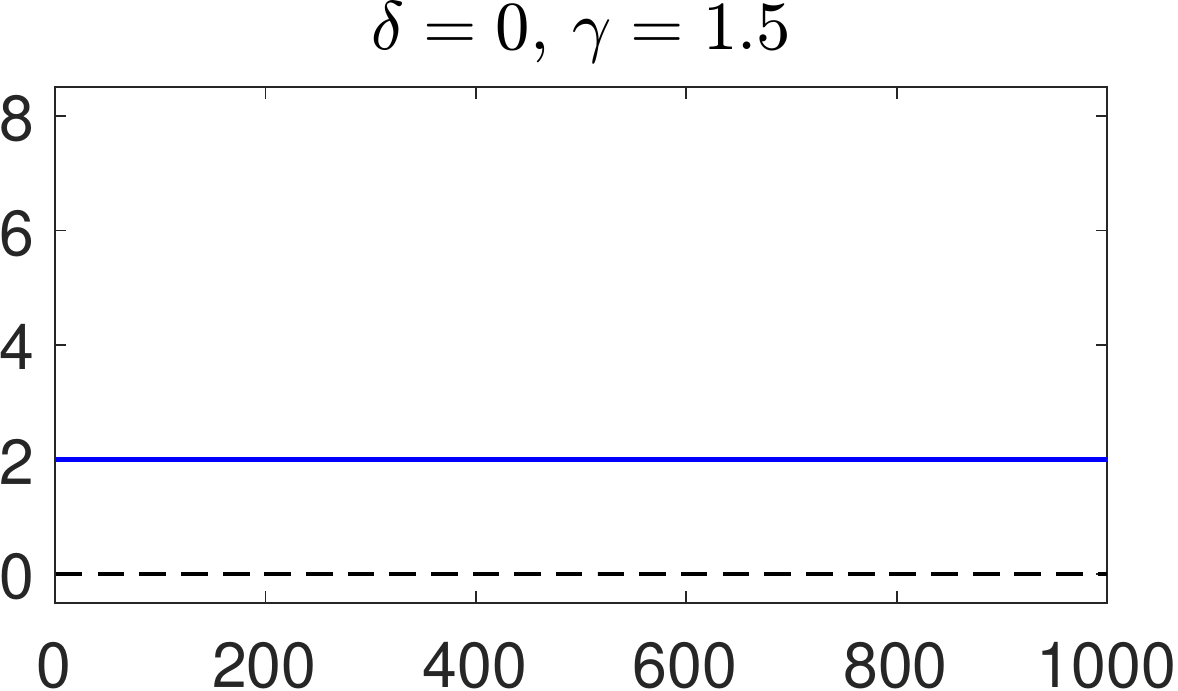} & \includegraphics[width=3cm]{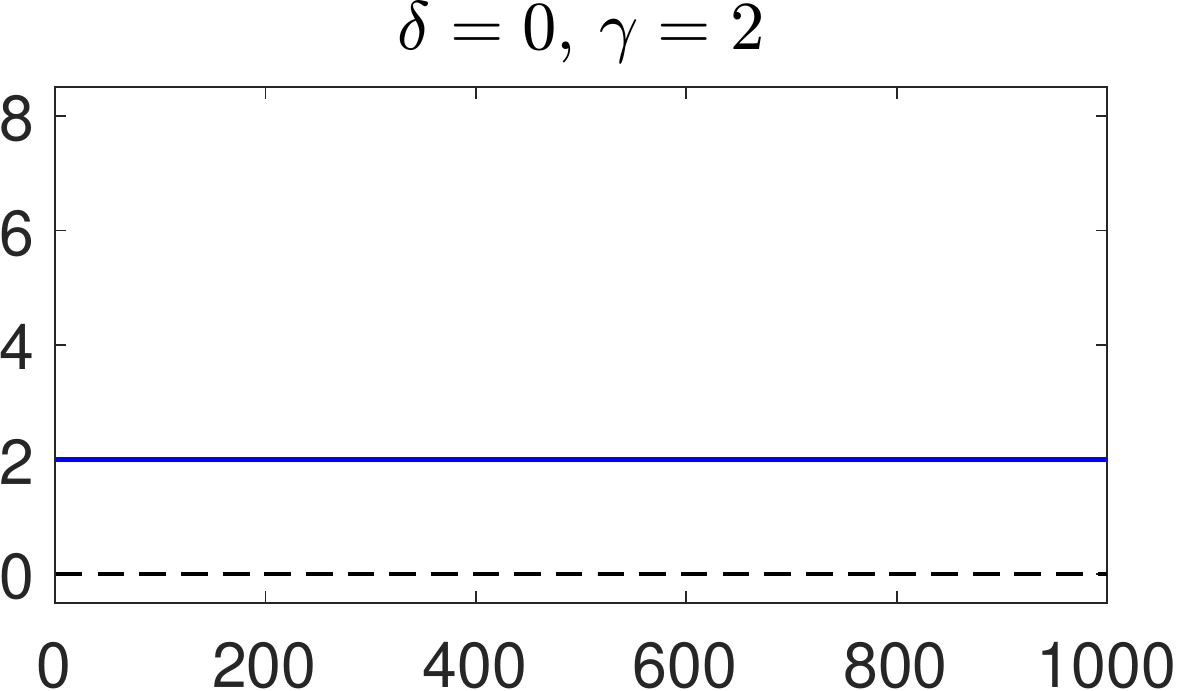} \\
		\includegraphics[width=3cm]{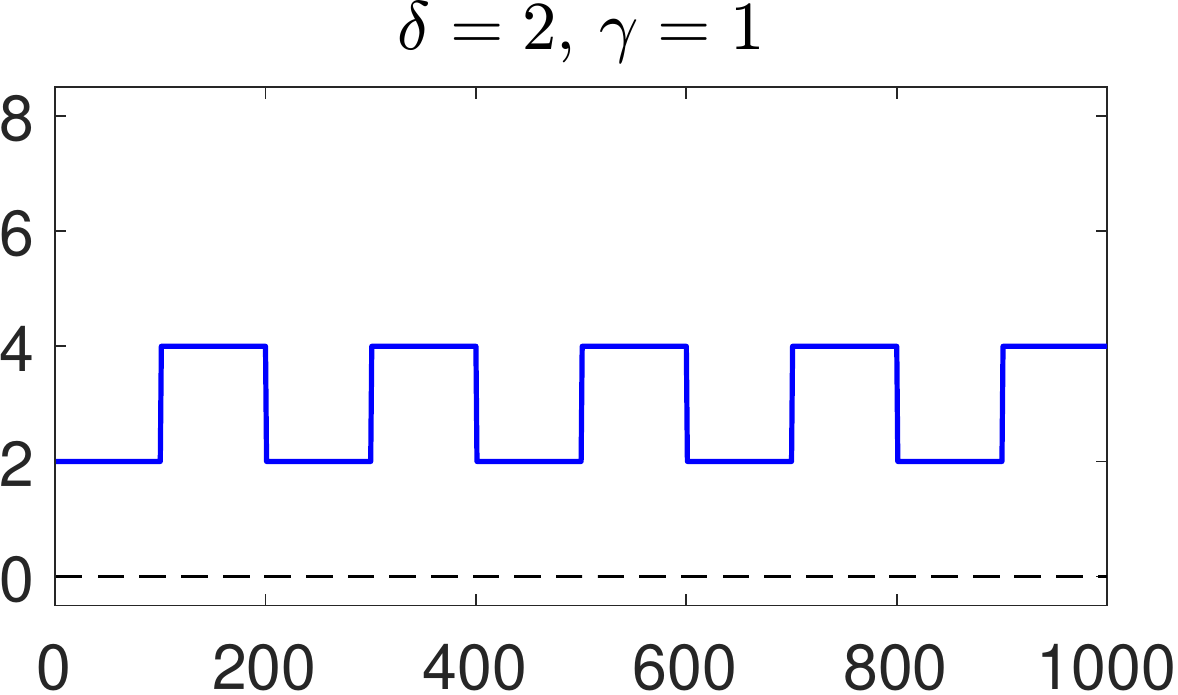} & \includegraphics[width=3cm]{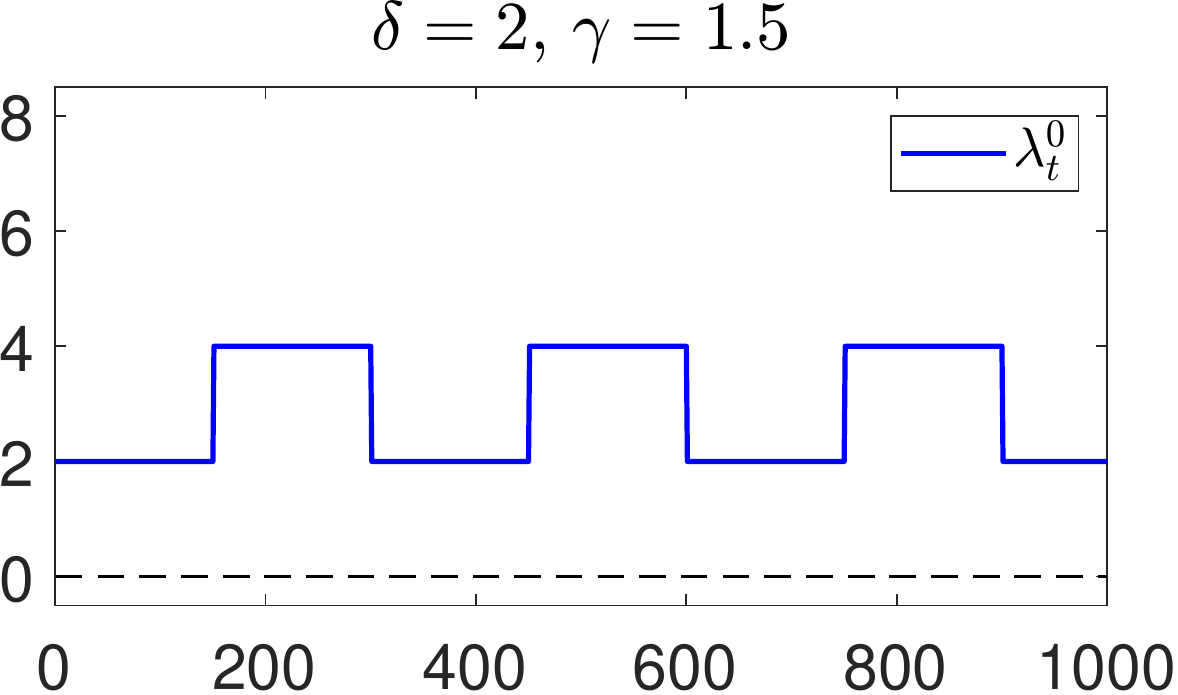} & \includegraphics[width=3cm]{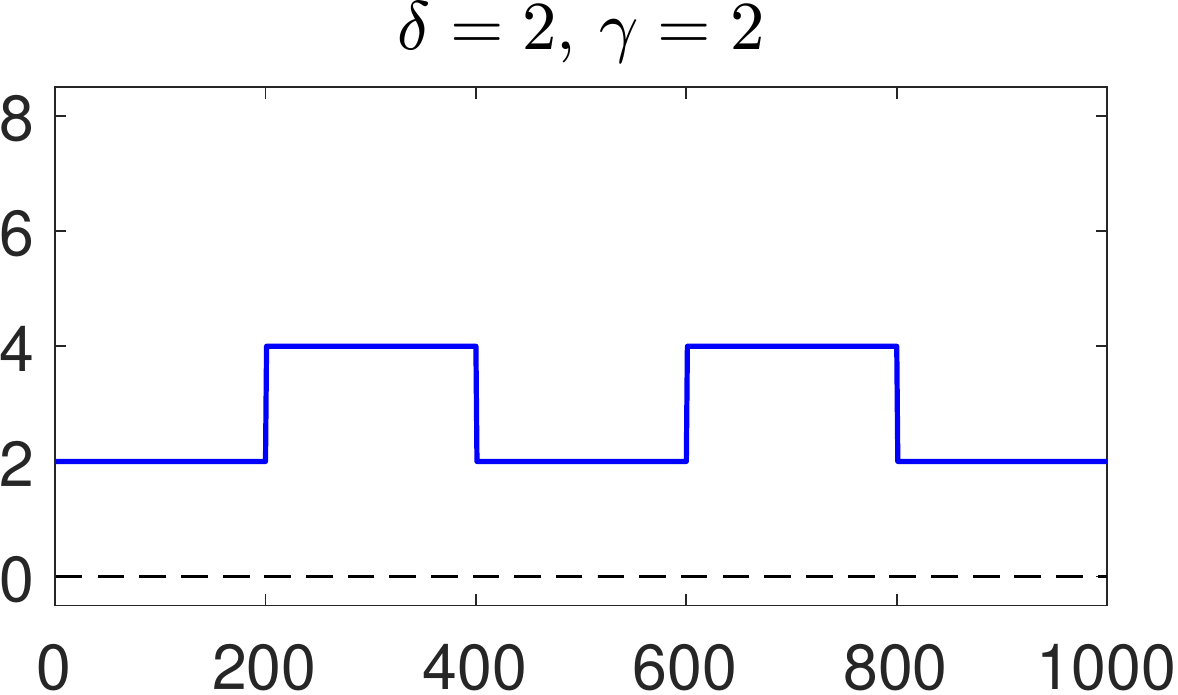} \\
			\includegraphics[width=3cm]{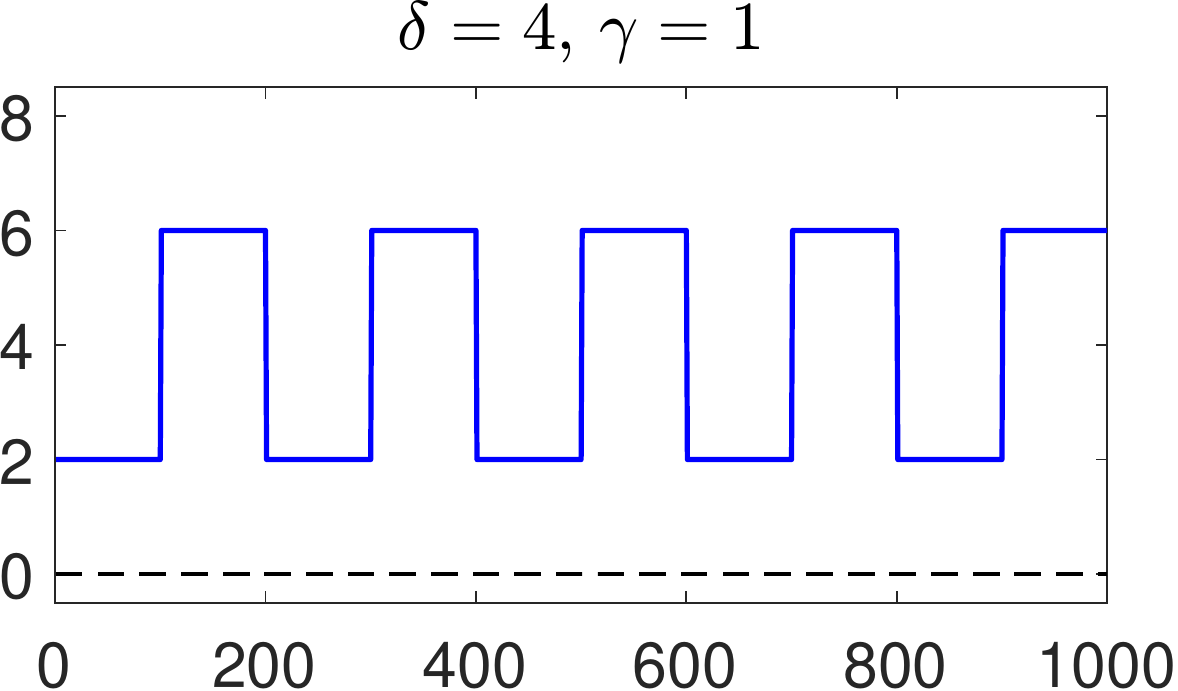} & \includegraphics[width=3cm]{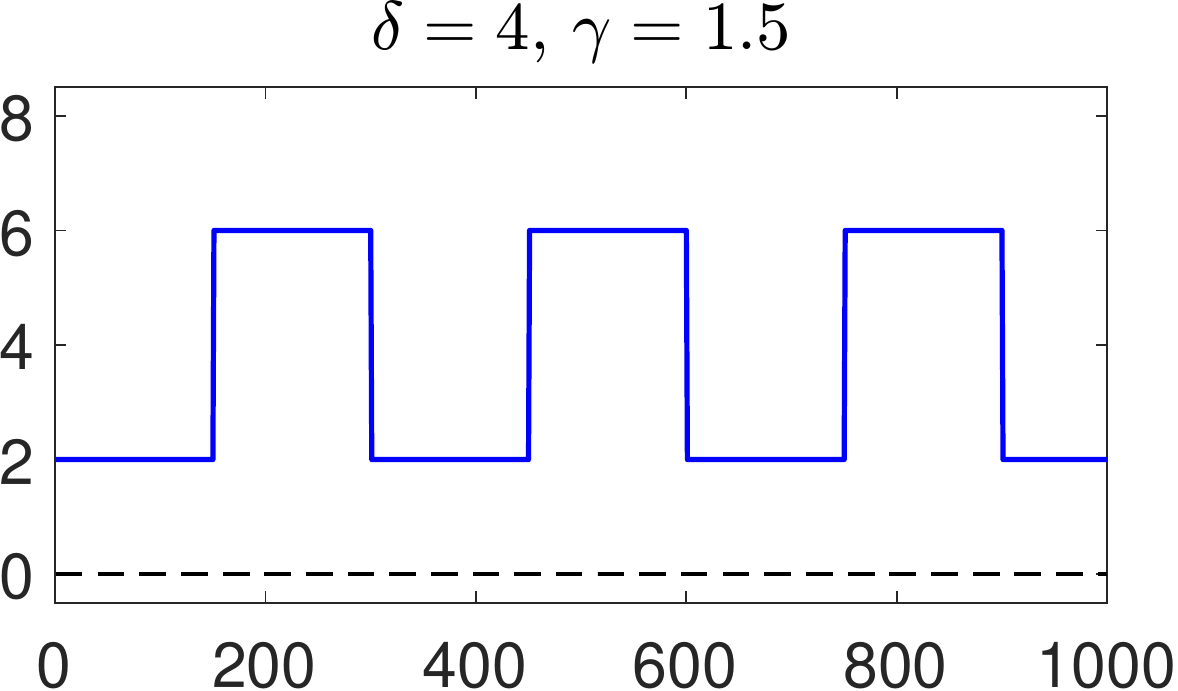} & \includegraphics[width=3cm]{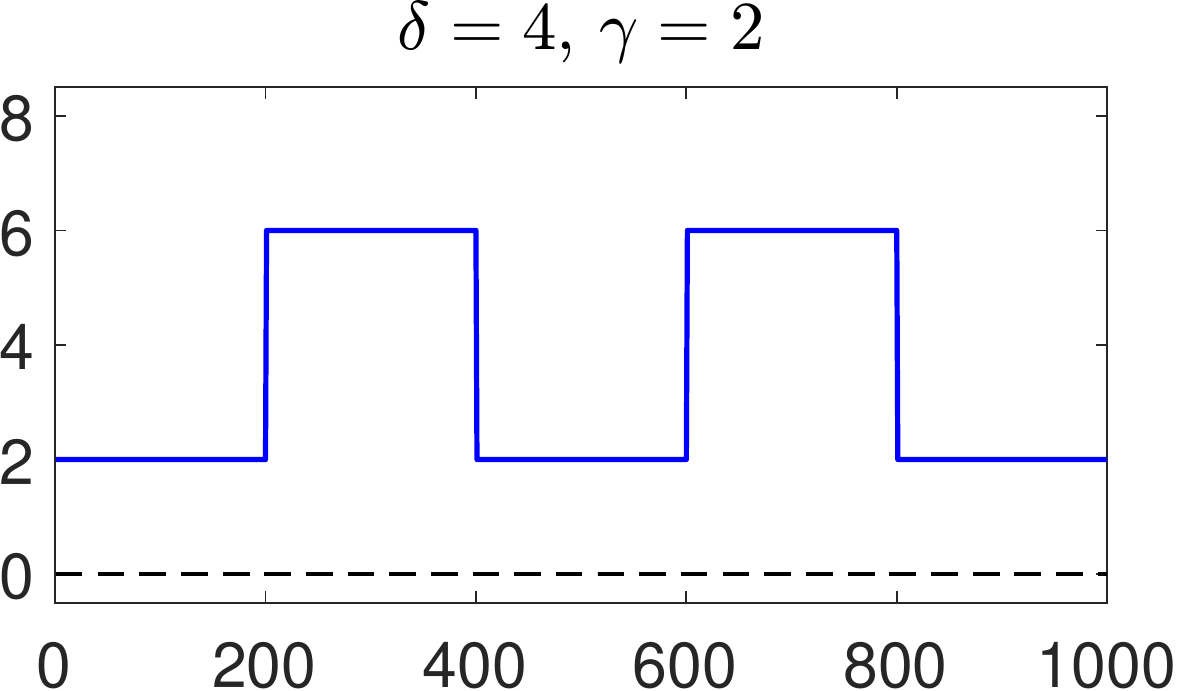} \\
\end{tabular} 
		\caption{Dynamic of $\lambda_t^0$ for $\delta \in \{0,2,4 \}$ and $\gamma \in \{1, 1.5, 2\}$. \label{fig:lambda_0}}
\end{figure}

\newpage

\begin{figure}[h]
	\centering
	 \begin{tabular}{c}
	  \multicolumn{1}{c}{$T=250$} \\
\includegraphics[width=10cm]{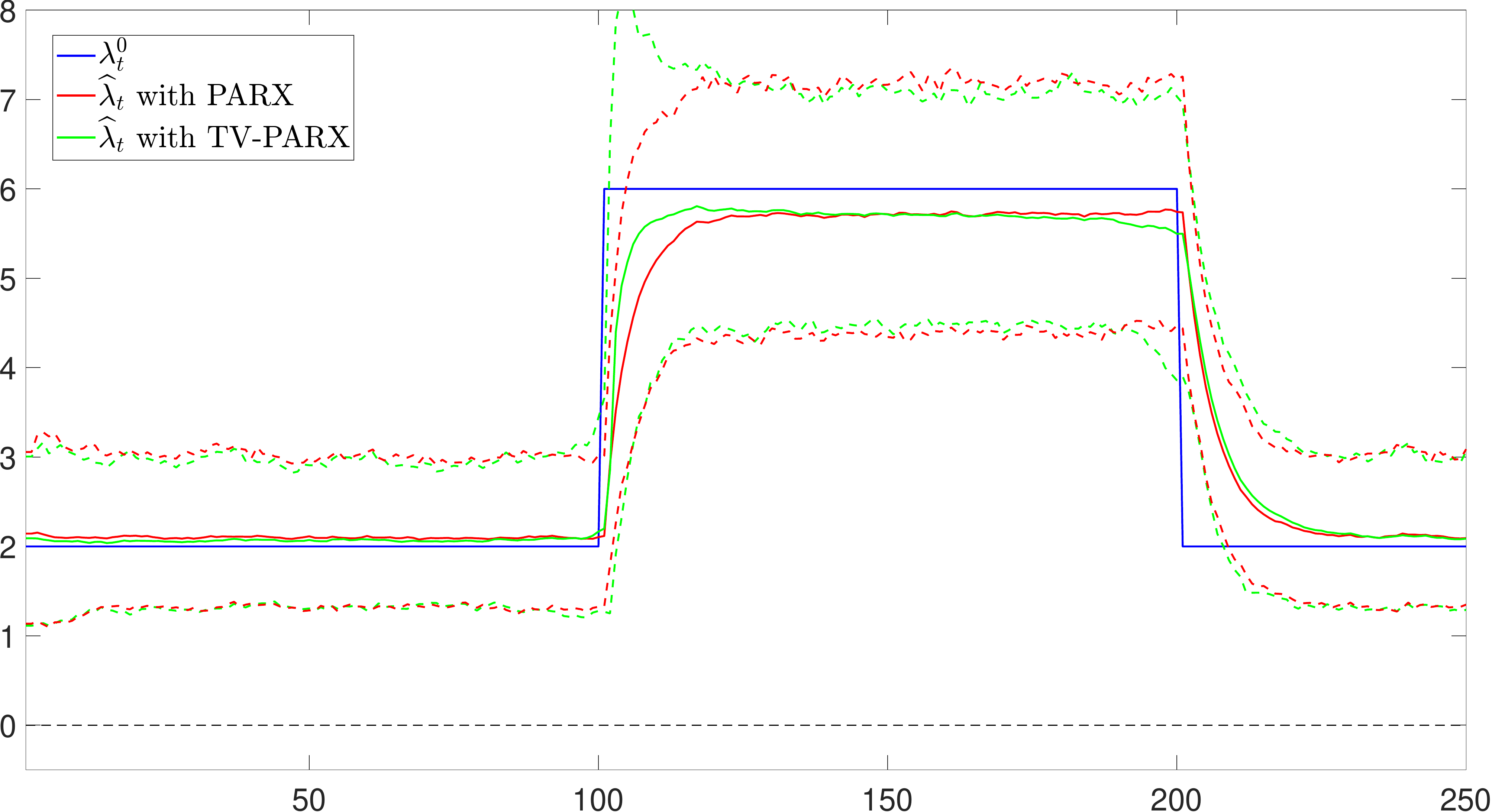} \\
	  \multicolumn{1}{c}{$T=500$} \\
\includegraphics[width=10cm]{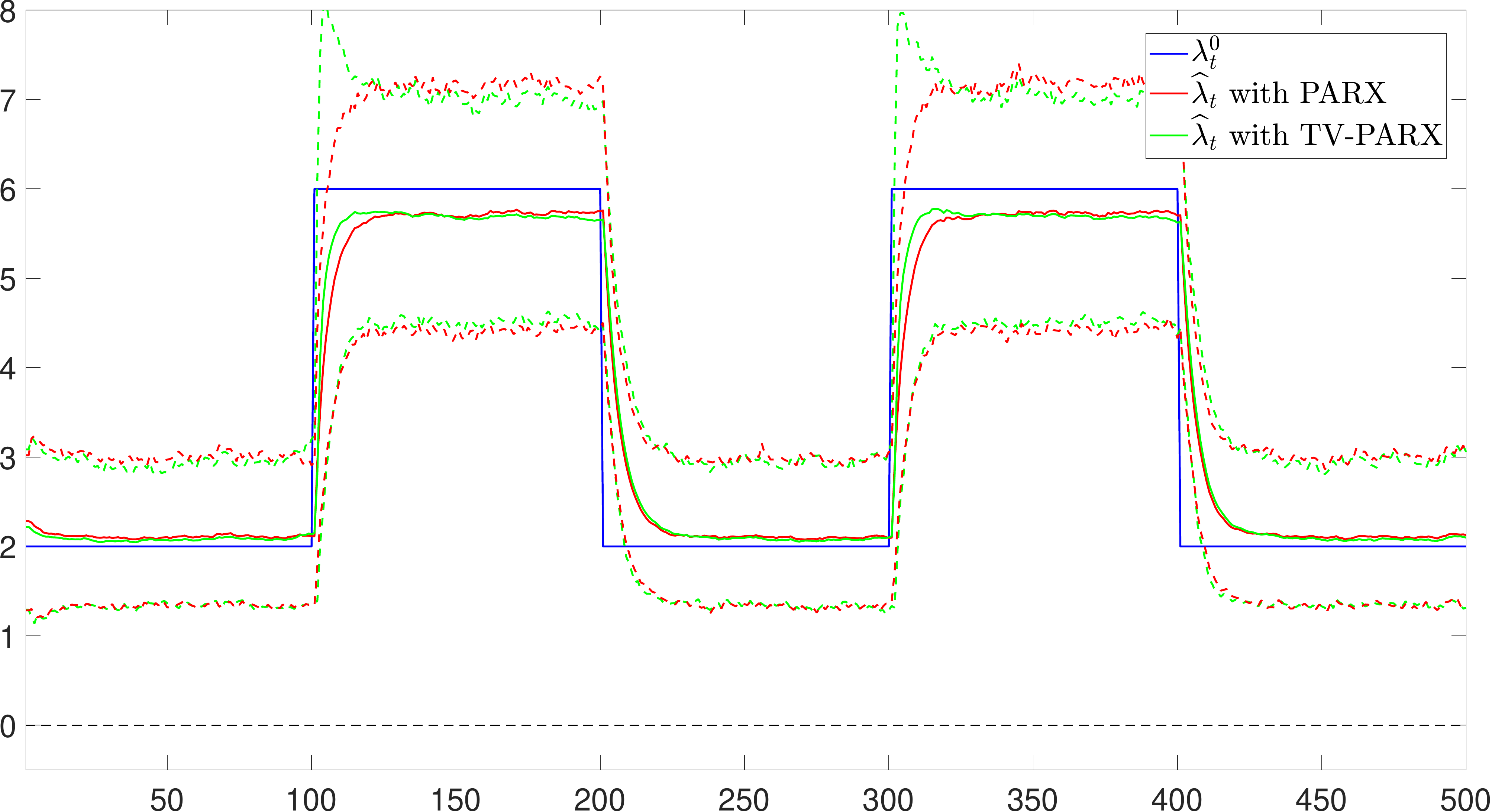} \\
	  \multicolumn{1}{c}{$T=1000$} \\
\includegraphics[width=10cm]{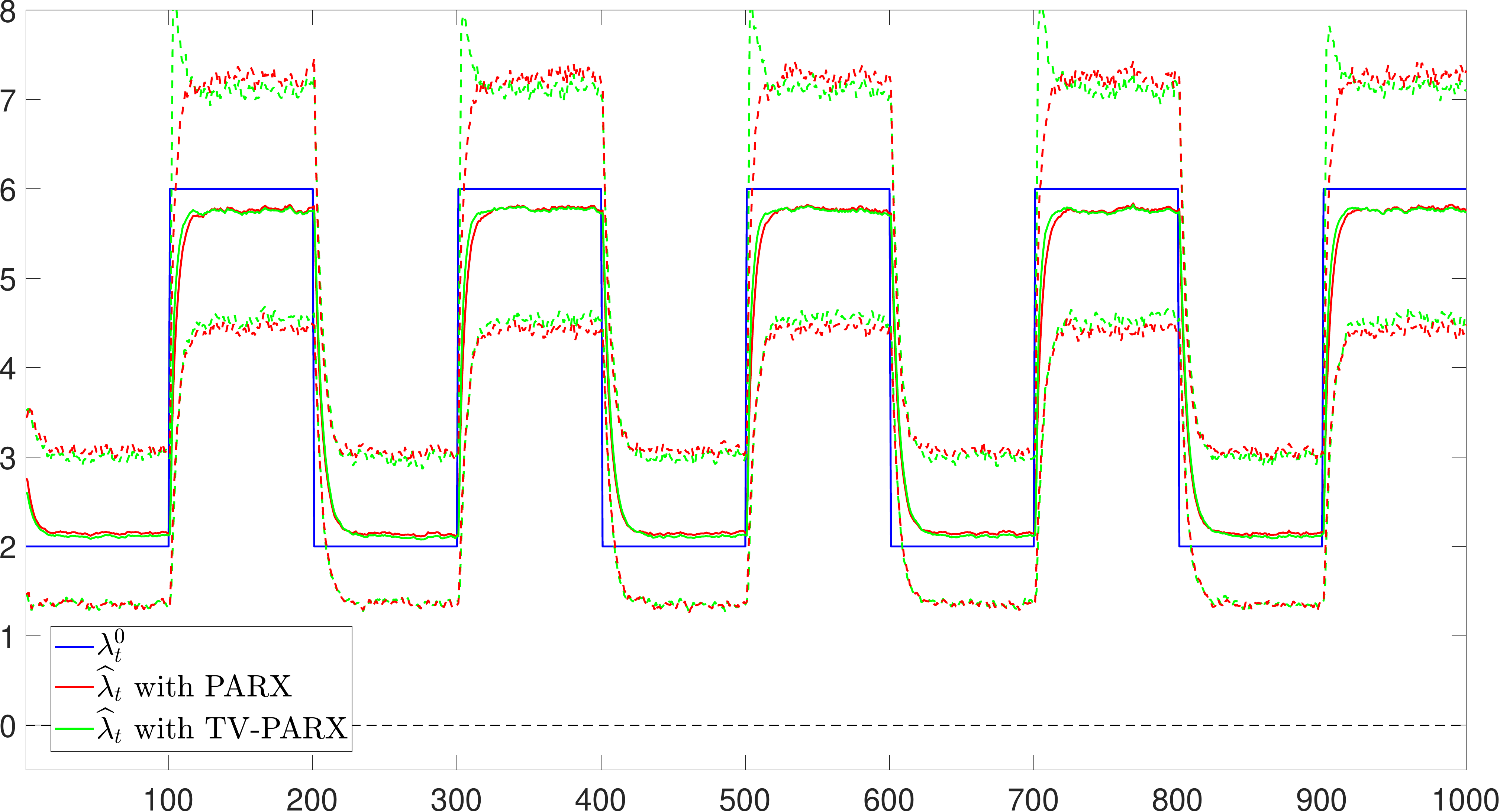} \\
\end{tabular} 
\caption{Estimation of $\lambda_t$ for $\delta = 4 $ and $\gamma = 1$ with the 95 \% confidence bands obtained with the TV-PARX (in green) and the static PARX (in red). \label{fig:mc_est}}
\end{figure}


\newpage

\begin{figure}[h]
\includegraphics[width=16cm]{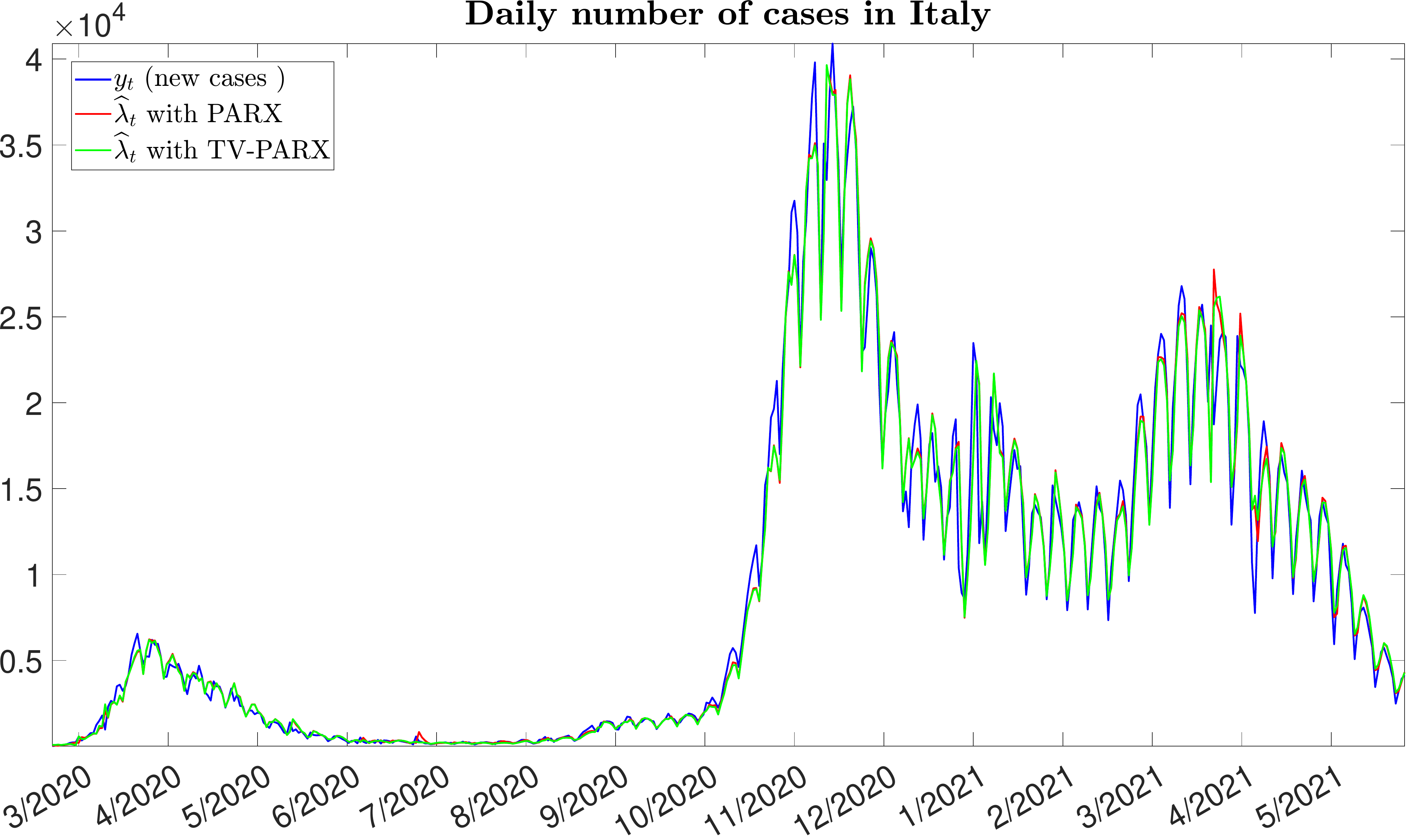}
\caption{Fitted vs actual daily number of cases of Covid-19 in Italy. \label{fig:emp_fitted_actual}}
\end{figure}

\newpage

\begin{figure}[h]
\includegraphics[width=16cm]{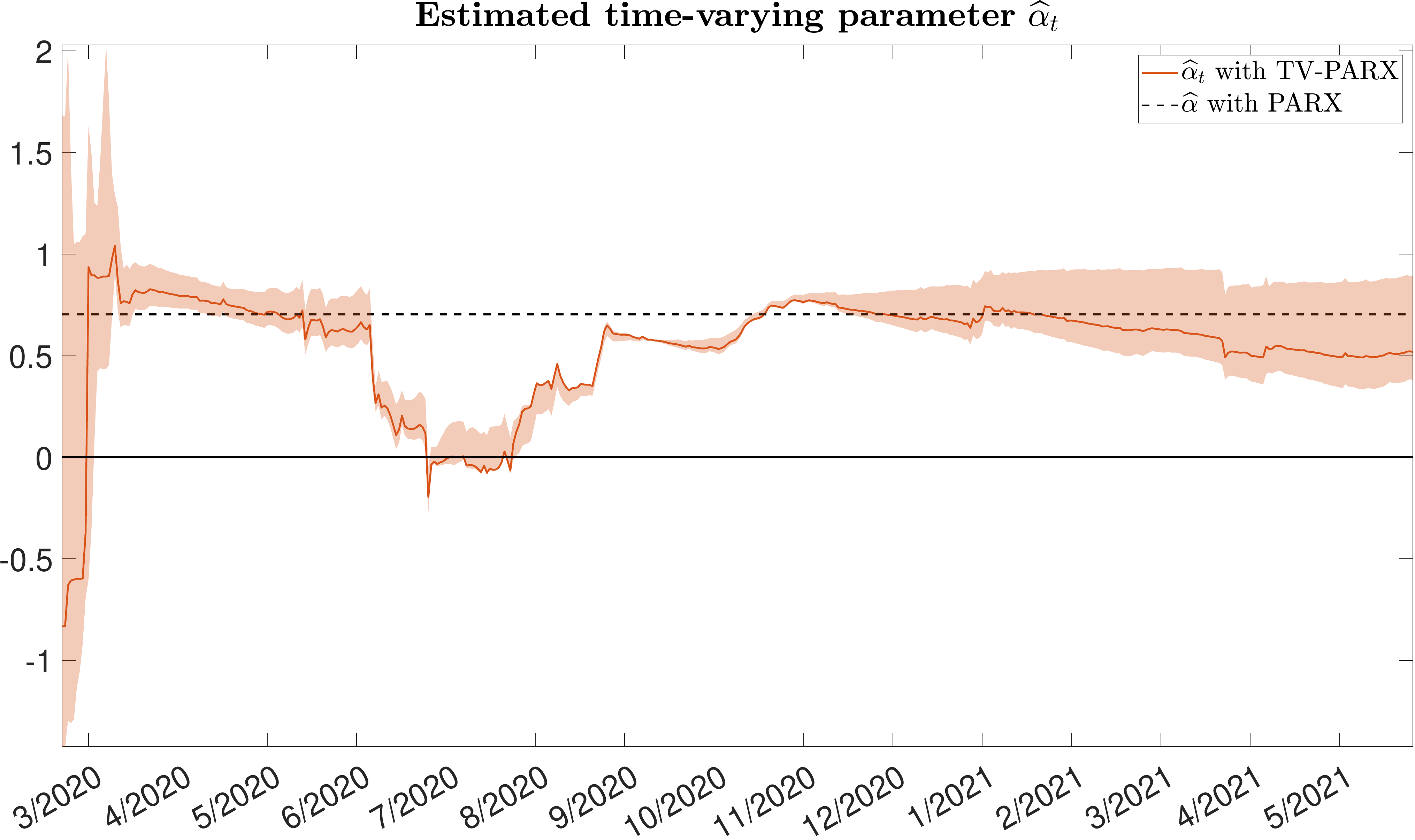}
\caption{Estimated time-varying parameter $\alpha_t$ in \eqref{covid} for the daily number of cases of Covid-19 in Italy. Shaded areas are the 95\% confidence intervals. \label{fig:emp_alpha}}
\end{figure}

\newpage

\begin{figure}[h]
\includegraphics[width=16cm]{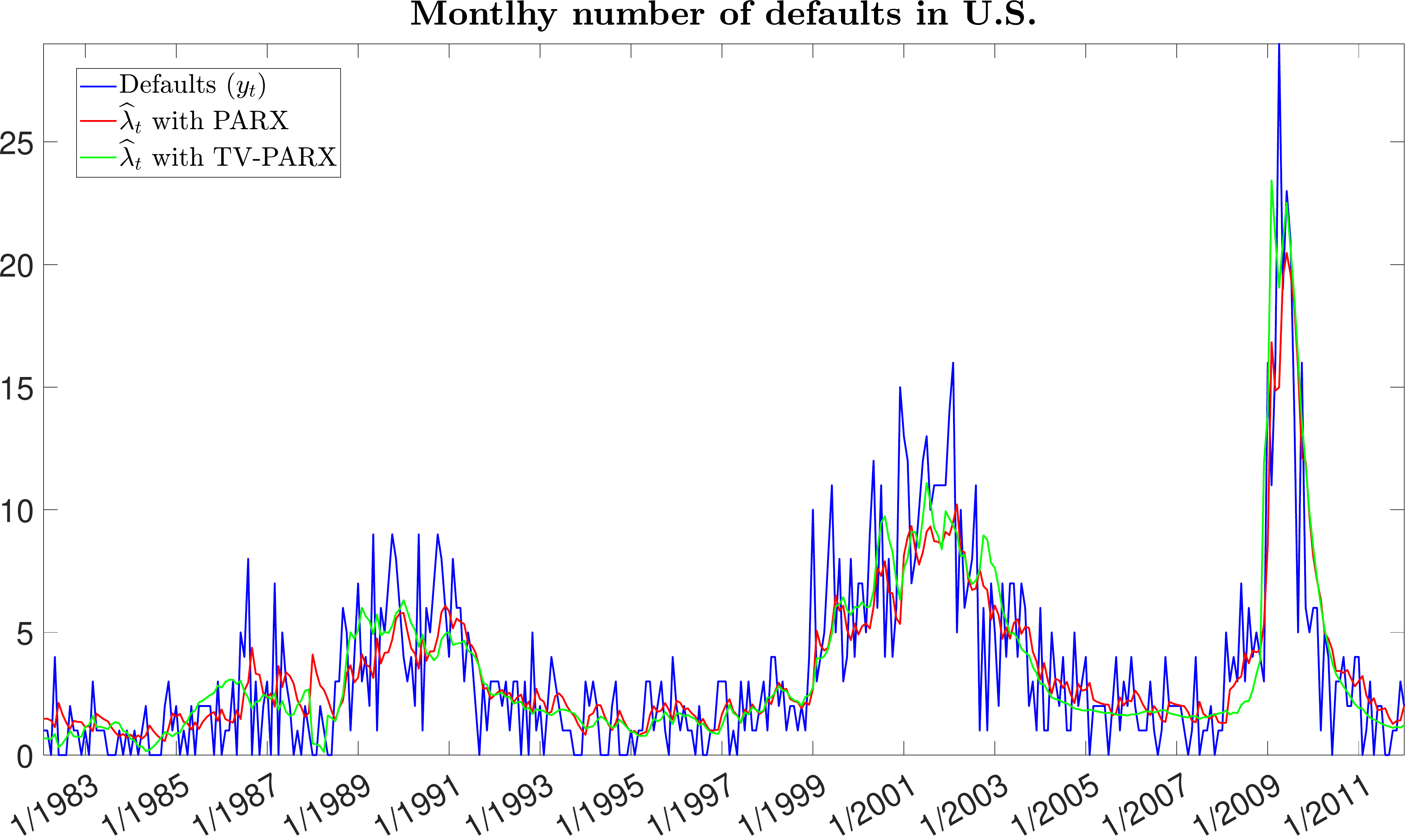}
\caption{Fitted vs actual monthly number of corporate defaults in US. \label{fig:emp_fitted_actual_defaults}}
\end{figure}

\newpage

\begin{figure}[h]
\begin{center}
\includegraphics[width=10cm]{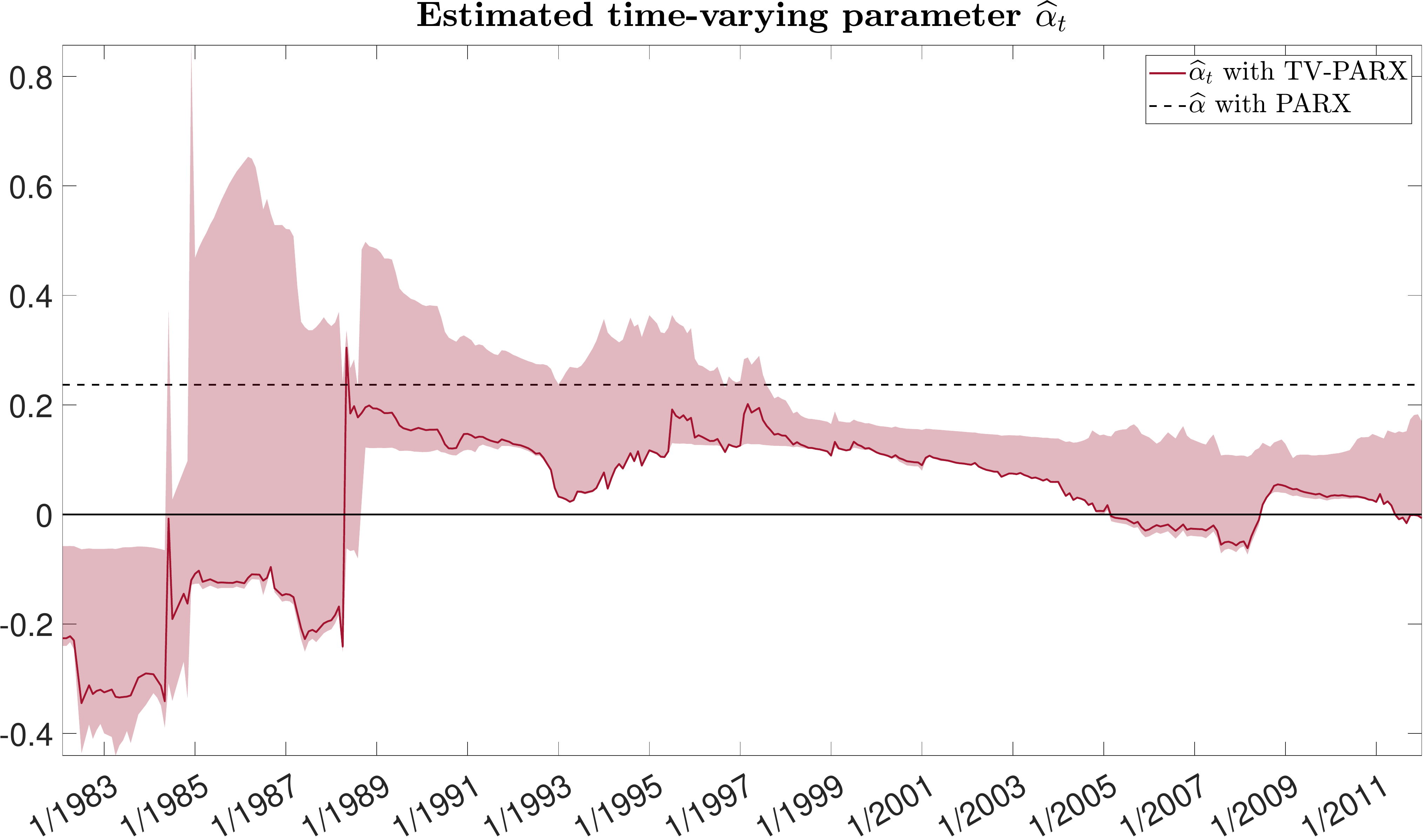} \\
\includegraphics[width=10cm]{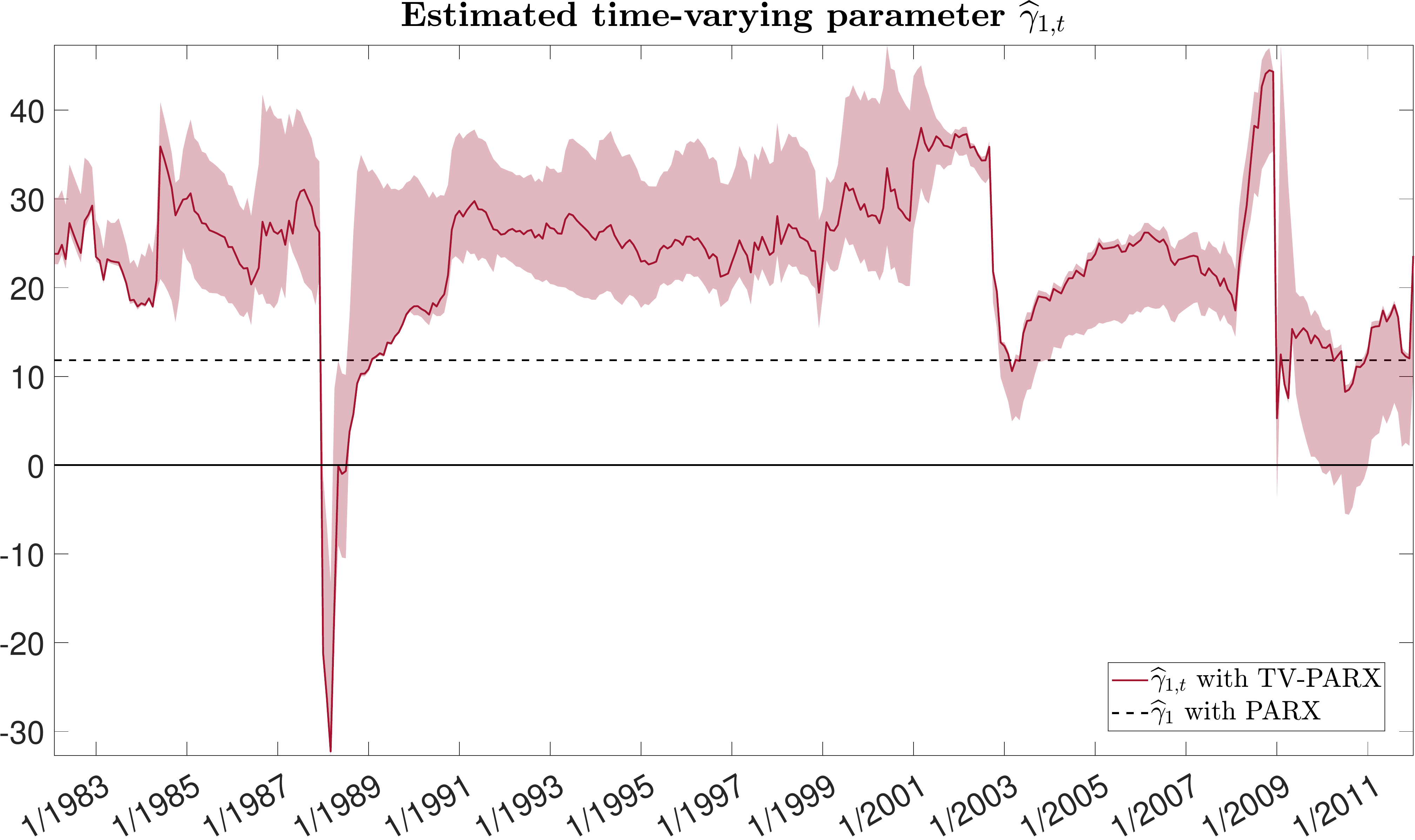} \\
\includegraphics[width=10cm]{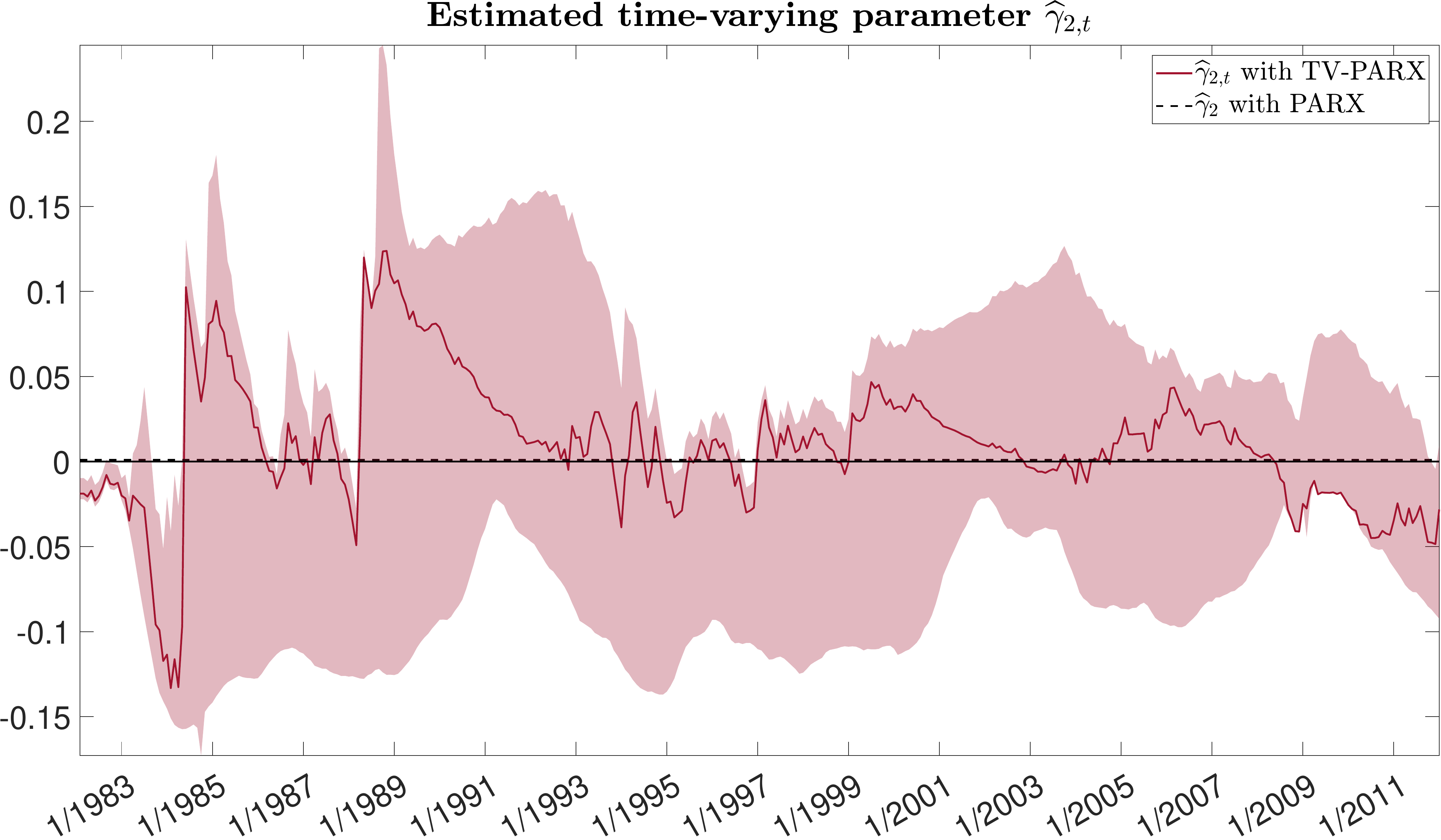} \\
\caption{Estimated time-varying parameter $\alpha_t$, $\gamma_{1,t}$ and $\gamma_{2,t}$ in \eqref{eq:emp_default} for the monthly number of corporate defaults in US. Shaded areas are the 95\% confidence intervals. \label{fig:emp_alphagamma_defaults}}
\end{center}
\end{figure}

%
%
%


\clearpage
\appendix

\renewcommand{\thelemma}{A.\arabic{lemma}}

\section{Proofs}
\label{proofs}

\subsection*{Proof of Proposition \ref{prop_ergodic_lnLambda}}
Following \cite{davis2003}, we define $e_t := (y_t - \lambda_t)\lambda^{-1}_{t-1}$. Under the initial conditions that $e_s = 0$ and $y_s = 0$ for $s\leq 0$, it is shown in \cite{davis2003} that $\{ e_t \}_{t\in \mathbb{Z}}$ is a martingale difference sequence, so that $\mathbb{E}[e_t | \mathcal{F}_{t-1}]=0$, and $\mathbb{E}[e_t^2] = \mathbb{E}[\lambda_t^{-1}]$. Therefore, as long as $|\phi_\alpha|<1$, by unfolding the recursion in \eqref{eq:tv_alpha_parx} we obtain that
\begin{align}
\label{alpha_unfolded}
\alpha_{t+1} = \frac{\delta_\alpha}{1 - \phi_\alpha} 
+ \sum_{i=0}^\infty \kappa_\alpha \phi_\alpha^i e_{t-i} e_{t-i-1}.
\end{align}
Thus, it is straightforward to see that $\bar{\alpha} = \mathbb{E}[\alpha_t] = \delta_\alpha / (1-\phi_\alpha)$. Now, we can focus on the recursion in \eqref{ln_tv_lambda} and to prove stability and geometric ergodicity we use the Meyn-Tweedie's criterion with the test function $V(x) = |x|$, see \cite{Meyn2009}, and we substitute $\alpha_t$ with $\bar{\alpha}$ in \eqref{ln_tv_lambda}, in order to distinguish four different cases:

\emph{(i)} $\bar{\alpha}>0$. We have
\begin{align*}
\mathbb{E}[|\log \lambda_{t+1}| \, | \log \lambda_{t} = \log \lambda]
=
\mathbb{E}[|\omega + \beta \log \lambda + \bar{\alpha} (y_t - \lambda)\lambda^{-1}| 
\, | \log \lambda_{t} = \log \lambda].
\end{align*}
If $\log \lambda_{t} = \log \lambda > 0$ and $\log \lambda\rightarrow \infty$, then $\lambda\rightarrow \infty$ and note that $\lim_{\lambda\rightarrow \infty} (y_t - \lambda)\lambda^{-1} = \lim_{\lambda\rightarrow \infty} (N_t(\lambda)/\lambda - 1) = 0$ almost surely, by Lemma A.1 of \cite{Wang2014}. Thus
\begin{align*}
\mathbb{E}[|\log \lambda_{t+1}| \, | \log \lambda_{t} = \log \lambda]
\approx \beta \log \lambda,
\end{align*}
which implies that the condition $0<\beta<1$ ensures that the process $\{ \log \lambda_{t} \}_{t\in\mathbb{Z}}$ is geometrically ergodic. If $\log \lambda_{t} = \log \lambda < 0$ and $|\log \lambda| \rightarrow \infty$ we take the process in \eqref{ln_tv_lambda} two steps forward and write 
\begin{align}
\label{two_steps_lnlambda}
\log \lambda_{t+1}
=
\omega + \beta \omega + \beta [ \beta \log \lambda_{t-1} + \bar{\alpha} (y_{t-1} - \lambda_{t-1})\lambda_{t-1}^{-1} ]  + \bar{\alpha} (y_{t} - \lambda_{t})\lambda_{t}^{-1}.
\end{align}
Now, the fact that $\log \lambda_{t} = \log \lambda < 0$ with $|\log \lambda| \rightarrow \infty$ implies that $\lambda_{t} = \lambda$ becomes very small, so that $(y_{t} - \lambda)\lambda^{-1} \rightarrow 0$, because $y_{t}$ is Poisson distributed with intensity $\lambda$. Moreover, as $\log \lambda_{t} = \log \lambda < 0$
\begin{align*}
\log \lambda_{t}
=
\omega +
\beta \log \lambda_{t-1} + \bar{\alpha} (y_{t-1} - \lambda_{t-1})\lambda_{t-1}^{-1},
\end{align*}
we also obtain that $\log \lambda_{t-1}$ must be negative. Accordingly, $(y_{t-1} - \lambda_{t-1})\lambda_{t-1}^{-1}$ could be made arbitrarily close to zero by simply choosing $\log \lambda_{t-1}$ negative and small enough. Thus, recalling that $e_{t-1} = (y_{t-1} - \lambda_{t-1})\lambda_{t-1}^{-1}$, the following Taylor expansion is valid
\begin{align*}
e_{t-1} =
 \log 1 + \frac{1}{e_{t-1} + 1}\bigg|_{e_{t-1} = 0} e_{t-1}
 \approx
 \log({e}_{t-1} +1) 
 \approx
 \log\bigg( \frac{\max(y_{t-1},c)}{\lambda_{t-1}}\bigg),
\end{align*}
for some $c \in (0,1]$. Hence, we may approximate \eqref{two_steps_lnlambda} with 
\begin{align}
\label{two_steps_lnlambda2}
\log \lambda_{t+1}
&\approx
\omega + \beta \omega + \beta^2 \log \lambda_{t-1} 
+\beta \bar{\alpha} \log\bigg( \frac{\max(y_{t-1},c)}{\lambda_{t-1}}\bigg)  \nonumber\\
&=
\omega + \beta \omega + \beta^2 \log \lambda_{t-1} 
+\beta \bar{\alpha} \log (\max(y_{t-1},c))
-\beta \bar{\alpha} \log \lambda_{t-1}.
\end{align}
In conclusion, for $\log \lambda_{t-1}$ negative and small enough it holds that
\begin{align*}
\mathbb{E}[|\log \lambda_{t+1}| \, | \log \lambda_{t-1} = \log \lambda ]
\leq
K + \beta ( \beta + \bar{\alpha}) |\log \lambda|,
\end{align*}
and the condition $\beta ( \beta + \bar{\alpha}) < 1$ ensures the geometric ergodicity of the process $\{ \log \lambda_{t} \}_{t\in\mathbb{Z}}$.

\emph{(ii)} $\bar{\alpha}<0$. We can easily show that by following similar arguments of \emph{(i)} for $\log \lambda_{t} = \log \lambda >0$ large, then  
\begin{align*}
\mathbb{E}[|\log \lambda_{t+1}| \, | \log \lambda_{t-1} = \log \lambda ] 
\approx
\beta \log \lambda.
\end{align*}
On the other hand, if $\log \lambda_{t} = \log \lambda <0$, with $|\log \lambda|$ large, then
\begin{align*}
\mathbb{E}[|\log \lambda_{t+1}| \, | \log \lambda_{t-1} = \log \lambda ] 
\approx
\beta| \beta + \bar{\alpha}| | \log \lambda|.
\end{align*}
It remains only to show that  $\{ \lambda_t \}_{t\in\mathbb{Z}}$ is geometrically ergodic with $\mathbb{E}[\lambda_t]<\infty$. Again, we use the Meyn-Tweedie's criterion. We define a new test function $V(x) = e^x$ and follow the same arguments as before, and consider the same four different cases.

For $\log\lambda_t = \log\lambda > 0$ and $\log\lambda$ large, then 
\begin{align*}
\mathbb{E}[|\lambda_{t+1}| \, | \log\lambda_t = \log\lambda]
 \leq \exp\{ \beta \log\lambda \} 
 = \exp\{ (\beta - 1) \log\lambda \} \exp\{ \log\lambda \},
\end{align*}
which converges since $0<\beta<1$. Analogously, for $\log\lambda_t = \log\lambda < 0$ with $|\log\lambda|$ large, then
\begin{align*}
\mathbb{E}[|\lambda_{t+1}| \, | \log\lambda_t = \log\lambda] 
\leq \exp\{ ( \beta |\beta + \bar{\alpha}|-1) |\log\lambda| \}
\exp\{ |\log\lambda| \},
\end{align*}
which converges since $0<\beta|\beta + \bar{\alpha}|<1$.

Therefore, the conditions stated above, are also sufficients for the geometric ergodicity of $\{ \lambda_t \}_{t\in \mathbb{Z}}$ with $\mathbb{E}[\lambda_t]<\infty$.

\subsection*{Proof of Proposition \ref{prop_sec_mom}}
We first show that the the stationary and ergodic $\{ \alpha_t \}_{t\in\mathbb{Z}}$ has finite $k$-th moment, $\forall k>0$. From equation \eqref{alpha_unfolded} in the proof of Proposition \ref{prop_ergodic_lnLambda} we observe that
\begin{align*}
\mathbb{E}[|\alpha_{t+1}|^k]
\leq& 
\bigg|  \frac{\delta_\alpha}{1-\phi_\alpha}\bigg|^k
+
\mathbb{E}\bigg[\bigg|\sum_{i=0}^\infty \kappa_\alpha 
\phi_\alpha^i e_{t-i} e_{t-i-1} \bigg|^k \bigg]\\
\leq&
\bigg|  \frac{\delta_\alpha}{1-\phi_\alpha}\bigg|^k
+
\sum_{i=0}^\infty |\kappa_\alpha|^k 
|\phi_\alpha|^{k i} \mathbb{E}[ |e_{t-i} e_{t-i-1}|^k ].
\end{align*}
Since $|\phi_\alpha|<1$, it suffices to prove that $\mathbb{E}[ |e_{t} e_{t-1}|^k ]<\infty$.

We can then use the law of iterated expectations, in order to show that
\begin{align*}
\mathbb{E}[|e_{t} e_{t-1}|^k]
=&
\mathbb{E}[|e|^k_{t} |e_{t-1}|^k]
=
\mathbb{E}[\mathbb{E}[ |e_{t}|^k |e_{t-1}|^k \,| \lambda_{t-1} = \lambda]]
=
\mathbb{E}[ |e_{t}|^k \mathbb{E}[ |e_{t-1}|^k \,| \lambda_{t-1} = \lambda]].
\end{align*}
Now, from Lemma A.1 of \cite{Wang2014}, we know that the process $\{ e^k_{t} \}_{t\in \mathbb{Z}}$ is uniformly integrable for any integer $k$ $\geq 1$. Therefore, $\exists C_1, C_2 > 0$ such that
\begin{align*}
\mathbb{E}[ |e_{t}|^k \mathbb{E}[ |e_{t-1}|^k \,| \lambda_{t-1} = \lambda]]
\leq 
C_1 \, \mathbb{E}[ |e_{t}|^k ]
=
 C_1 \, \mathbb{E}[ \mathbb{E}[|e_{t}|^k \,| \lambda_{t} = \lambda]] 
\leq 
C_1 C_2 < \infty.
\end{align*} 
We conclude that $\mathbb{E}[|\alpha_{t+1}|^k] <\infty$. Now we turn our attention to the recursion \eqref{ln_tv_lambda}. The existence of the $k$-th moment for $\{ \log\lambda_t \}_{t\in\mathbb{Z}}$ can be proved by following similar arguments of \cite{fokianos2009poisson}. In particular, analogously to the proof of Proposition 2.1 of \cite{fokianos2009poisson}, we define a new test function $V(x)=1+x^k	$. Then
\begin{align*}
\mathbb{E}[V(\log \lambda_{t+1}) \, | \log \lambda_{t} = \log \lambda]
=&
1 +
\mathbb{E}[(\omega + \beta \log \lambda + \alpha_t (y_t - \lambda)\lambda^{-1})^k
\, | \log \lambda_{t} = \log \lambda]\\
=&
1 + \sum_{i=0}^k \binom{k}{i} (\beta \log\lambda)^{k-i}
\mathbb{E}[(\omega + \alpha_{t+1} (y_t - \lambda)\lambda^{-1})^i
\, | \log \lambda_{t} = \log \lambda].
\end{align*}
Note that, $\omega$ do not depend on $\log\lambda$ and hence, it can be neglected as $\log\lambda \rightarrow \infty$. This means that it suffices to verify that $\mathbb{E}[(\alpha_{t+1} (y_t - \lambda)\lambda^{-1})^k \, | \log \lambda_{t} = \log \lambda] < \infty$. However, it can be seen that by the H\"{o}lder's inequality
\begin{align*}
\mathbb{E}[|\alpha_{t+1} (y_t - \lambda)\lambda^{-1}|^k]
\leq
\Big(\mathbb{E}[|\alpha_{t+1}|^{2k}]\Big)^{\frac{1}{2k}} 
\Big(\mathbb{E}[|(y_t - \lambda)\lambda^{-1}|^{2k}]\Big)^{\frac{1}{2k}}<\infty,
\end{align*}
which implies that $\mathbb{E}[(\alpha_{t+1} (y_t - \lambda)\lambda^{-1})^k \, | \log \lambda_{t} = \log \lambda] \leq C < \infty$, for some constant $C>0$ and $\forall k>0$. Therefore,
\begin{align*}
\mathbb{E}[V(\log \lambda_{t+1}) \, | \log \lambda_{t} = \log \lambda]
=&
1 + \beta^k \log\lambda + \sum_{j=0}^{k-1} C_j(\log\lambda),
\end{align*}
where $C_j(\log\lambda)$ depends on the coefficients $\omega$, $\alpha$ and $\beta$, and, depending on which cases of proof of Proposition \ref{prop_ergodic_lnLambda} we are taking into account, possibly on $\log\lambda$. In particular, from the arguments discussed above it can be see that for $\log\lambda_t = \log\lambda > 0$ with $\log\lambda\rightarrow \infty$, then $\mathbb{E}[|\log\lambda_{t+1}|^k \, | \log\lambda_t = \log\lambda] \approx |\beta|^k |\log\lambda|^k$, whereas for $\log\lambda_t = \log\lambda < 0$ with $|\log\lambda| \rightarrow \infty$, then $\mathbb{E}[|\log\lambda_{t+1}|^k \, | \log\lambda_t = \log\lambda] \approx |\beta|^k |\beta + \bar{\alpha}|^k |\log\lambda|^k$. Therefore, the conditions stated in Proposition \ref{prop_ergodic_lnLambda} are sufficients for the existence of the $k$-th order moment of $\{ \log\lambda_t \}_{t\in \mathbb{Z}}$.

Finally, by slightly modifying the proof of Proposition 2.2 of \cite{fokianos2009poisson} we obtain that, under the condition of Proposition \ref{prop_ergodic_lnLambda}, the chain $\{y_t, \log\lambda_t\}_{t\in\mathbb{Z}}$ is geometrically ergodic with $V(y,\lambda) = 1+y^k+\log\lambda^k$.

\subsection*{Proof of Proposition \ref{prop_invertibility}}

To derive the contraction condition stated in Proposition \ref{prop_invertibility} we proceed as follows. We consider the model defined by the filtering equations:
\begin{center}
$\log \lambda_{t+1}(\boldsymbol{\theta})
=
\omega + \beta \log \lambda_t(\boldsymbol{\theta}) 
+ \alpha_{t+1}(\boldsymbol{\theta})
(y_t \lambda^{-1}_{t}(\boldsymbol{\theta})-1),$
\end{center}
\begin{center}
$\alpha_{t+1}(\boldsymbol{\theta})
=
\delta_\alpha + \phi_\alpha \alpha_{t}(\boldsymbol{\theta})
+\kappa_\alpha (y_t \lambda^{-1}_{t}(\boldsymbol{\theta})-1) 
(y_{t-1} \lambda^{-1}_{t-1}(\boldsymbol{\theta})-1).$
\end{center}
First note that by iterating the recursion for $\alpha_{t+1}(\boldsymbol{\theta})$, we obtain
\begin{align}
\label{unfold}
\alpha_{t+1}(\boldsymbol{\theta})
=
\frac{\delta_\alpha}{1-\phi_\alpha} + \phi_\alpha^{t-1} \alpha_1(\boldsymbol{\theta}) + \sum_{k=0}^{t-1} \kappa_\alpha \phi_\alpha^{k}
(y_{t-k} \lambda^{-1}_{t-k}(\boldsymbol{\theta})-1) 
(y_{t-k-1} \lambda^{-1}_{t-k-1}(\boldsymbol{\theta})-1),
\end{align}
and so, by applying the Cauchy rule, it is well-known that there exists a stationary, ergodic
and non anticipative solution to the stochastic recursion in \eqref{unfold} if $0 \leq \phi_\alpha < 1$, which is implied by assumption. Thus, 
\begin{align*}
\sup_{\boldsymbol{\theta} \in \boldsymbol{\Theta}}
| \hat{\alpha}_{t+1}(\boldsymbol{\theta}) -  {\alpha}_{t+1}(\boldsymbol{\theta}) |
\leq
\sup_{\boldsymbol{\theta} \in \boldsymbol{\Theta}}
|\phi^{t-1}_\alpha|
\sup_{\boldsymbol{\theta} \in \boldsymbol{\Theta}}
| \hat{\alpha}_{1}(\boldsymbol{\theta}) -  {\alpha}_{1}(\boldsymbol{\theta}) |
\xrightarrow[]{e.a.s.}  0
\,\,\,\,\,\text{as} \,\,\,\,\, t \rightarrow \infty.
\end{align*}
Now, note that $\alpha_{t+1}(\boldsymbol{\theta})$ is bounded below by
\begin{align*}
\alpha_{t+1}(\boldsymbol{\theta}) \geq
\frac{\delta_\alpha + \kappa_\alpha}{1-\phi_\alpha}.
\end{align*}
Therefore
\begin{align*}
\log \lambda_{t+1}(\boldsymbol{\theta})
\geq&
\omega + \beta \log \lambda_t(\boldsymbol{\theta}) 
+ \frac{\delta_\alpha + \kappa_\alpha}{1-\phi_\alpha}
(y_t \lambda^{-1}_{t}(\boldsymbol{\theta})-1)\\
\geq&
\omega + \beta \log \lambda_t(\boldsymbol{\theta}) 
-\frac{\delta_\alpha + \kappa_\alpha}{1-\phi_\alpha}.
\end{align*}
Since $\omega\in\mathbb{R}$ and $0\leq \beta <1$, by recursive substitutions, we obtain that
\begin{align}
\label{ell_bound}
\log \lambda_{t+1}(\boldsymbol{\theta})
\geq&
\frac{\omega - \frac{\delta_\alpha + \kappa_\alpha}{1-\phi_\alpha}}{1-\beta}:= \ell.
\end{align}
Clearly, this implies that
\begin{align*}
\alpha_{t+1}(\boldsymbol{\theta})
\leq 
\delta_\alpha + \phi_\alpha \alpha_{t}(\boldsymbol{\theta})
+\kappa_\alpha (y_t e^{-\ell}-1) 
(y_{t-1} e^{-\ell}-1) := \bar{\alpha}_{t+1}(\boldsymbol{\theta}),
\end{align*}
and hence
\begin{align*}
\lambda_{t+1}(\boldsymbol{\theta})
\leq&
\exp \{\omega + \beta \log \lambda_t(\boldsymbol{\theta}) 
+ \bar{\alpha}_{t+1}(\boldsymbol{\theta})
(y_t e^{-\ell} -1)\}.
\end{align*}
The contraction condition is then obtained by noting that
\begin{align*}
\bigg|
\frac{\partial \lambda_{t+1}(\boldsymbol{\theta})}
{\partial \lambda_{t}(\boldsymbol{\theta})}
\bigg|
\leq&
\exp \{\omega + \beta \log \lambda_t(\boldsymbol{\theta}) 
+ \bar{\alpha}_{t+1}(\boldsymbol{\theta})
(y_t e^{-\ell}-1)\}
\frac{\beta}{\lambda_{t}(\boldsymbol{\theta})} \\
=&
\exp \{\omega 
+ \bar{\alpha}_{t+1}(\boldsymbol{\theta})
(y_t e^{-\ell}-1)\}
\lambda^\beta_t(\boldsymbol{\theta}) 
\frac{\beta}{\lambda_{t}(\boldsymbol{\theta})} \\
=&
\exp \{\omega 
+ \bar{\alpha}_{t+1}(\boldsymbol{\theta})
(y_t e^{-\ell}-1)\}
\frac{\beta}{\lambda^{1-\beta}_{t}(\boldsymbol{\theta})}.
\end{align*}
Since the inequality in \eqref{ell_bound} holds true, it must be that
\begin{align*}
\bigg|
\frac{\partial \lambda_{t+1}(\boldsymbol{\theta})}
{\partial \lambda_{t}(\boldsymbol{\theta})}
\bigg|
\leq&
\exp \{\omega 
+ \bar{\alpha}_{t+1}(\boldsymbol{\theta})
(y_t e^{-\ell}-1)\}
\frac{\beta}{e^{\ell(1-\beta)}_{t}}\\
=&
\beta \exp \{\omega 
+ \bar{\alpha}_{t+1}(\boldsymbol{\theta})
(y_t e^{-\ell}-1) - \ell(1-\beta) \},
\end{align*}
which yields the contraction condition stated in the Proposition.

In fact, by the mean value theorem it holds that
\begin{align*}
\sup_{\boldsymbol{\theta} \in \boldsymbol{\Theta}}
| \hat{\lambda}_{t+1}(\boldsymbol{\theta})
- \lambda_{t+1}(\boldsymbol{\theta}) |
\leq 
\sup_{\boldsymbol{\theta} \in \boldsymbol{\Theta}}
\bigg|
\frac{\partial \hat{\lambda}_{t+1}(\boldsymbol{\theta})}
{\partial {\lambda}^\star}
\bigg|
\sup_{\boldsymbol{\theta} \in \boldsymbol{\Theta}}
| \hat{\lambda}_{t}(\boldsymbol{\theta})
- \lambda_{t}(\boldsymbol{\theta}) |,
\end{align*}
where ${\lambda}^\star \in (\hat{\lambda}_{t}(\boldsymbol{\theta}), \lambda_{t}(\boldsymbol{\theta}))$. The Proposition follows from 
\begin{align*}
\mathbb{E} \bigg[
\log \sup_{\boldsymbol{\theta} \in \boldsymbol{\Theta}}\bigg|
\frac{\partial \lambda_{t+1}(\boldsymbol{\theta})}
{\partial \lambda_{t}(\boldsymbol{\theta})}
\bigg|
\bigg]
< 0, 
\end{align*}
by application of Proposition 5.2.12 in \cite{Straumann2005}. For the application of Proposition 5.2.12 in \cite{Straumann2005}, we have also checked that $\mathbb{E}[\log^+ \sup_{\boldsymbol{\theta} \in \boldsymbol{\Theta}}|\lambda_{t+1}(\boldsymbol{\theta}) - \lambda |]<\infty$ for a constant $\lambda>0$ and 
\begin{align*}
\mathbb{E} \bigg[
\log^+ \sup_{\boldsymbol{\theta} \in \boldsymbol{\Theta}}\bigg|
\frac{\partial \lambda_{t+1}(\boldsymbol{\theta})}
{\partial \lambda_{t}(\boldsymbol{\theta})}
\bigg|
\bigg]
< \infty.
\end{align*}
Note that $\mathbb{E} \big[
\log^+ \sup_{\boldsymbol{\theta} \in \boldsymbol{\Theta}}\big|
\frac{\partial \lambda_{t+1}(\boldsymbol{\theta})}
{\partial \lambda_{t}(\boldsymbol{\theta})}
\big|
\big]
< \infty$ directly follows from $\mathbb{E} \big[
\log \sup_{\boldsymbol{\theta} \in \boldsymbol{\Theta}}\big|
\frac{\partial \lambda_{t+1}(\boldsymbol{\theta})}
{\partial \lambda_{t}(\boldsymbol{\theta})}
\big|
\big]
< 0$. Moreover, we get $\mathbb{E}[\log^+ \sup_{\boldsymbol{\theta} \in \boldsymbol{\Theta}}|\lambda_{t+1}(\boldsymbol{\theta}) - \lambda |]<\infty$ for a constant $\lambda>0$ by simply noting that $y_t$ has a bounded moment from Proposition \ref{prop_ergodic_lnLambda}.

Finally, we address the second part of the Proposition, i.e., the existence of an $m >0$ such that $\mathbb{E}[\sup_{\boldsymbol{\theta} \in \boldsymbol{\Theta}} |\lambda_t(\boldsymbol{\theta})|^m]<\infty$. However, under the stricter contraction condition in \eqref{contr_cond_mom}, this result follows directly by applying Proposition 3.2 of \cite{Blasques2021}, since their condition $(iv)$ follows from our condition \eqref{contr_cond_mom}.

\subsection*{Proof of Lemma \ref{lemma_as_approxLogLik}}
By the triangle inequality, it holds that
\begin{align}
\label{triangle_LogLik}
\sup_{\boldsymbol{\theta}\in\boldsymbol{\Theta}}
| \hat{L}_T(\boldsymbol{\theta}) - {L}(\boldsymbol{\theta}) |
\leq
\sup_{\boldsymbol{\theta}\in\boldsymbol{\Theta}}
| \hat{L}_T(\boldsymbol{\theta}) - {L}_T(\boldsymbol{\theta}) |
+
\sup_{\boldsymbol{\theta}\in\boldsymbol{\Theta}}
| {L}_T(\boldsymbol{\theta}) - {L}(\boldsymbol{\theta}) |,
\end{align}
where ${L}_T(\boldsymbol{\theta}) = \frac{1}{T}\sum_{t=1}^T l_t$ is stationary and ergodic since it is a continuous function of the stationary and ergodic processes $\{\lambda_t(\boldsymbol{\theta})\}_{t\in\mathbb{Z}}$ and $\{\alpha_t(\boldsymbol{\theta})\}_{t\in\mathbb{Z}}$. We begin by showing that the first term in the RHS of \eqref{triangle_LogLik} converges to zero almost surely. By the mean value theorem we have
\begin{align*}
\sup_{\boldsymbol{\theta}\in\boldsymbol{\Theta}}
| \hat{L}_T(\boldsymbol{\theta}) - {L}_T(\boldsymbol{\theta}) |
\leq
\frac{1}{T}
\sum_{t=1}^T
\sup_{\boldsymbol{\theta}\in\boldsymbol{\Theta}}
\begin{Vmatrix}
\frac{\partial l_t(\boldsymbol{\theta})}{\partial \lambda^\star_t(\boldsymbol{\theta})}\\
\frac{\partial l_t(\boldsymbol{\theta})}{\partial \alpha^\star_t(\boldsymbol{\theta})}
\end{Vmatrix}
\sup_{\boldsymbol{\theta}\in\boldsymbol{\Theta}}
\begin{Vmatrix}
\hat{\lambda}_t(\boldsymbol{\theta}) -  \lambda_t(\boldsymbol{\theta})\\
\hat{\alpha}_t(\boldsymbol{\theta}) -  \alpha_t(\boldsymbol{\theta})
\end{Vmatrix},
\end{align*}
where $\lambda^\star_t(\boldsymbol{\theta}) \in (\hat{\lambda}_t(\boldsymbol{\theta}), \lambda_t(\boldsymbol{\theta}))$, $\alpha^\star_t(\boldsymbol{\theta}) \in (\hat{\alpha}_t(\boldsymbol{\theta}), \alpha_t(\boldsymbol{\theta}))$ and the parial derivatives are
$
 \frac{\partial l_t(\boldsymbol{\theta})}{\partial \lambda^\star_t(\boldsymbol{\theta})}
 =
 \frac{y_t}{\lambda^\star_t(\boldsymbol{\theta})} - 1
$ and
$
\frac{\partial l_t(\boldsymbol{\theta})}{\partial \alpha_t(\boldsymbol{\theta})}
 =
 (\frac{y_t}{\lambda_t(\boldsymbol{\theta})} - 1)
 (\frac{y_{t-1}}{\lambda_{t-1}(\boldsymbol{\theta})} - 1)
$.

Note that, from the inequality \eqref{ell_bound} derived in the proof of Proposition \ref{prop_invertibility} and the compactness of $\boldsymbol{\Theta}$, $\log \lambda^\star_t(\boldsymbol{\theta})$ is bounded below by $\ell$. Thus 
$ \sup_{\boldsymbol{\theta}\in\boldsymbol{\Theta}}
\big|
\frac{\partial l_t(\boldsymbol{\theta})}{\partial \lambda^\star_t(\boldsymbol{\theta})}
\big|
\leq  |{y_t} e^{-\ell} - 1|,
$
which implies that
\begin{align}
\label{mvt1}
\sup_{\boldsymbol{\theta}\in\boldsymbol{\Theta}}
| \hat{L}_T(\boldsymbol{\theta}) - {L}_T(\boldsymbol{\theta}) |
\leq
\frac{1}{T}
\sum_{t=1}^T
\begin{Vmatrix}
({y_t} e^{-\ell} - 1)\\
({y_t} e^{-\ell} - 1)({y_{t-1}} e^{-\ell} - 1)
\end{Vmatrix}
\sup_{\boldsymbol{\theta}\in\boldsymbol{\Theta}}
\begin{Vmatrix}
\hat{\lambda}_t(\boldsymbol{\theta}) -  \lambda_t(\boldsymbol{\theta})\\
\hat{\alpha}_t(\boldsymbol{\theta}) -  \alpha_t(\boldsymbol{\theta})
\end{Vmatrix}. 
\end{align}
From Proposition \ref{prop_invertibility} it is clear that the second term in the RHS of \eqref{mvt1} converges to zero almost surely and exponentially fast. Therefore, if we can prove that $({y_t} e^{-\ell} - 1)({y_{t-1}} e^{-\ell} - 1)$ has a bounded $\log$-moment, the desired convergence of the first term in the RHS of \eqref{triangle_LogLik} follows by an application of Lemma 2.5.1 of \cite{Straumann2005}. However, the $\log$-moment condition is trivially verified, since by Proposition \ref{prop_ergodic_lnLambda}, Lemma 2.5.3 of \cite{Straumann2005} and the Jensen's inequality imply that
\begin{align*}
\mathbb{E}[\log^+|({y_t} e^{-\ell} - 1)({y_{t-1}} e^{-\ell} - 1)|]
\leq
\mathbb{E}[\log^+|{y_t} e^{-\ell} - 1|]
+\mathbb{E}[\log^+|{y_{t-1}} e^{-\ell} - 1|]< c_1 + c_2 <\infty,
\end{align*}
for any $c_1,c_2>0$ because $y_t$ has a bounded moment. Therefore, $\exists K > 0$ such that
\begin{align*}
\sup_{\boldsymbol{\theta}\in\boldsymbol{\Theta}}
| \hat{L}_T(\boldsymbol{\theta}) - {L}_T(\boldsymbol{\theta}) |
\leq
\frac{1}{T}
\sum_{t=1}^T
K \times \sup_{\boldsymbol{\theta}\in\boldsymbol{\Theta}}
\begin{Vmatrix}
\hat{\lambda}_t(\boldsymbol{\theta}) -  \lambda_t(\boldsymbol{\theta})\\
\hat{\alpha}_t(\boldsymbol{\theta}) -  \alpha_t(\boldsymbol{\theta})
\end{Vmatrix}
\xrightarrow[]{a.s.}  0
\,\,\,\,\,\text{as} \,\,\,\,\, T \rightarrow \infty.
\end{align*}
Next, we study the almost sure convergence of the second term in the RHS of \eqref{triangle_LogLik}. This result, could be obtained by applying the uniform law of large number for ergodic sequences of \cite{rao1962}. In particular, we need to prove that the stationary and ergodic sequence $\{l_t(\boldsymbol{\theta})\}_{t\in\mathbb{Z}}$ has a uniformly bounded moment over the compact parameter space, i.e., $\mathbb{E}[\sup_{\boldsymbol{\theta}\in\boldsymbol{\Theta}}
| l_1(\boldsymbol{\theta})|]<\infty$.

Since $\log(x) \leq x - 1$ $\forall x > 0$, we have
\begin{align*}
|\ell_t(\boldsymbol{\theta})|
=
|y_t \log \lambda_t(\boldsymbol{\theta}) - \lambda_t((\boldsymbol{\theta}))|
\leq
|y_t ( \lambda_t(\boldsymbol{\theta}) -1 ) - \lambda_t((\boldsymbol{\theta}))|,
\end{align*}
Thus, from the Cauchy-Schwarz inequality we obtain
\begin{align*}
\mathbb{E}\bigg[\sup_{\boldsymbol{\theta} \in \boldsymbol{\Theta}} |l_t(\boldsymbol{\theta})|\bigg]
\leq
\sqrt{\mathbb{E}[y_t^2] 
\mathbb{E}\bigg[\sup_{\boldsymbol{\theta} \in \boldsymbol{\Theta}} |\lambda_t^2(\boldsymbol{\theta})|\bigg]}
+
\mathbb{E}[y_t] + 
\mathbb{E}\bigg[\sup_{\boldsymbol{\theta} \in \boldsymbol{\Theta}} |\lambda_t(\boldsymbol{\theta})|\bigg]
<\infty,
\end{align*}
since $\mathbb{E}[y_t]< \infty$ and $\mathbb{E}[y^2_t]< \infty$ are ensured by Proposition \ref{prop_sec_mom}, and $\mathbb{E}[\sup_{\boldsymbol{\theta} \in \boldsymbol{\Theta}} |\lambda_t^2(\boldsymbol{\theta})|]< \infty$ is ensured by Proposition \ref{prop_invertibility}. Therefore, by the uniform law of large number of \cite{rao1962} 
\begin{align*}
\sup_{\boldsymbol{\theta} \in \boldsymbol{\Theta}} 
| \hat{L}_T(\boldsymbol{\theta})
 - 
L(\boldsymbol{\theta})|
 \xrightarrow[]{\text{a.s.}} 0
\tab \text{as} \tab T \rightarrow \infty.
\end{align*}
By collecting the results, we obtain from \eqref{triangle_LogLik}
\begin{align*}
\sup_{\boldsymbol{\theta}\in\boldsymbol{\Theta}}
| \hat{L}_T(\boldsymbol{\theta}) - {L}(\boldsymbol{\theta}) |
\leq o_p(1) + o_p(1),
\end{align*}
which concludes the proof.

\subsection*{Proof of Lemma \ref{lemma_identifiability}}
Similarly to \cite{Wang2014}, we can show that by the law of iterated expectations and the elementary inequality $\log(x) \leq x - 1$ $\forall x > 0$, it holds that
\begin{align*}
\mathbb{E}[l_t(\boldsymbol{\theta}) - l_t(\boldsymbol{\theta}_0)]
=&
\mathbb{E}
\bigg[
{y_t} \log \frac{\lambda_t(\boldsymbol{\theta})}{\lambda_t(\boldsymbol{\theta}_0)}
+
 \lambda_t(\boldsymbol{\theta})
 -\lambda_t(\boldsymbol{\theta}_0)
\bigg] \\
\leq&
\mathbb{E}
\bigg[
\lambda_t(\boldsymbol{\theta}_0) 
\bigg(
 \frac{\lambda_t(\boldsymbol{\theta})}{\lambda_t(\boldsymbol{\theta}_0)}
 -1
 \bigg)
+
 \lambda_t(\boldsymbol{\theta})
 -\lambda_t(\boldsymbol{\theta}_0)
\bigg]
= 0,
\end{align*}
if and only if $\lambda_t(\boldsymbol{\theta}) = \lambda_t(\boldsymbol{\theta}_0)$ almost surely. However, the latter equality holds if and only if ${\log\lambda}_{t}(\boldsymbol{\theta}) = {\log\lambda}_{t}(\boldsymbol{\theta}_0)$ almost surely.

To prove that this cannot be true, assume that $\exists \boldsymbol{\theta} \in \boldsymbol{\Theta}$ such that $\log \lambda_t(\boldsymbol{\theta}) = \log \lambda_t(\boldsymbol{\theta}_0)$, then
\begin{align}
\label{ln_lambda_diff}
\log \lambda_t(\boldsymbol{\theta}) 
- \log \lambda_t(\boldsymbol{\theta}_0)
=
\omega - \omega_0 
+ (\beta - \beta_0)  \log \lambda_{t-1}(\boldsymbol{\theta}_0)
+ (\alpha_t(\boldsymbol{\theta}) - \alpha_t(\boldsymbol{\theta}_0))
e_{t-1},
\end{align}
almost surely. From equation \eqref{ln_lambda_diff} above it can be seen that if $\log \lambda_t(\boldsymbol{\theta}) = \log \lambda_t(\boldsymbol{\theta}_0)$, then it must holds that $\alpha_t(\boldsymbol{\theta}) = \alpha_t(\boldsymbol{\theta}_0)$ almost surely. However, note that
\begin{align}
\label{alpha_diff}
\alpha_t(\boldsymbol{\theta}) - \alpha_t(\boldsymbol{\theta}_0)
=
\delta_\alpha - \delta_{\alpha,0}
+ (\phi_\alpha - \phi_{\alpha,0})\alpha_t(\boldsymbol{\theta}_0)
+ (\kappa_\alpha - \kappa_{\alpha,0})e_{t-1} e_{t-2},
\end{align}
but the process $\{e_{t}\}_{t\in \mathbb{Z}}$ forms a martingale difference sequence, and thus, since $\kappa_\alpha \neq 0$ by assumption, it can be seen that if $\alpha_t(\boldsymbol{\theta}) - \alpha_t(\boldsymbol{\theta}_0) = 0$ almost surely, we must have $(\delta_\alpha, \phi_\alpha, \kappa_\alpha) = (\delta_{\alpha,0}, \phi_{\alpha,0}, \kappa_{\alpha,0})$. It is evident now from \eqref{ln_lambda_diff} that this result further implies that if $\alpha_t(\boldsymbol{\theta}) - \alpha_t(\boldsymbol{\theta}_0) = 0$ and $\log \lambda_t(\boldsymbol{\theta}) - \log \lambda_t(\boldsymbol{\theta}_0)=0$ almost surely, it must also holds that $(\omega, \beta) = (\omega_0, \beta_0)$, and hence $\boldsymbol{\theta} = \boldsymbol{\theta}_0$.

\subsection*{Proof of Theorem \ref{thm_consistency}}

With Lemma \ref{lemma_as_approxLogLik}, which yields the uniform convergence of the $\log$-likelihood function, and Lemma \ref{lemma_identifiability}, which shows the identifiability and uniqueness of the maximizers $\boldsymbol{\theta}_0 \in \boldsymbol{\Theta}$, at hand, we directly obtain strong consistency of the MLE $\hat{\boldsymbol{\theta}}_T$ by applying Theorem 3.4 of \cite{white_1994}.

\subsection*{Proof of Proposition \ref{prop_norm_score}}
To establish the asymptotic normality of the normalized score function \eqref{score_fun} we verify \eqref{condition_hall1} and \eqref{condition_hall2}.

First, we note that $\log\lambda_t^{\boldsymbol{\theta}_0}$ is a continuous function of $\{ \lambda_t \}_{t\in\mathbb{Z}}$ and $\{ \alpha_t \}_{t\in\mathbb{Z}}$, that are strictly stationary and ergodic by Proposition \ref{prop_ergodic_lnLambda}. Hence, $\{ \log\lambda^{\boldsymbol{\theta}}_t \}_{t\in\mathbb{Z}}$ is strictly stationary and ergodic.

Then, to verify conditions \eqref{condition_hall1}, we note that
\begin{align*}
\sum_{t=1}^{T}\mathbb{E}[\boldsymbol{\eta}_t \boldsymbol{\eta}^\prime_t | \mathcal{F}_{t-1}]
=&
\frac{1}{T}\sum_{t=1}^{T}
\mathbb{E}[e_t^2|\mathcal{F}_{t-1}]\lambda_t^2 
(\log\lambda_t^{\boldsymbol{\theta}_0}\log\lambda_t^{\boldsymbol{\theta}_0^\prime})\\
=&
\frac{1}{T}\sum_{t=1}^{T}
\mathbb{E}[\mathbb{E}[e_t^2 | \lambda_t]|\mathcal{F}_{t-1}]\lambda_t^2 
(\log\lambda_t^{\boldsymbol{\theta}_0}\log\lambda_t^{\boldsymbol{\theta}_0^\prime})\\
=&
\frac{1}{T}\sum_{t=1}^{T}
\mathbb{E}[\lambda^{-1}_t|\mathcal{F}_{t-1}]\lambda_t^2 
(\log\lambda_t^{\boldsymbol{\theta}_0}\log\lambda_t^{\boldsymbol{\theta}_0^\prime})\\
=&
\frac{1}{T}\sum_{t=1}^{T}
\lambda_t
(\log\lambda_t^{\boldsymbol{\theta}_0}\log\lambda_t^{\boldsymbol{\theta}_0^\prime}).
\end{align*}
Therefore, if $\mathbb{E}[\|\lambda_t
(\log\lambda_t^{\boldsymbol{\theta}_0}\log\lambda_t^{\boldsymbol{\theta}_0^\prime})\|]<\infty$, we can apply the ergodic theorem in order to obtain
\begin{align*}
\lim_{t\rightarrow \infty}
\frac{1}{T}\sum_{t=1}^{T}
\lambda_t
(\log\lambda_t^{\boldsymbol{\theta}_0}\log\lambda_t^{\boldsymbol{\theta}_0^\prime})
=
\mathbb{E}[\lambda_t
(\log\lambda_t^{\boldsymbol{\theta}_0}\log\lambda_t^{\boldsymbol{\theta}_0^\prime})],
\end{align*}
almost surely. Clearly, the bounded moment condition $\mathbb{E}[\|\lambda_t
(\log\lambda_t^{\boldsymbol{\theta}_0}\log\lambda_t^{\boldsymbol{\theta}_0^\prime})\|]<\infty$ is automatically satisfied if $\mathbb{E}[\|\lambda_t (\log\lambda_t^{\boldsymbol{\xi}_0}\log\lambda_t^{\boldsymbol{\xi}_0^\prime})\|]<\infty$ and $\mathbb{E}[\|\lambda_t(\alpha_t^{\boldsymbol{\psi}_0}\alpha_t^{\boldsymbol{\psi}_0^\prime})\|]<\infty$, which in turn is implied by verifying that $\mathbb{E}[|\lambda_t (\log\lambda_t^{\omega})^2|]<\infty$ and $\mathbb{E}[|\lambda_t (\log\lambda_t^{\beta})^2|]<\infty$, and moreover $\mathbb{E}[|\lambda_t (\alpha_t^{\delta_\alpha})^2|]<\infty$, $\mathbb{E}[|\lambda_t (\alpha_t^{\phi_\alpha})^2|]<\infty$ and $\mathbb{E}[|\lambda_t (\alpha_t^{\kappa_\alpha})^2|]<\infty$.

Let us denote with $\|\,\cdot\,\|_2$ the $L_2$-norm. By the Minkowsky's inequality, we have
\begin{align*}
\| \lambda_t^{1/2} \log\lambda_t^{\omega} \|_2
\leq&
\| \lambda_t^{1/2} \|_2
+
\sum_{i=1}^\infty \bigg\|\lambda_t^{1/2} \prod_{j=0}^i A_{t-i-1} \bigg\|_2.
\end{align*}
However, note that $\lambda_t$ has a bounded moment, $\{ A_t \}_{t\in\mathbb{Z}}$ is strictly stationary and ergodic, and $\mathbb{E}[|A_t|^2]<1$ by assumption. Therefore,
\begin{align*}
\| \lambda_t^{1/2} \log\lambda_t^{\omega} \|_2
\leq
c_1 
+
c_1 \sum_{i=1}^\infty \|A_{t-i-1} \|^i_2
<\infty.
\end{align*}
By using similar arguments as before, we obtain
\begin{align*}
\| \lambda_t^{1/2} \log\lambda_t^{\beta} \|_2
\leq
c_1 c_2
+
c_1 c_2 \sum_{i=1}^\infty \|A_{t-i-1} \|^i_2
<\infty.
\end{align*}
Now, since $0\leq|\phi_\alpha^2|<1$ by assumption and, as shown in Proposition \ref{prop_sec_mom}, $\{ e_t e_{t-1} \}_{t\in\mathbb{Z}}$ has $k$ bounded moment for every $k>0$, we can easily show that
\begin{align*}
\| \lambda_t^{1/2} \alpha_t^{\boldsymbol{\psi}} \|_2 < \infty.
\end{align*}
Therefore, the desired moment bound $\mathbb{E}[\|\lambda_t(\log\lambda_t^{\boldsymbol{\theta}_0}\log\lambda_t^{\boldsymbol{\theta}_0^\prime})\|]<\infty$ hold true.

The second required condition in \eqref{condition_hall2} can then be easily using the established moment bound, together with the stationarity and ergodicity of $\{ \log\lambda_t \}_{t\in\mathbb{Z}}$, see \cite{Davis2005}.

\subsection*{Proof of Theorem \ref{thm_asy_norm}}
As done for the $\log$-likehood function in \eqref{approx_log_lik_tot}, we first need to study the impact that the initialization of the filters $\hat{\lambda}_1(\boldsymbol{\theta})$ and $\hat{\alpha}_1(\boldsymbol{\theta})$ have on the score function. In particular, we study the convergence of the non-stationary $\hat{\bar{\boldsymbol{{\eta}}}}_{T}(\boldsymbol{\theta}) = \sqrt{T}\frac{\partial \hat{L}_T(\boldsymbol{\theta})}{\partial\boldsymbol{\theta}}$, which is a continuous function of the filters $\{ \hat{\lambda}_t(\boldsymbol{\theta}) \}_{t\in\mathbb{N}}$ and $\{ \hat{\alpha}_t(\boldsymbol{\theta}) \}_{t\in\mathbb{N}}$,  to its stationary and ergodic counterpart ${\bar{\boldsymbol{{\eta}}}}_{T}(\boldsymbol{\theta}) = \sqrt{T}\frac{\partial {L}_T(\boldsymbol{\theta})}{\partial\boldsymbol{\theta}}$.

By the triangle inequality, we have
\begin{align}
\label{ineq_approx_score}
\sup_{\boldsymbol{\theta}\in\boldsymbol{\Theta}}
\|
\hat{\bar{\boldsymbol{{\eta}}}}_{T}(\boldsymbol{\theta})
-
{\bar{\boldsymbol{{\eta}}}}_{T}(\boldsymbol{\theta})
\|
\leq&
\frac{1}{\sqrt{T}}\sum_{t=1}^{T}
\sup_{\boldsymbol{\theta}\in\boldsymbol{\Theta}}
\bigg\|
\frac{\partial \hat{l}_t(\boldsymbol{\theta})}{\partial \hat{\lambda}_t(\boldsymbol{\theta})}
\frac{\partial \hat{\lambda}_t(\boldsymbol{\theta})}{\partial \boldsymbol{\theta}}
-
\frac{\partial l_t(\boldsymbol{\theta})}{\partial {\lambda}_t(\boldsymbol{\theta})}
\frac{\partial {\lambda}_t(\boldsymbol{\theta})}{\partial \boldsymbol{\theta}}
\bigg\|\nonumber\\
&+
\frac{1}{\sqrt{T}}\sum_{t=1}^{T}
\sup_{\boldsymbol{\theta}\in\boldsymbol{\Theta}}
\bigg\|
\frac{\partial \hat{l}_t(\boldsymbol{\theta})}{\partial \hat{\alpha}_t(\boldsymbol{\theta})}
\frac{\partial \hat{\alpha}_t(\boldsymbol{\theta})}{\partial \boldsymbol{\theta}}
-
\frac{\partial l_t(\boldsymbol{\theta})}{\partial {\alpha}_t(\boldsymbol{\theta})}
\frac{\partial {\alpha}_t(\boldsymbol{\theta})}{\partial \boldsymbol{\theta}}
\bigg\|,
\end{align}
almost surely. 

Let us consider the first term in the RHS of inequality \eqref{ineq_approx_score}. By the mean value theorem, we have
\begin{align*}
\sup_{\boldsymbol{\theta}\in\boldsymbol{\Theta}}
\bigg\|
\frac{\partial \hat{l}_t(\boldsymbol{\theta})}{\partial \hat{\lambda}_t(\boldsymbol{\theta})}
\frac{\partial \hat{\lambda}_t(\boldsymbol{\theta})}{\partial \boldsymbol{\theta}}
-
\frac{\partial l_t(\boldsymbol{\theta})}{\partial {\lambda}_t(\boldsymbol{\theta})}
\frac{\partial {\lambda}_t(\boldsymbol{\theta})}{\partial \boldsymbol{\theta}}
\bigg\|
\leq&
\sup_{\boldsymbol{\theta}\in\boldsymbol{\Theta}}
\bigg\|
\frac{\partial \hat{l}_t(\boldsymbol{\theta})}{\partial \hat{\lambda}_t(\boldsymbol{\theta})}
\bigg\|
\sup_{\boldsymbol{\theta}\in\boldsymbol{\Theta}}
\bigg\|
\frac{\partial \hat{\lambda}_t(\boldsymbol{\theta})}{\partial \boldsymbol{\theta}}
-
\frac{\partial {\lambda}_t(\boldsymbol{\theta})}{\partial \boldsymbol{\theta}}
\bigg\|\\
&+
\sup_{\boldsymbol{\theta}\in\boldsymbol{\Theta}}
\bigg\|
\frac{\partial \hat{l}_t(\boldsymbol{\theta})}{\partial \hat{\lambda}_t(\boldsymbol{\theta})}
-
\frac{\partial l_t(\boldsymbol{\theta})}{\partial {\lambda}_t(\boldsymbol{\theta})}
\bigg\|
\sup_{\boldsymbol{\theta}\in\boldsymbol{\Theta}}
\bigg\|
\frac{\partial {\lambda}_t(\boldsymbol{\theta})}{\partial \boldsymbol{\theta}}
\bigg\|\\
\leq&
\sup_{\boldsymbol{\theta}\in\boldsymbol{\Theta}}
\bigg\|
\frac{\partial \hat{l}_t(\boldsymbol{\theta})}{\partial \hat{\lambda}_t(\boldsymbol{\theta})}
\bigg\|
\sup_{\boldsymbol{\theta}\in\boldsymbol{\Theta}}
\bigg\|
\frac{\partial \hat{\lambda}_t(\boldsymbol{\theta})}{\partial \boldsymbol{\theta}}
-
\frac{\partial {\lambda}_t(\boldsymbol{\theta})}{\partial \boldsymbol{\theta}}
\bigg\|\\
&+
\sup_{\boldsymbol{\theta}\in\boldsymbol{\Theta}}
\bigg\|
\frac{\partial^2 \hat{l}_t(\boldsymbol{\theta})}{\partial {\lambda}^\star}
\bigg\|
\sup_{\boldsymbol{\theta}\in\boldsymbol{\Theta}}
\bigg\|
\frac{\partial {\lambda}_t(\boldsymbol{\theta})}{\partial \boldsymbol{\theta}}
\bigg\|
\sup_{\boldsymbol{\theta}\in\boldsymbol{\Theta}}
\|\hat{\lambda}_t(\boldsymbol{\theta}) 
-{\lambda}_t(\boldsymbol{\theta})\|,
\end{align*}
where ${\lambda}^\star \in(\hat{\lambda}_t(\boldsymbol{\theta}), {\lambda}_t(\boldsymbol{\theta}))$. As noted in the proof of Lemma \ref{lemma_as_approxLogLik}, it is clearly seen that $\mathbb{E}\Big[\sup_{\boldsymbol{\theta}\in\boldsymbol{\Theta}}\Big\| \frac{\partial \hat{l}_t(\boldsymbol{\theta})}{\partial \hat{\lambda}_t(\boldsymbol{\theta})}\Big\|\Big] < \infty$ and $\mathbb{E}\Big[\sup_{\boldsymbol{\theta}\in\boldsymbol{\Theta}}\Big\| \frac{\partial^2 \hat{l}_t(\boldsymbol{\theta})}{\partial^2 \hat{\lambda}_t(\boldsymbol{\theta})}\Big\|\Big] < \infty$. Thus, the results implied by Proposition \ref{prop_invertibility} and Lemma \ref{lemma_deriv_proc}, together with Lemma 2.5.2 of \cite{Straumann2005}, we obtain that
\begin{align*}
\sup_{\boldsymbol{\theta}\in\boldsymbol{\Theta}}
\bigg\|
\frac{\partial \hat{l}_t(\boldsymbol{\theta})}{\partial \hat{\lambda}_t(\boldsymbol{\theta})}
\frac{\partial \hat{\lambda}_t(\boldsymbol{\theta})}{\partial \boldsymbol{\theta}}
-
\frac{\partial l_t(\boldsymbol{\theta})}{\partial {\lambda}_t(\boldsymbol{\theta})}
\frac{\partial {\lambda}_t(\boldsymbol{\theta})}{\partial \boldsymbol{\theta}}
\bigg\|\xrightarrow[]{a.s.}  0
\,\,\,\,\,\text{as} \,\,\,\,\, t \rightarrow \infty.
\end{align*}
By using similar arguments, it is also straightforward to see that the same result hold true for the second term in the RHS of inequality \eqref{ineq_approx_score}, hence
\begin{align*}
\sup_{\boldsymbol{\theta}\in\boldsymbol{\Theta}}
\bigg\|
\frac{\partial \hat{l}_t(\boldsymbol{\theta})}{\partial \hat{\alpha}_t(\boldsymbol{\theta})}
\frac{\partial \hat{\alpha}_t(\boldsymbol{\theta})}{\partial \boldsymbol{\theta}}
-
\frac{\partial l_t(\boldsymbol{\theta})}{\partial {\alpha}_t(\boldsymbol{\theta})}
\frac{\partial {\alpha}_t(\boldsymbol{\theta})}{\partial \boldsymbol{\theta}}
\bigg\|\xrightarrow[]{a.s.}  0
\,\,\,\,\,\text{as} \,\,\,\,\, t \rightarrow \infty.
\end{align*}
Putting all together, we thus conclude that the normalized score function satisfies
\begin{align*}
\sup_{\boldsymbol{\theta}\in\boldsymbol{\Theta}}
\|
\hat{\bar{\boldsymbol{{\eta}}}}_{T}(\boldsymbol{\theta})
-
{\bar{\boldsymbol{{\eta}}}}_{T}(\boldsymbol{\theta})
\|
\leq
o_p(1).
\end{align*}
Now, having proved the convergence of the normalized score function, for $T$ large enough, we have
\begin{align}
\label{taylor_exp}
\boldsymbol{0}
=
\hat{\bar{\boldsymbol{{\eta}}}}_{T}(\hat{\boldsymbol{\theta}}_T)
=
\bar{\boldsymbol{{\eta}}}_{T}(\hat{\boldsymbol{\theta}}_T) + o_p(1)
=
\bar{\boldsymbol{{\eta}}}_{T}({\boldsymbol{\theta}}_0)
-
\boldsymbol{J}_T^\star \sqrt{T} (\hat{\boldsymbol{\theta}}_T - {\boldsymbol{\theta}}_0),
\end{align}
where
\begin{align*}
\boldsymbol{J}_T^\star
=
-\frac{\partial^2 L_T(\boldsymbol{\theta}^\star)}{\partial\boldsymbol{\theta}\partial\boldsymbol{\theta}^\prime}
=
-\frac{1}{T}\sum_{t=1}^T\frac{\partial^2 l_t(\boldsymbol{\theta}^\star)}{\partial\boldsymbol{\theta}\partial\boldsymbol{\theta}^\prime}
,
\end{align*}
for some $\boldsymbol{\theta}^\star \in (\hat{\boldsymbol{\theta}}_T, \boldsymbol{\theta}_0)$. Note that, from Proposition \ref{prop_norm_score}, $\bar{\boldsymbol{{\eta}}}_{T}({\boldsymbol{\theta}}_0) \Rightarrow \mathcal{N}(\boldsymbol{0}, \boldsymbol{V})$ as $T \rightarrow \infty$. Therefore, it only remains to show that $\boldsymbol{J}_T^\star \xrightarrow[]{a.s.} \boldsymbol{J}$  as $T \rightarrow \infty$.

However, by combining the strong consistency of the MLE together with Lemma \ref{lemma_bound_sec_deriv_LogLik}, and the uniform law of large numbers for stationary and ergodic processes of \cite{rao1962}, we note that since $\hat{\boldsymbol{\theta}}_T \xrightarrow[]{a.s.} \boldsymbol{\theta}_0$ as $T\rightarrow \infty$, it must holds that $\hat{\boldsymbol{\theta}}_T \in V({\boldsymbol{\theta}}_0)$ almost surely, where $V({\boldsymbol{\theta}}_0)$ denotes a neighborhood of ${\boldsymbol{\theta}}_0$. Then, we obtain that $\hat{\boldsymbol{J}}_T = -\frac{\partial^2 L_T(\hat{\boldsymbol{\theta}}_T)}{\partial\boldsymbol{\theta}\partial\boldsymbol{\theta}^\prime}\xrightarrow[]{a.s.} \boldsymbol{J}$. Moreover, note also that when the model is correct it holds that $\boldsymbol{V} = \boldsymbol{J}$, which follows by standard arguments since $l(\boldsymbol{\theta}_0)$ is the true conditional probability mass function of the Poisson distribution.

Finally, we get that $\boldsymbol{J}^\star_T \xrightarrow[]{a.s.} \boldsymbol{J}$ and the conclusion now simply follows by combining the expansion in equation \eqref{taylor_exp} together with the results obtained in Proposition \ref{prop_norm_score}.

\begin{lemma}
\label{lemma_deriv_proc}
Let the assumptions of Proposition \ref{prop_norm_score} hold, then
\begin{align*}
\sup_{\boldsymbol{\theta}\in\boldsymbol{\Theta}}
\bigg\|
\frac{\partial \hat{\lambda}_t(\boldsymbol{\theta})}{\partial \boldsymbol{\theta}} 
-
\frac{\partial {\lambda}_t(\boldsymbol{\theta})}{\partial \boldsymbol{\theta}}
\bigg\|\xrightarrow[]{e.a.s.}  0
\,\,\,\,\,\text{as} \,\,\,\,\, t \rightarrow \infty,
\end{align*}
where $\frac{\partial {\lambda}_t(\boldsymbol{\theta})}{\partial \boldsymbol{\theta}}$ is the stationary and ergodic derivative process of ${\lambda}_t(\boldsymbol{\theta})$, which satisfies $\mathbb{E}\Big[\sup_{\boldsymbol{\theta}\in\boldsymbol{\Theta}}\Big\| \frac{\partial {\lambda}_t(\boldsymbol{\theta})}{\partial \boldsymbol{\theta}}\Big\|^2\Big]<\infty$.
\end{lemma}
\begin{proof}
Since $\boldsymbol{\theta} = (\boldsymbol{\xi}^\prime, \boldsymbol{\psi}^\prime)^\prime$, we define the derivative processes evaluated at some $\boldsymbol{\theta}\in\boldsymbol{\Theta}$ as $ \log\lambda^{\boldsymbol{\theta}}_t(\boldsymbol{\theta})=( \log\lambda^{\boldsymbol{\xi}^\prime}_t(\boldsymbol{\theta}), \log\lambda^{\boldsymbol{\psi}^\prime}_t(\boldsymbol{\theta}) )^\prime$, where
\begin{align*}
\log\lambda^{\boldsymbol{\xi}}_t(\boldsymbol{\theta})
=
\begin{bmatrix}
\frac{\partial\log\lambda_t(\boldsymbol{\theta})}{\partial \omega}\\
\frac{\partial\log\lambda_t(\boldsymbol{\theta})}{\partial \beta}
\end{bmatrix}
:=
\begin{bmatrix}
\log\lambda^{\omega}_t(\boldsymbol{\theta}) \\
\log\lambda^{\beta}_t(\boldsymbol{\theta})
\end{bmatrix}
=&
\begin{bmatrix}
1+
A_{t-1}(\boldsymbol{\theta}) \log\lambda^{\omega}_{t-1}(\boldsymbol{\theta})\\
\log\lambda_{t-1}(\boldsymbol{\theta})+
A_{t-1}(\boldsymbol{\theta}) \log\lambda^{\beta}_{t-1}(\boldsymbol{\theta})
\end{bmatrix},
\end{align*}
with 
\begin{align*}
A_t(\boldsymbol{\theta}) = 
\bigg[\beta - \alpha_t(\boldsymbol{\theta}) y_t \lambda^{-1}_t(\boldsymbol{\theta}) + (y_t - \lambda_t(\boldsymbol{\theta}))\frac{\partial \alpha_t(\boldsymbol{\theta})}{\partial \lambda_t(\boldsymbol{\theta})}\bigg]\lambda_t(\boldsymbol{\theta}).
\end{align*}

Moreover, $\log\lambda^{\boldsymbol{\psi}}_t(\boldsymbol{\theta}) = \alpha^{\boldsymbol{\psi}}_t(\boldsymbol{\theta}) (y_t - \lambda_t(\boldsymbol{\theta}))\lambda_t^{-1}(\boldsymbol{\theta})$, where
\begin{align*}
\alpha^{\boldsymbol{\psi}}_t(\boldsymbol{\theta})
&=
\begin{bmatrix}
\frac{\partial\alpha_t(\boldsymbol{\theta})}{\partial \delta_{\alpha}}\\
\frac{\partial\alpha_t(\boldsymbol{\theta})}{\partial \phi_{\alpha}} \\
\frac{\partial\alpha_t(\boldsymbol{\theta})}{\partial \kappa_{\alpha}}
\end{bmatrix}\\
:=&
\begin{bmatrix}
\alpha_t^{\delta_{\alpha}}(\boldsymbol{\theta}) \\
\alpha_t^{\phi_{\alpha}}(\boldsymbol{\theta}) \\
\alpha_t^{\kappa_{\alpha}}(\boldsymbol{\theta})
\end{bmatrix}
=
\begin{bmatrix}
1 + 
\phi_{\alpha} \alpha_{t-1}^{\delta_{\alpha}}(\boldsymbol{\theta}) \\
 \alpha_{t-1}(\boldsymbol{\theta}) +
\phi_{\alpha} \alpha_{t-1}^{\phi_{\alpha}}(\boldsymbol{\theta}) \\
 (y_{t-1} - \lambda_{t-1}(\boldsymbol{\theta}))\lambda_{t-1}^{-1}(\boldsymbol{\theta})
 (y_{t-2} - \lambda_{t-2}(\boldsymbol{\theta}))\lambda_{t-2}^{-1}(\boldsymbol{\theta}) +
\phi_{\alpha} \alpha_{t-1}^{\kappa_{\alpha}}(\boldsymbol{\theta})
\end{bmatrix}.
\end{align*}
To verify the uniform convergence $\sup_{\boldsymbol{\theta}\in\boldsymbol{\Theta}}
\|
\frac{\partial \hat{\lambda}_t(\boldsymbol{\theta})}{\partial \boldsymbol{\theta}} 
-
\frac{\partial {\lambda}_t(\boldsymbol{\theta})}{\partial \boldsymbol{\theta}}
\bigg\|\xrightarrow[]{e.a.s.}  0$ we first check the condition in Theorem 2.6.4. \cite{Straumann2005}, which is equivalent in only proving that $\sup_{\boldsymbol{\theta}\in\boldsymbol{\Theta}}
\|
\hat{A}_t(\boldsymbol{\theta})
-
A_t(\boldsymbol{\theta})
\|\xrightarrow[]{e.a.s.}  0$, where $\hat{A}_t(\boldsymbol{\theta})$ and ${A}_t(\boldsymbol{\theta})$ are continuous function of $\hat{\lambda}_t(\boldsymbol{\theta})$ and ${\lambda}_t(\boldsymbol{\theta})$ (and so of $\hat{\alpha}_t(\boldsymbol{\theta})$ and ${\alpha}_t(\boldsymbol{\theta})$), respectively, since the uniform convergence of $\log\hat{\lambda}_t(\boldsymbol{\theta})$, $\hat{\lambda}_t(\boldsymbol{\theta})$, $\hat{\lambda}^{-1}_t(\boldsymbol{\theta})$ and $\hat{\alpha}_t(\boldsymbol{\theta})$ are ensured by Proposition \ref{prop_invertibility}. 

However, from the results of Propositions \ref{prop_ergodic_lnLambda} and \ref{prop_sec_mom}, and by an application of the mean value theorem, it is immediate to note that $\exists K>0$ such that 
\begin{align*}
\sup_{\boldsymbol{\theta}\in\boldsymbol{\Theta}}
\|
\hat{A}_t(\boldsymbol{\theta})
-
A_t(\boldsymbol{\theta})
\|\leq
K \times 
\sup_{\boldsymbol{\theta}\in\boldsymbol{\Theta}}
\|
\hat{\lambda}_t(\boldsymbol{\theta})
-
\lambda_t(\boldsymbol{\theta})
\|\xrightarrow[]{e.a.s.}  0.
\end{align*}
The desired uniform convergence of the derivative process $\frac{\partial \hat{\lambda}_t(\boldsymbol{\theta})}{\partial \boldsymbol{\theta}}$ is then a clear consequence of Lemma 2.5.4 of \cite{Straumann2005}.

Finally, the second moment bound $\mathbb{E}\Big[\sup_{\boldsymbol{\theta}\in\boldsymbol{\Theta}}\Big\| \frac{\partial {\lambda}_t(\boldsymbol{\theta})}{\partial \boldsymbol{\theta}}\Big\|^2\Big]<1$ follows directly since $\mathbb{E}[\sup_{\boldsymbol{\theta}\in\boldsymbol{\Theta}}|A_t(\boldsymbol{\theta})|^2]<\infty$ holds by assumption together with $\mathbb{E}[\sup_{\boldsymbol{\theta}\in\boldsymbol{\Theta}}|\lambda_t(\boldsymbol{\theta})|^m]<\infty$ by Proposition \ref{prop_invertibility} and $\mathbb{E}[y_t^2]<\infty$ by Proposition \ref{prop_sec_mom}.
\end{proof}

\begin{lemma}
\label{lemma_sec_deriv_proc}
Let the assumptions of Proposition \ref{prop_norm_score} hold, then
\begin{align*}
\sup_{\boldsymbol{\theta}\in\boldsymbol{\Theta}}
\bigg\|
\frac{\partial^2 \hat{\lambda}_t(\boldsymbol{\theta})}{\partial \boldsymbol{\theta}\partial \boldsymbol{\theta}^\prime} 
-
\frac{\partial^2 {\lambda}_t(\boldsymbol{\theta})}{\partial \boldsymbol{\theta}\partial\boldsymbol{\theta}^\prime}
\bigg\|\xrightarrow[]{e.a.s.}  0
\,\,\,\,\,\text{as} \,\,\,\,\, t \rightarrow \infty,
\end{align*}
where $\frac{\partial^2 {\lambda}_t(\boldsymbol{\theta})}{\partial \boldsymbol{\theta}\partial\boldsymbol{\theta}^\prime}$ is the stationary and ergodic second derivative process of ${\lambda}_t(\boldsymbol{\theta})$, which satisfies $\mathbb{E}\Big[\sup_{\boldsymbol{\theta}\in\boldsymbol{\Theta}}\Big\| \frac{\partial^2 {\lambda}_t(\boldsymbol{\theta})}{\partial \boldsymbol{\theta}\partial\boldsymbol{\theta}^\prime}\Big\|\Big]<\infty$.
\end{lemma}
\begin{proof}
Using the same notation of Lemma \ref{lemma_deriv_proc}, the second derivative processes evaluated at some $\boldsymbol{\theta}\in \boldsymbol{\Theta}$ as $\text{vech}\log\lambda_t^{\boldsymbol{\theta}\boldsymbol{\theta}}(\boldsymbol{\theta}) 
=
(\text{vech}\lambda_t^{\boldsymbol{\xi}\boldsymbol{\xi}}(\boldsymbol{\theta}), 
\text{vec}\lambda_t^{\boldsymbol{\xi}\boldsymbol{\psi}}(\boldsymbol{\theta}),
\text{vech}\lambda_t^{\boldsymbol{\psi}\boldsymbol{\psi}}(\boldsymbol{\theta}))^\prime$, where
\begin{align*}
\text{vech}\log\lambda^{\boldsymbol{\xi}\boldsymbol{\xi}}_t(\boldsymbol{\theta})
&=
\text{vech}
\begin{bmatrix}
\frac{\partial^2\log\lambda_t(\boldsymbol{\theta})}{\partial \omega^2} &
\frac{\partial\log\lambda_t(\boldsymbol{\theta})}{\partial \omega\partial \beta}\\
\cdot & 
\frac{\partial^2\log\lambda_t(\boldsymbol{\theta})}{\partial \beta^2}
\end{bmatrix}\\
:=&
\begin{bmatrix}
\log\lambda^{\omega\omega}_t(\boldsymbol{\theta})\\
\log\lambda^{\omega\beta}_t(\boldsymbol{\theta})\\
\log\lambda^{\beta\beta}_t(\boldsymbol{\theta})
\end{bmatrix}
=
\begin{bmatrix}
B_{t-1}(\boldsymbol{\theta})(\log\lambda_{t-1}^\omega(\boldsymbol{\theta}))^2
+ A_{t-1}(\boldsymbol{\theta})\log\lambda_{t-1}^{\omega\omega}(\boldsymbol{\theta})\\
B_{t-1}(\boldsymbol{\theta})\log\lambda_{t-1}^\omega(\boldsymbol{\theta})
\log\lambda_{t-1}^\beta(\boldsymbol{\theta})
+\log\lambda_{t-1}^\beta(\boldsymbol{\theta})
+ A_{t-1}(\boldsymbol{\theta})\log\lambda_{t-1}^{\omega\beta}(\boldsymbol{\theta})\\
B_{t-1}(\boldsymbol{\theta})(\log\lambda_{t-1}^\beta(\boldsymbol{\theta}))^2
+2\log\lambda_{t-1}^\beta(\boldsymbol{\theta})
+ A_{t-1}(\boldsymbol{\theta})\log\lambda_{t-1}^{\beta\beta}(\boldsymbol{\theta})
\end{bmatrix},
\end{align*}
with $A_t(\boldsymbol{\theta})$ as defined in Lemma \ref{lemma_deriv_proc} and
\begin{align*}
B_t(\boldsymbol{\theta}) = \bigg\{& \bigg[\alpha_t(\boldsymbol{\theta}) y_t \lambda^{-2}_t(\boldsymbol{\theta}) - \frac{\partial \alpha_t(\boldsymbol{\theta})}{\partial \lambda_t(\boldsymbol{\theta})} + (y_t-\lambda_t(\boldsymbol{\theta}))\frac{\partial^2 \alpha_t(\boldsymbol{\theta})}{\partial \lambda^2_t(\boldsymbol{\theta})} \bigg]\lambda_t(\boldsymbol{\theta})\\
&+ \bigg[\beta - \alpha_t(\boldsymbol{\theta}) y_t \lambda^{-1}_t(\boldsymbol{\theta}) + (y_t - \lambda_t(\boldsymbol{\theta}))\frac{\partial \alpha_t(\boldsymbol{\theta})}{\partial \lambda_t(\boldsymbol{\theta})}\bigg] \bigg\}\lambda_t(\boldsymbol{\theta}).
\end{align*}
Now, 
\begin{align*}
\text{vec}&\log\lambda^{\boldsymbol{\xi}\boldsymbol{\psi}}_t(\boldsymbol{\theta})
=
\text{vec}
\begin{bmatrix}
\frac{\partial^2\log\lambda_t(\boldsymbol{\theta})}{\partial \omega\partial {\delta_\alpha}} &
\frac{\partial^2\log\lambda_t(\boldsymbol{\theta})}{\partial \omega\partial {\phi_\alpha}} &
\frac{\partial^2\log\lambda_t(\boldsymbol{\theta})}{\partial \omega\partial {\kappa_\alpha}} \\
\frac{\partial^2\log\lambda_t(\boldsymbol{\theta})}{\partial \beta\partial {\delta_\alpha}} &
\frac{\partial^2\log\lambda_t(\boldsymbol{\theta})}{\partial \beta\partial {\phi_\alpha}} &
\frac{\partial^2\log\lambda_t(\boldsymbol{\theta})}{\partial \beta\partial {\kappa_\alpha}}
\end{bmatrix}
:=
\begin{bmatrix}
\log\lambda^{\omega\delta_\alpha}_t(\boldsymbol{\theta})\\
\log\lambda^{\omega\phi_\alpha}_t(\boldsymbol{\theta})\\
\log\lambda^{\omega\kappa_\alpha}_t(\boldsymbol{\theta})\\
\log\lambda^{\beta\delta_\alpha}_t(\boldsymbol{\theta})\\
\log\lambda^{\beta\phi_\alpha}_t(\boldsymbol{\theta})\\
\log\lambda^{\beta\kappa_\alpha}_t(\boldsymbol{\theta})
\end{bmatrix}\\
=&
\begin{bmatrix}
C_{t-1}(\boldsymbol{\theta})\alpha^{\delta_\alpha}_{t-1}(\boldsymbol{\theta})
\log\lambda^\omega_{t-1}(\boldsymbol{\theta})
+ A_{t-1}(\boldsymbol{\theta})\log\lambda^{\omega\delta_\alpha}_{t-1}(\boldsymbol{\theta})\\
C_{t-1}(\boldsymbol{\theta})\alpha^{\phi_\alpha}_{t-1}(\boldsymbol{\theta})
\log\lambda^\omega_{t-1}(\boldsymbol{\theta})
+ A_{t-1}(\boldsymbol{\theta})\log\lambda^{\omega\phi_\alpha}_{t-1}(\boldsymbol{\theta})\\
C_{t-1}(\boldsymbol{\theta})\alpha^{\kappa_\alpha}_{t-1}(\boldsymbol{\theta})
\log\lambda^\omega_{t-1}(\boldsymbol{\theta})
+ A_{t-1}(\boldsymbol{\theta})\log\lambda^{\omega\kappa_\alpha}_{t-1}(\boldsymbol{\theta})\\
\alpha_{t-1}^{\delta_{\alpha}}(\boldsymbol{\theta}) 
(y_{t-1} - \lambda_{t-1}(\boldsymbol{\theta}))
\lambda^{-1}_{t-1}(\boldsymbol{\theta})
+ C_{t-1}(\boldsymbol{\theta})\alpha^{\delta_\alpha}_{t-1}(\boldsymbol{\theta})
\log\lambda^\beta_{t-1}(\boldsymbol{\theta})
+ A_{t-1}(\boldsymbol{\theta})\log\lambda^{\beta\delta_\alpha}_{t-1}(\boldsymbol{\theta})\\
\alpha_{t-1}^{\phi_{\alpha}}(\boldsymbol{\theta}) 
(y_{t-1} - \lambda_{t-1}(\boldsymbol{\theta}))
\lambda^{-1}_{t-1}(\boldsymbol{\theta})
+ C_{t-1}(\boldsymbol{\theta})\alpha^{\phi_\alpha}_{t-1}(\boldsymbol{\theta})
\log\lambda^\omega_{t-1}(\boldsymbol{\theta})
+ A_{t-1}(\boldsymbol{\theta})\log\lambda^{\beta\phi_\alpha}_{t-1}(\boldsymbol{\theta})\\
\alpha_{t-1}^{\kappa_{\alpha}}(\boldsymbol{\theta})
(y_{t-1} - \lambda_{t-1}(\boldsymbol{\theta}))
\lambda^{-1}_{t-1}(\boldsymbol{\theta})
+ C_{t-1}(\boldsymbol{\theta})\alpha^{\kappa_\alpha}_{t-1}(\boldsymbol{\theta})
\log\lambda^\beta_{t-1}(\boldsymbol{\theta})
+ A_{t-1}(\boldsymbol{\theta})\log\lambda^{\beta\kappa_\alpha}_{t-1}(\boldsymbol{\theta})
\end{bmatrix},
\end{align*}
with
\begin{align*}
C_t(\boldsymbol{\theta})
=
-y_t\lambda_t(\boldsymbol{\theta}).
\end{align*}
Finally, we have $\text{vech}\log\lambda^{\boldsymbol{\psi\psi}}_t(\boldsymbol{\theta})=\text{vech}\,\alpha_t^{\boldsymbol{\psi\psi}}(\boldsymbol{\theta})(y_t - \lambda_t(\boldsymbol{\theta}))\lambda^{-1}_t(\boldsymbol{\theta})$, where
\begin{align*}
\text{vech}\,\alpha_t^{\boldsymbol{\psi\psi}}(\boldsymbol{\theta})
&=
\text{vech}
\begin{bmatrix}
\frac{\partial^2\alpha_t(\boldsymbol{\theta})}{\partial\delta_\alpha^2 } &
\frac{\partial^2\alpha_t(\boldsymbol{\theta})}{\partial\delta_\alpha \partial\phi_\alpha } &
\frac{\partial^2\alpha_t(\boldsymbol{\theta})}{\partial\delta_\alpha \partial\kappa_\alpha } \\
\cdot &
\frac{\partial^2\alpha_t(\boldsymbol{\theta})}{\partial\phi^2_\alpha } &
\frac{\partial^2\alpha_t(\boldsymbol{\theta})}{\partial\phi_\alpha \partial\kappa_\alpha } \\
\cdot &
\cdot &
\frac{\partial^2\alpha_t(\boldsymbol{\theta})}{\partial\kappa^2_\alpha }
\end{bmatrix}\\
:=&
\begin{bmatrix}
\alpha^{\delta_\alpha\delta_\alpha}_t(\boldsymbol{\theta})\\
\alpha^{\delta_\alpha\phi_\alpha}_t(\boldsymbol{\theta})\\
\alpha^{\delta_\alpha\kappa_\alpha}_t(\boldsymbol{\theta})\\
\alpha^{\phi_\alpha\phi_\alpha}_t(\boldsymbol{\theta})\\
\alpha^{\phi_\alpha\kappa_\alpha}_t(\boldsymbol{\theta})\\
\alpha^{\kappa_\alpha\kappa_\alpha}_t(\boldsymbol{\theta})
\end{bmatrix}
=
\begin{bmatrix}
\phi_\alpha \alpha^{\delta_\alpha\delta_\alpha}_{t-1}(\boldsymbol{\theta})\\
\alpha^{\delta_\alpha}_{t-1}(\boldsymbol{\theta})
+\phi_\alpha \alpha^{\delta_\alpha\phi_\alpha}_{t-1}(\boldsymbol{\theta})\\
\phi_\alpha \alpha^{\delta_\alpha\kappa_\alpha}_{t-1}(\boldsymbol{\theta})\\
(1+\phi_\alpha)\alpha^{\phi_\alpha}_{t-1}(\boldsymbol{\theta})
+ \phi_\alpha \alpha^{\phi_\alpha\phi_\alpha}_{t-1}(\boldsymbol{\theta})\\
\alpha^{\kappa_\alpha}_{t-1}(\boldsymbol{\theta})
+ \phi_\alpha \alpha^{\phi_\alpha\kappa_\alpha}_{t-1}(\boldsymbol{\theta})\\
\phi_\alpha \alpha^{\kappa_\alpha\kappa_\alpha}_{t-1}(\boldsymbol{\theta})
\end{bmatrix}.
\end{align*}
Analogously to the proof of Lemma \ref{lemma_deriv_proc}, to verify the uniform convergence $\sup_{\boldsymbol{\theta}\in\boldsymbol{\Theta}}
\|
\frac{\partial^2 \hat{\lambda}_t(\boldsymbol{\theta})}{\partial \boldsymbol{\theta}\partial \boldsymbol{\theta}^\prime} 
-
\frac{\partial^2 {\lambda}_t(\boldsymbol{\theta})}{\partial \boldsymbol{\theta} \partial \boldsymbol{\theta}^\prime}
\bigg\|\xrightarrow[]{e.a.s.}  0$ we need to check the condition in Theorem 2.6.4. \cite{Straumann2005}, which, in addition to the previous proof, it results in proving that $\sup_{\boldsymbol{\theta}\in\boldsymbol{\Theta}}
\|
\hat{B}_t(\boldsymbol{\theta})
-
B_t(\boldsymbol{\theta})
\|\xrightarrow[]{e.a.s.}  0$ and $\sup_{\boldsymbol{\theta}\in\boldsymbol{\Theta}}
\|
\hat{C}_t(\boldsymbol{\theta})
-
C_t(\boldsymbol{\theta})
\|\xrightarrow[]{e.a.s.}  0$ where $\hat{B}_t(\boldsymbol{\theta}), \hat{C}_t(\boldsymbol{\theta})$ and ${B}_t(\boldsymbol{\theta}), {C}_t(\boldsymbol{\theta})$, are continuous function of $\hat{\lambda}_t(\boldsymbol{\theta})$ and ${\lambda}_t(\boldsymbol{\theta})$ (and so of $\hat{\alpha}_t(\boldsymbol{\theta})$ and ${\alpha}_t(\boldsymbol{\theta})$), respectively. As remarked in the Proof of Lemma \ref{lemma_deriv_proc}, given the results obtained in Proposition \ref{prop_invertibility}, we can easily show the required uniform convergences by repeated applications the mean value theorem together with Lemma 2.5.4 of \cite{Straumann2005}.

Finally, the moment bound $\mathbb{E}\Big[\sup_{\boldsymbol{\theta}\in\boldsymbol{\Theta}}\Big\| \frac{\partial^2 {\lambda}_t(\boldsymbol{\theta})}{\partial \boldsymbol{\theta}\partial \boldsymbol{\theta}^\prime}\Big\|\Big]<1$ follows directly since $\mathbb{E}[\sup_{\boldsymbol{\theta}\in\boldsymbol{\Theta}}|A_t(\boldsymbol{\theta})|^2]<\infty$ holds by assumption together with $\mathbb{E}[\sup_{\boldsymbol{\theta}\in\boldsymbol{\Theta}}|\lambda_t(\boldsymbol{\theta})|^m]<\infty$ by Proposition \ref{prop_invertibility} and $\mathbb{E}[y_t^2]<\infty$ by Proposition \ref{prop_sec_mom}, which further imply $\mathbb{E}[\sup_{\boldsymbol{\theta}\in\boldsymbol{\Theta}}|B_t(\boldsymbol{\theta})|]<\infty$ and $\mathbb{E}[\sup_{\boldsymbol{\theta}\in\boldsymbol{\Theta}}|C_t(\boldsymbol{\theta})|]<\infty$.
\end{proof}

\begin{lemma}
\label{lemma_bound_sec_deriv_LogLik}
Let the assumptions of Theorem \ref{thm_asy_norm} hold, then the second derivatives of the $\log$-likelihood function have a uniformly bounded moment over the compact parameter space, that is
\begin{align*}
\mathbb{E}\bigg[
\sup_{\boldsymbol{\theta}\in\boldsymbol{\Theta}}
\bigg\|
\frac{\partial^2 l_t(\boldsymbol{\theta})}
{\partial\boldsymbol{\theta}\partial\boldsymbol{\theta}^\prime}
\bigg\|
\bigg]<\infty.
\end{align*}
\end{lemma}
\begin{proof}
By straightforward calculations, it can be seen that
\begin{align*}
\frac{\partial^2 l_t(\boldsymbol{\theta})}
{\partial\boldsymbol{\theta}\partial\boldsymbol{\theta}^\prime}
=
\bigg(\frac{y_t}{\lambda_t(\boldsymbol{\theta})} - 1 \bigg)
\frac{\partial^2 {\lambda}_t(\boldsymbol{\theta})}{\partial \boldsymbol{\theta}\partial\boldsymbol{\theta}^\prime}
-
\frac{y_t}{\lambda^2_t(\boldsymbol{\theta})}
\frac{\partial {\lambda}_t(\boldsymbol{\theta})}{\partial \boldsymbol{\theta}}
\frac{\partial {\lambda}_t(\boldsymbol{\theta})}{\partial\boldsymbol{\theta}^\prime}.
\end{align*}
Thus, as in the proof of Lemma \ref{lemma_as_approxLogLik}, we have 
\begin{align*}
\sup_{\boldsymbol{\theta}\in\boldsymbol{\Theta}}
\bigg\| \frac{\partial^2 l_t(\boldsymbol{\theta})}
{\partial\boldsymbol{\theta}\partial\boldsymbol{\theta}^\prime}
\bigg\|
\leq& 
\sup_{\boldsymbol{\theta}\in\boldsymbol{\Theta}}
\bigg\|\frac{y_t}{\lambda_t(\boldsymbol{\theta})} - 1\bigg\|
\sup_{\boldsymbol{\theta}\in\boldsymbol{\Theta}}
\bigg\|\frac{\partial^2 {\lambda}_t(\boldsymbol{\theta})}{\partial \boldsymbol{\theta}\partial\boldsymbol{\theta}^\prime}\bigg\|
+
\sup_{\boldsymbol{\theta}\in\boldsymbol{\Theta}}
\bigg\|\frac{y_t}{\lambda^2_t(\boldsymbol{\theta})}\bigg\|
\sup_{\boldsymbol{\theta}\in\boldsymbol{\Theta}}
\bigg\|\frac{\partial {\lambda}_t(\boldsymbol{\theta})}{\partial \boldsymbol{\theta}}
\bigg\|^2\\
\leq&
c_1 \bigg\|\frac{y_t}{e^\ell} - 1\bigg\| 
+ c_2 \bigg\|\frac{y_t}{e^{2\ell}}\bigg\| 
<\infty,
\end{align*}
since $e^\ell$ is bounded away from zero and $\mathbb{E}[y_t]<\infty$ follows by Proposition \ref{prop_sec_mom}. Hence, the claimed moment bound follows by an application of Lemmas \ref{lemma_deriv_proc} and \ref{lemma_sec_deriv_proc}.
\end{proof}

\end{document}